\documentclass[12pt,english]{article}
\usepackage[T1]{fontenc}
\usepackage[utf8]{inputenc}
\usepackage{color}
\usepackage{babel}
\usepackage{array}
\usepackage{refstyle}
\usepackage{rotfloat}
\usepackage{booktabs}
\usepackage{bbding}
\usepackage{multirow}
\usepackage{amsmath}
\usepackage{amsthm}
\usepackage{amssymb}
\usepackage{graphicx}
\makeatletter

\AtBeginDocument{\providecommand\subsecref[1]{\ref{subsec:#1}}}
\AtBeginDocument{\providecommand\eqref[1]{\ref{eq:#1}}}
\providecommand{\tabularnewline}{\\}
\newcommand{\lyxdot}{.}

\RS@ifundefined{subsecref}
  {\newref{subsec}{name = \RSsectxt}}
  {}
\RS@ifundefined{thmref}
  {\def\RSthmtxt{theorem~}\newref{thm}{name = \RSthmtxt}}
  {}
\RS@ifundefined{lemref}
  {\def\RSlemtxt{lemma~}\newref{lem}{name = \RSlemtxt}}
  {}

\usepackage[natbibapa]{apacite}
\theoremstyle{plain}
\newtheorem{thm}{\protect\theoremname}
\theoremstyle{remark}
\newtheorem{claim}[thm]{\protect\claimname}

\usepackage{threeparttable}
\let \citet \citep

\usepackage{authblk}

\date{}

\author[1]{Yuval Harel}
\author[1]{Ron Meir}
\author[2]{Manfred Opper}
\affil[1]{Department of Electrical Engineering, Technion -- Israel Institute of Technology, Haifa, Israel}
\affil[2]{Department of Electrical Engineering and Computer Science, Technical University Berlin, Berlin 10587, Germany}

\@ifundefined{showcaptionsetup}{}{%
 \PassOptionsToPackage{caption=false}{subfig}}
\usepackage{subfig}
\makeatother

\providecommand{\claimname}{Claim}
\providecommand{\theoremname}{Theorem}

\begin{document}
\global\long\def\P{\mathbf{P}}%
\global\long\def\E{\mathbf{E}}%
\global\long\def\transpose{^{\intercal}}%
\global\long\def\mc#1{\mathcal{#1}}%

\title{Optimal decoding of dynamic stimuli encoded by heterogeneous populations
of spiking neurons -- a closed form approximation}

\maketitle
\vspace{-20bp}

\begin{center}
Published in Neural Computation, August 2018, Vol. 30, No. 8
\par\end{center}

\medskip{}

\begin{abstract}
Neural decoding may be formulated as dynamic state estimation (filtering)
based on point process observations, a generally intractable problem.
Numerical sampling techniques are often practically useful for the
decoding of real neural data. However, they are less useful as theoretical
tools for modeling and understanding sensory neural systems, since
they lead to limited conceptual insight about optimal encoding and
decoding strategies. We consider sensory neural populations characterized
by a distribution over neuron parameters. We develop an analytically
tractable Bayesian approximation to optimal filtering based on the
observation of spiking activity, that greatly facilitates the analysis
of optimal encoding in situations deviating from common assumptions
of uniform coding. Continuous distributions are used to approximate
large populations with few parameters, resulting in a filter whose
complexity does not grow with the population size, and allowing optimization
of population parameters rather than individual tuning functions.
Numerical comparison with particle filtering demonstrates the quality
of the approximation. The analytic framework leads to insights which
are difficult to obtain from numerical algorithms, and is consistent
with biological observations about the distribution of sensory cells'
preferred stimuli.
\end{abstract}

\section{Introduction}

Populations of sensory neurons encode information about the external
world through their spiking activity. To understand this encoding,
it is natural to model it as an optimal or near-optimal code in the
context of some task performed by higher brain regions, using performance
criteria such as decoding error or motor performance. A Bayesian theory
of neural decoding is useful to characterize optimal encoding, as
the computation of performance criteria typically involves the posterior
distribution of the world state conditioned on spiking activity. 

We model the external world state as a random process, observed through
a set of sensory neuron-like elements characterized by multi-dimensional
tuning functions, representing the elements' average firing rate (see
Figure \ref{setting}). The actual firing of each cell is random and
is given by a Point Process (PP) with rate determined by the external
state and by the cell's tuning function \citet{DayAbb05}. Under this
model, decoding of sensory spike trains may be formulated as a filtering
problem based on PP observations, thus falling within the purview
of nonlinear filtering theory (\citet{SnyMil91}, \citet{Bremaud81}).
Inferring the hidden state under such circumstances has been widely
studied within the Computational Neuroscience literature \citet{DayAbb05,Macke2015}.
Beyond neuroscience, PP-based filtering has been used for position
sensing and tracking in optical communication \citet[sec. 4]{Snyder1977},
control of computer communication networks \citet{Segall1978}, queuing
\citet{Bremaud81} and econometrics \citet{Frey2001}. 

A significant amount of work has been devoted in recent years to the
development of algorithms for fast approximation of the posterior
distribution, leading to an extensive literature (see \citet{Macke2015}
and refs within for a recent review). Much of this work is devoted
to the development of effective sampling techniques, leading to highly
performing finite-dimensional filters that can be applied profitably
to real neural data. These approaches are usually formulated in discrete
time, as befits implementation on digital computers, and lead to complex
mathematical expressions for the posterior distributions, which are
difficult to interpret qualitatively. In this work we are less concerned
with algorithmic issues, and more with establishing closed-form analytic
expressions for approximately optimal continuous time filters, and
using these to characterize the nature of near-optimal encoders, namely
determining the structure and distribution of tuning functions for
optimal state inference. A significant advantage of the closed form
expressions over purely numerical techniques is the insight and intuition
that is gained from them about \textit{qualitative} aspects of the
system. Moreover, the leverage gained by the analytic computation
contributes to reducing the variance inherent to Monte Carlo approaches.
Thus, in this work we do not compare our results to algorithmically
oriented discrete-time filters, but rather to other continuous-time
analytically expressible filters for dynamically varying signals,
with the aim of gaining insight about optimal decoding and encoding
within an analytic framework. 

The problem of filtering a continuous-time diffusion process through
PP observations is solved formally under general conditions in \citet{Snyder1972}
(see also \citet{Segall1976} and \citet{Solo2000}), where a stochastic
PDE for the infinite-dimensional posterior state distribution is derived.
However, this PDE is intractable in general, and not easily amenable
to qualitative or even numerical analysis. Several previous works
have derived exact or approximate finite-dimensional filters, under
various simplifying assumptions. In many of these works (e.g. \citet{RhoSny1977,Komaee2010,YaeMei10,SusMeiOpp13,Twum-Danso2001}),
the tuning functions are chosen so that the total firing rate ---
i.e., the sum of firing rates of all neurons --- is independent of
the state, an assumption we refer to as \emph{uniform coding} (see
Figure \ref{uniform}).  In \citet{RhoSny1977}, an exact finite-dimensional
filter is derived for the case of linear dynamics with Gaussian noise
observed through uniform coding with Gaussian tuning functions\footnote{Although we describe this work using neuroscience terminology, the
motivation and formulation in \citet{RhoSny1977} is not related to
neuroscience.}. The more general setting of uniform coding with arbitrary tuning
functions is considered in \citet{Komaee2010}, where an approximate
filter is obtained.

Other works derive the posterior distribution for non-Markovian state
dynamics, modeled as a Gaussian processes. In \citet{Huys2007}, the
posterior is derived exactly, but its computation is not recursive,
requiring memory of the entire spike history. A recursive version
for Gaussian processes with a Matérn kernel auto-correlation is derived
in \citet{SusMeiOpp11}. Both these works assume uniform coding with
Gaussian tuning functions.

For reasons of mathematical tractability, few previous analytically
oriented works studied neural decoding without the uniform coding
assumption, in spite of the experimental importance and relevance
of non-uniform coding. We discuss such works in comparison to the
present work in section \ref{sec:PreviousWorks} 

The problem of optimal encoding by neural populations has been studied
mostly in the static case. A natural optimality criterion is the estimation
Mean Square Error (MSE). Some works (e.g. \citet{HarMcAlp04}, \citet{GanSim14},
and many others) optimize Fisher information, which serves as a proxy
to the MSE of unbiased estimators through the Cramér-Rao bound \citet{Rao1945}
or, in the Bayesian setting, the Van Trees inequality \citet{GillLevit95}.
Fisher information of neural spiking activity is easy to compute analytically,
at least in the static case \citet[section 3.3]{DayAbb05}, and it
can be used without solving the decoding problem. This approach has
been used to study non-uniform coding of static stimuli by heterogeneous
populations in many works, including \citet{CheDra08,EckBerTol11,GanSim14}.
However, optimizing Fisher information may yield misleading qualitative
results regarding the MSE-optimal encoding \citet{Bethge2002,YaeMei10,Pilarski2015}.
Although, under appropriate conditions, the inverse of Fisher information
approaches the minimum attainable MSE in the limit of infinite decoding
time, it may be a poor proxy for the MSE for finite decoding times,
which are of particular importance in natural settings and in control
problems. Exact computation of the estimation MSE is possible in some
restricted settings: some works along those lines are discussed in
Section \ref{subsec:PreviousWork-Encoding}. 

A possible alternative is the computation of estimation MSE for a
given filter through Monte Carlo simulations. This approach is complicated
by high variability between trials, which means many trials are necessary
for each value of the parameters. Consequently, optimization becomes
very time consuming, possibly impractical when using numerical filters
such as particle filtering, or large neural populations with many
parameters.

In this work, we derive an approximate online filter for the neural
decoding problem in continuous time, and demonstrate its use in investigating
optimal neural encoding. We consider neural populations characterized
by a distribution over neuron parameters. Continuous distributions
are used to approximate large populations with few parameters, resulting
in a filter whose complexity does not grow with the population size,
and allowing optimization of population parameters rather than individual
tuning functions. We suggest further reducing computational complexity
for the encoding problem by using the estimated posterior variance
as an approximation to estimation MSE, as discussed in appendix \ref{sec:Variance-as-proxy}. 

Technically, given the intractable infinite-dimensional nature of
the posterior distribution, we use a projection method replacing the
full posterior at each point in time by a projection onto a simple
family of distributions (Gaussian in our case). This approach, originally
developed in the Filtering literature \citet{Maybeck79,BriHanLeg99},
and termed Assumed Density Filtering (ADF), has been successfully
used more recently in Machine Learning \citet{Opper98,Minka01}. \textcolor{black}{We
derive approximate filters for Gaussian tuning functions, and for
several distributions over tuning function centers, including the
case of a finite population. These filters may be combined to obtain
a filter for heterogeneous mixtures of homogeneous sub-populations.}
We are not aware of any previous work providing an effective closed-form
filter for heterogeneous populations of sensory neurons characterized
by a small number of parameters.

\textbf{Main contributions}: \emph{(i)} Derivation of closed-form
recursive expressions for the continuous-time posterior mean and variance
within the ADF approximation, in the context of large \emph{non-uniform
populations }characterized by a small number of parameters. \emph{(ii)}
Demonstrating the quality of the ADF approximation by comparison to
state-of-the-art particle filtering methods. \emph{(iii)} Characterization
of optimal adaptation (encoding) for sensory cells in a more general
setting than hitherto considered (non-uniform coding, dynamic signals).
\emph{(iv)} Demonstrating the interesting interplay between prior
information and neuronal firing, showing how in certain situations,
the absence of spikes can be highly informative (this phenomenon is
absent under uniform coding). 

Preliminary results discussed in this paper were presented at a conference
\citet{Harel2015}. These included only the special case of Gaussian
distribution of preferred stimuli. The present paper provides a more
general and rigorous formulation of the mathematical framework. By
separately considering different terms in the approximate filter,
we find that updates at spike time depend only on the tuning function
of the spiking neuron, and apply generally to any population distribution.
For the updates between spikes, we provide closed-form expressions
for cases that were not discussed in \citet{Harel2015} -- namely,
uniform populations on an interval and finite heterogeneous mixtures
-- and non-closed-form expressions for the general case, given as
integrals involving the distribution of neuron parameters. We further
supplement our previously published results with numerical evaluation
of the filter's accuracy, an additional example application, and a
detailed comparison with previous works.

\section{Problem Overview\label{sec:Problem-Overview}}

\begin{figure}
\bigskip{}

\includegraphics[width=1\columnwidth]{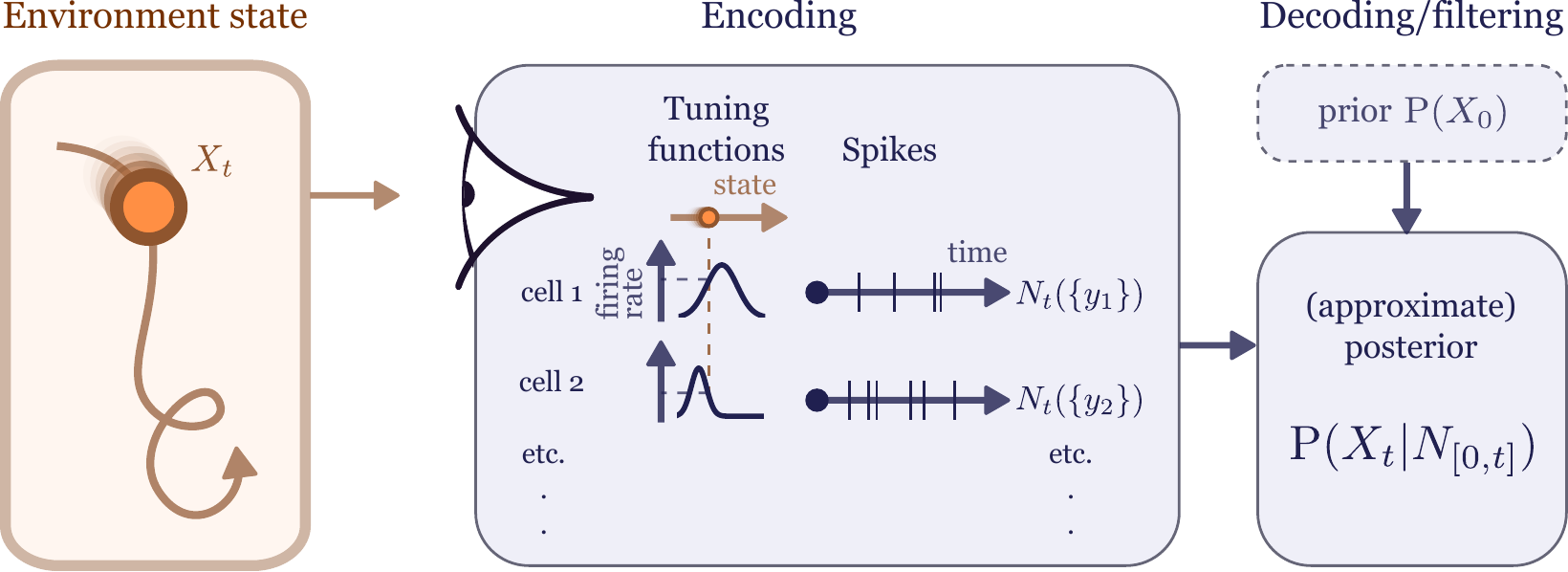}

\bigskip{}
\caption{Problem setting}
\label{setting}
\end{figure}
Consider a dynamical system with state $X_{t}\in\mathbb{R}^{n}$,
observed through the firing patterns of $M$ sensory neurons, as illustrated
in Figure \ref{setting}. Each neuron fires stochastically and independently,
with the $i$th neuron having firing rate $\lambda^{i}\left(X_{t}\right)$.
More detailed assumptions about the dynamics of the state and observation
processes are described in later sections. In this context, we are
interested in the question of optimal encoding and decoding. By \emph{decoding
}we mean computing (exactly or approximately) the full posterior distribution
of $X_{t}$ given $\mathcal{N}_{t}$, which is the history of neural
spikes up to time $t$. The problem of optimal encoding is then the
problem of optimal sensory cell configuration, i.e., finding the optimal
rate function $\left\{ \lambda^{i}\left(\cdot\right)\right\} _{i=1}^{M}$
so as to minimize some performance criterion. We assume the set $\left\{ \lambda^{i}\right\} _{i=1}^{M}$
belong to some parameterized family with parameter $\phi$.

To quantify the performance of the encoding-decoding system, we summarize
the result of decoding using a single estimator $\hat{X}_{t}=\hat{X}_{t}\left(\mathcal{N}_{t}\right)$,
and define the Mean Square Error (MSE) as $\epsilon_{t}\triangleq\mathrm{trace}[(X_{t}-\hat{X}_{t})(X_{t}-\hat{X}_{t})^{T}]$.
We seek $\hat{X}_{t}$ and $\phi$ that solve $\min_{\phi}\min_{\hat{X}_{t}}\mathrm{E}\left[\epsilon_{t}\right]=\min_{\phi}\mathrm{E}[\min_{\hat{X}_{t}}\mathrm{E}[\epsilon_{t}|\mathcal{N}_{t}]]$.
The inner minimization problem in this equation is solved by the MSE-optimal
decoder, which is the posterior mean $\hat{X}_{t}=\mu_{t}\triangleq\mathrm{E}\left[X_{t}|\mathcal{N}_{t}\right]$.
The posterior mean may be computed from the full posterior obtained
by decoding. The outer minimization problem is solved by the optimal
encoder. If decoding is exact, the problem of optimal encoding becomes
that of minimizing the expected posterior variance. Note that, although
we assume a fixed parameter $\phi$ which does not depend on time,
the optimal value of $\phi$ for which the minimum is obtained generally
depends on the time $t$ where the error is to be minimized. In principle,
the encoding/decoding problem can be solved for any value of $t$.
In order to assess performance it is convenient to consider the steady-state
limit $t\to\infty$ for the encoding problem.

Below, we approximately solve the decoding problem for any $t$. We
then explore the problem of choosing the steady-state optimal encoding
parameters $\phi$ using Monte Carlo simulations in an example motivated
by experimental results.

Having an efficient (closed-form) approximate filter allows performing
the Monte Carlo simulation at a significantly reduced computational
cost, relative to numerical methods such as particle filtering. The
computational cost is further reduced by averaging the computed posterior
variance across trials, rather than the squared error, thereby requiring
fewer trials. The mean of the posterior variance equals the MSE (of
the posterior mean), but has the advantage of being less noisy than
the squared error itself -- since by definition it is the mean of
the square error under conditioning on $\mathcal{N}_{t}$ .

\begin{figure}
\subfloat[uniform coding]{\includegraphics[width=0.5\columnwidth]{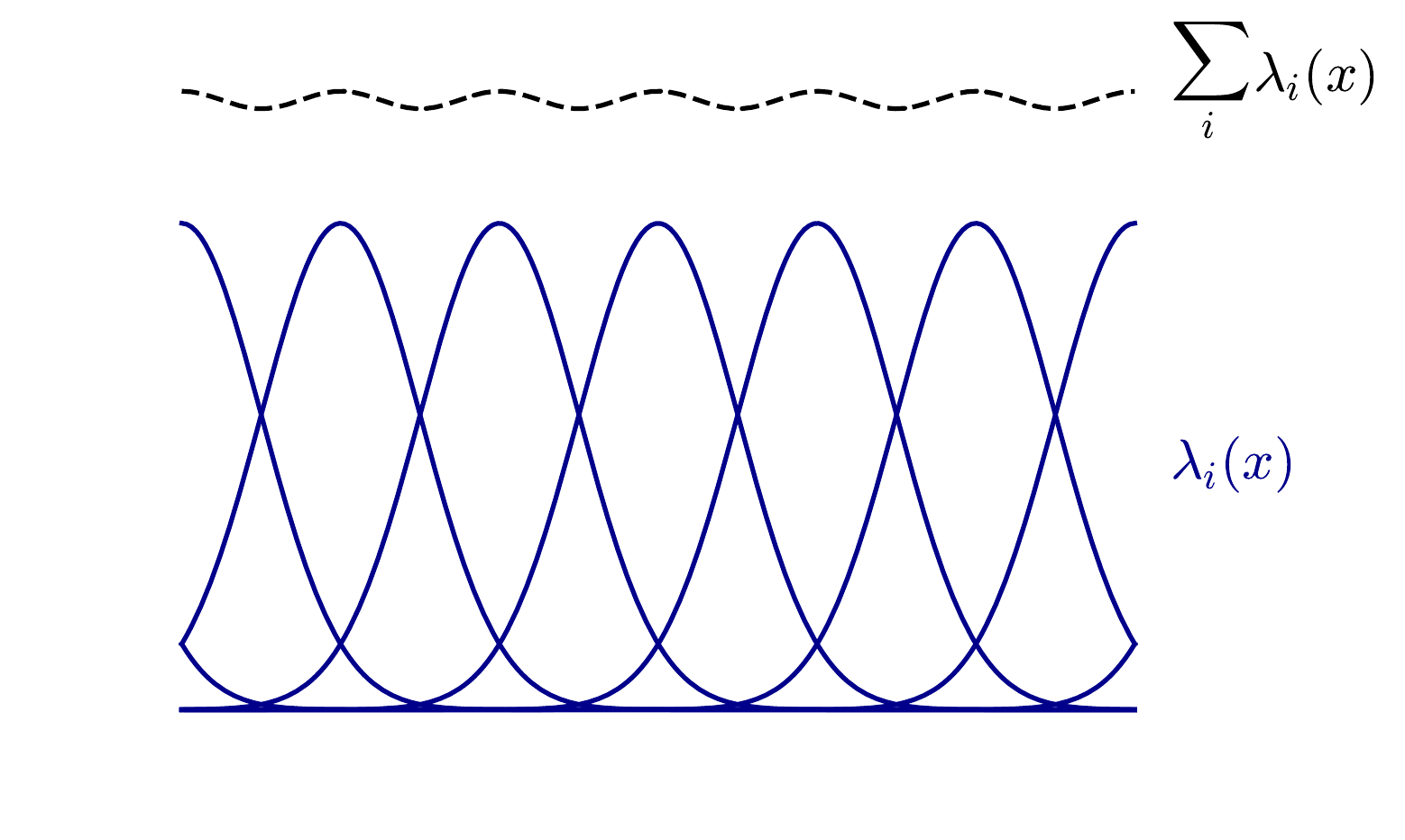}}\subfloat[non-uniform coding]{\includegraphics[width=0.5\columnwidth]{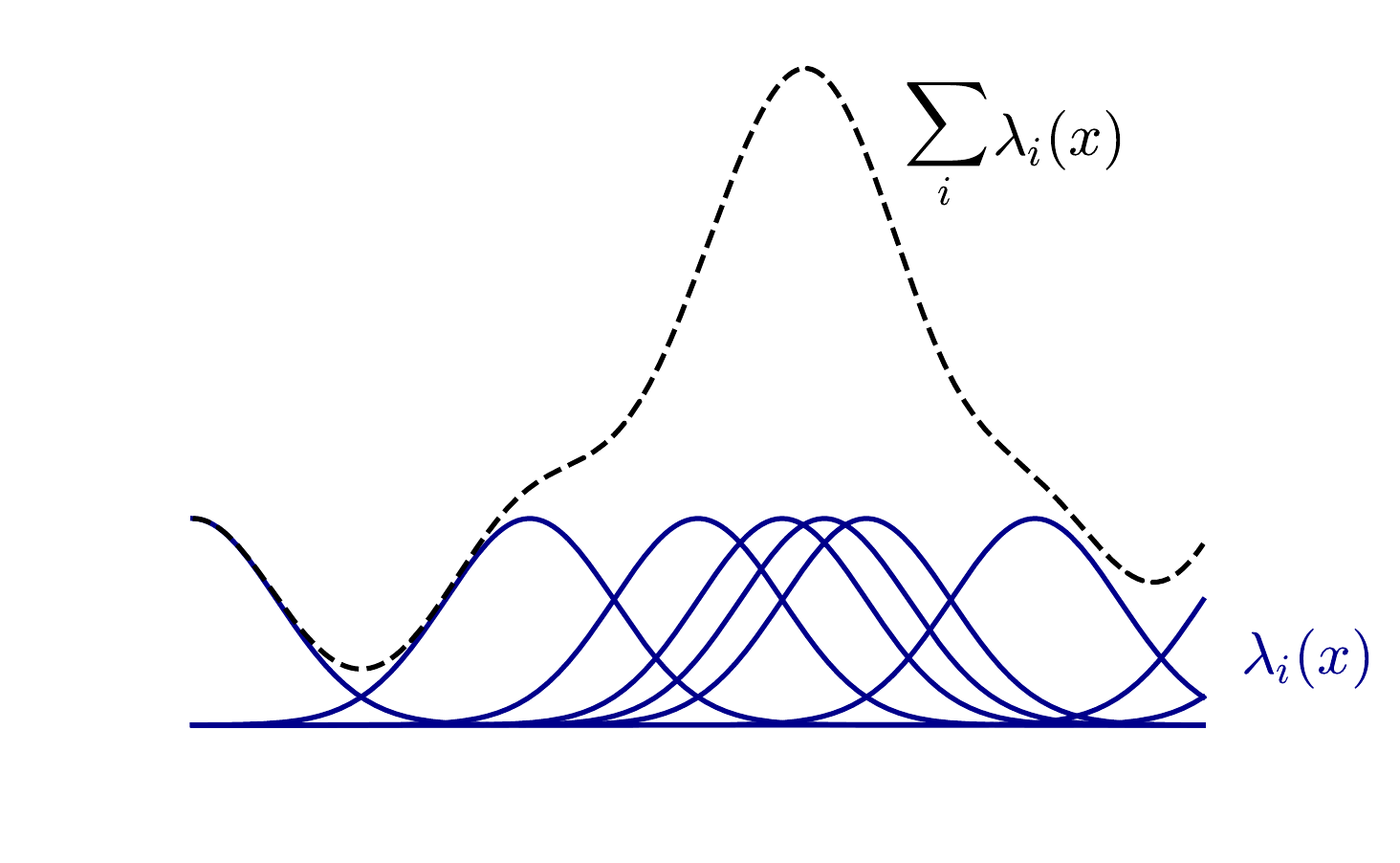}}\caption{The uniform coding property, $\Sigma_{i}\lambda^{i}\left(x\right)=\mathrm{const}$,
holds approximately in homogeneous populations with tuning function
centers located on a dense uniform grid, as demonstrated in (a) in
a 1-dimensional case. (b) illustrates non-uniform coding.}
\label{uniform}
\end{figure}

\section{Decoding}

\subsection{A Finite Population of Gaussian Neurons\label{subsec:FinitePopulation}}

For ease of exposition, we first formulate the problem for a finite
population of neurons. We address a more general setting in subsequent
sections.

\subsubsection{State and observation model\label{subsec:model-finite}}

The observed process $X$ is a diffusion process obeying the Stochastic
Differential Equation (SDE)\footnote{For an introduction to SDEs and the Wiener process see e.g. \citet{Oksendal2003}.
Intuitively, equation (\ref{eq:dynamics}) may be interpreted as a
differential equation with ``continuous-time Gaussian white noise''
$\xi_{t}$:
\[
\dot{X}_{t}=A\left(X_{t}\right)+D\left(X_{t}\right)\xi_{t},
\]
or as the limit as $\Delta t\to0$ of the discretized dynamics
\[
X_{\left(k+1\right)\Delta t}=X_{k\Delta t}+A\left(X_{k\Delta t}\right)\Delta t+D\left(X_{k\Delta t}\right)\xi_{k}\sqrt{\Delta t},
\]
where $\xi_{k}$ are independent standard Gaussian variables. }
\begin{equation}
dX_{t}=A\left(X_{t}\right)dt+D\left(X_{t}\right)dW_{t},\quad\left(t\geq0\right),\label{eq:dynamics}
\end{equation}
where $A\left(\cdot\right),D\left(\cdot\right)$ are arbitrary functions
such that (\ref{eq:dynamics}) has a unique\footnote{in the sense that any two solutions $X_{t}^{\left(1\right)},X_{t}^{\left(2\right)}$
defined over $\left[0,T\right]$ with $X_{0}^{\left(1\right)}=X_{0}^{\left(2\right)}$
satisfy $\P[X_{t}^{\left(1\right)}=X_{t}^{\left(2\right)}\;\forall t\in\left[0,T\right]]=1$.
See e.g. \citet[Theroem 5.2.1]{Oksendal2003} for sufficient conditions.} solution, and $W_{t}$ is a standard Wiener process whose increments
are independent of the history of all other random processes. The
integral with respect to $dW_{t}$ is to be interpreted in the Ito
sense. The initial condition $X_{0}$ is assumed to have a continuous
distribution with a known density.

The observation processes are $\{N^{i}\}_{i=1}^{M}$, where $N_{t}^{i}$
is the spike count of the $i$th neuron up to time $t$. Denote by
$\mathcal{N}_{t}=(N_{\left[0,t\right]}^{i})_{i=1}^{M}$ the history
of neural spikes up to time $t$, and by $N_{t}=\sum_{i=1}^{M}N_{t}^{i}$
the total number of spikes up to time $t$ from all neurons. We assume
that the $i$th neuron fires with rate $\lambda^{i}\left(X_{t}\right)$
at time $t$, independently of other neurons given the state history
$X_{\left[0,t\right]}$. More explicitly, this means 

\begin{align}
\P\left(N_{t+h}^{i}-N_{t}^{i}=k\Big|X_{\left[0,t\right]},\mathcal{N}_{t}\right)= & \begin{cases}
\lambda^{i}\left(X_{t}\right)h+o\left(h\right) & k=0\\
1-\lambda^{i}\left(X_{t}\right)h+o\left(h\right) & k=1\\
o\left(h\right) & k>1
\end{cases}\nonumber \\
 & \left(i=1,\ldots M,\;h\to0^{+}\right)\label{eq:rate-little-oh}
\end{align}
where $X_{\left[0,t\right]}$ denote the history up to time $t$ of
$X$, and $o\left(h\right)$ is little-o asymptotic notation, denoting
any function satisfying $o\left(h\right)/h\to0$ as $h\to0^{+}$.
Thus, each $N^{i}$ is a Doubly-Stochastic Poisson Process (DSPP,
see e.g., \citet{SnyMil91})\footnote{If (\ref{eq:dynamics}) includes an additional feedback (control)
term, $N^{i}$ are not DSPPs. Our results apply with little modification
to this case, as described in appendix \ref{sec:Derivation}.\label{fn:dspp}} with rate process $\lambda^{i}\left(X_{t}\right)$.

To achieve mathematical tractability, we assume that the tuning functions
$\lambda^{i}$ are Gaussian: the firing rate of the $i$th neuron
in response to state $x$ is given by 
\begin{equation}
\lambda^{i}\left(x\right)=h_{i}\exp\left(-\frac{1}{2}\left\Vert H_{i}x-\theta_{i}\right\Vert _{R_{i}}^{2}\right),\label{eq:gauss-tc}
\end{equation}
where $\theta_{i}\in\mathbb{R}^{n}$ is the neuron's preferred location,
$h_{i}\in\mathbb{R}_{+}$ is the neuron's maximal expected firing
rate, $H_{i}\in\mathbb{R}^{m\times n}$ and $R_{i}\in\mathbb{R}^{m\times m}$,
$m\le n$, are fixed matrices, each $R_{i}$ is positive-definite,
and the notation $\left\Vert y\right\Vert _{M}^{2}$ denotes $y^{T}My$.
The inclusion of the matrix $H_{i}$ allows using high-dimensional
models where only some dimensions are observed, for example when the
full state includes velocities but only locations are directly observable.
In typical applications, $H_{i}$ would be the same across all neurons,
or at least across all neurons of the same sensory modality.

In the sequel, we use the following standard notation,
\begin{equation}
\int_{a}^{b}h\left(t\right)dN_{t}^{i}\triangleq\sum_{j}\boldsymbol{1}\left\{ t_{j}^{i}\in\left[a,b\right]\right\} h\left(t_{j}^{i}\right),\label{eq:point-integral}
\end{equation}
for any function $h$, where $t_{j}^{i}$ is the time of the $j$th
point of the process $N^{i}$. This is the usual Lebesgue integral
of $h$ with respect to $N^{i}$ viewed as a discrete measure.

\subsubsection{Model limitations}

The model outlined above involves several simplifications to achieve
tractability. Namely, tuning functions are assumed to be Gaussian,
and firing rates are assumed to be independent of state history and
spike history given the current state (\ref{eq:rate-little-oh}),
yielding a DSPP model (see footnote \ref{fn:dspp} above).

Gaussian tuning functions are a reasonable model for some neural systems,
but are inadequate for others -- e.g., where the tuning is sigmoidal
or where there is a baseline firing rate regardless of stimulus value.
For simplicity, we focus on the Gaussian case in this work. It is
straightforward to extend the derivation presented here to piecewise
linear tuning -- which may be used to represent sigmoidal tuning
functions -- but the resulting expression are more cumbersome. We
have also developed closed-form results for tuning functions given
by sums of Gaussians; however, these require further approximations
in order to obtain analytic results, and are not discussed in this
work.

The assumption of history-independent rates may also limit the model's
applicability. Real sensory neurons exhibit firing-history dependence
in the form of refractory periods and rate adaptation \citet{DayAbb05},
state-history dependence such as input integration \citet{DayAbb05},
or correlations between the firing of different neurons conditioned
on the state \citet{Pillow2008}. These phenomena are captured by
some encoding models, such as simple integrate-and-fire models as
well as more complex physiological models like the Hodgkin-Huxley
model. However, the simplifying independence assumptions above are
common to all works presenting closed-form continuous-time filters
for point process observations that we are aware of. 

Note that characterization of the point processes in terms of their
history-conditioned firing rate, as opposed to finite-dimensional
distributions, does not in itself restrict the model's generality
in any substantial way (see \citet[Theorem 1]{Segall1975-modelling}).
Rather, the independence assumptions are expressed rigorously by the
fact that the right-hand side of (\ref{eq:rate-little-oh}) depends
neither on previous values of $X$, nor on previous spike times of
any neuron. Some of our analysis applies without modification when
rates are allowed to depend on \emph{spiking} history (specifically,
equations (\ref{eq:finite}) below), so it may be possible to extend
these techniques to some history-dependent models. However, when rates
may depend on the \emph{state} history, exact filtering may involve
the posterior distribution of the entire state history rather than
the current state, so that a different approach is probably required.

\subsubsection{Exact filtering equations\label{subsec:exact}}

Let $p_{t}^{\,\mathcal{N}}\left(\cdot\right)$ be the posterior density
of $X_{t}$ given the firing history $\mathcal{N}_{t}$, and $\E_{t}^{\mathcal{N}}\left[\cdot\right]$
the posterior expectation given $\mathcal{N}_{t}$. The prior density
$p_{0}^{\,\mathcal{N}}$ is assumed to be known. We denote by $\hat{\lambda}_{t}^{i}$
the rate of $N_{t}^{i}$ with respect to the history of spiking only
-- i.e., the rates that would appear in the right-hand side of (\ref{eq:rate-little-oh})
if the conditioning on the left were only on $\mc N_{t}$. These rates
are given by\footnote{See \citet[Theorem 2]{Segall1975-modelling}}
\[
\hat{\lambda}_{t}^{i}=\E_{t}^{\mathcal{N}}\left[\lambda^{i}\left(X_{t}\right)\right]=\int p_{t}^{\,\mathcal{N}}\left(x\right)\lambda^{i}\left(x\right)dx.
\]

The problem of filtering a diffusion process $X$ from a doubly stochastic
Poisson process driven by $X$ is formally solved in \citet{Snyder1972},
where the authors derive a stochastic PDE for the posterior density\footnote{the setting of \citet{Snyder1972} includes a single observation point
process. The extension to several point processes is obtained through
summation as in (\ref{eq:snyder-density}), and is a special case
of a more general PDE described in \citet{RhoSny1977} and discussed
in Appendix \ref{sec:Derivation}.},
\begin{equation}
dp_{t}^{\,\mathcal{N}}\left(x\right)=\left\{ \mathcal{L}^{*}p_{t}^{\,\mathcal{N}}\right\} \left(x\right)dt+p_{t}^{\,\mathcal{N}}\left(x\right)\sum_{i}\left(\frac{\lambda^{i}\left(x\right)}{\hat{\lambda}_{t}^{i}}-1\right)\left(dN_{t}^{i}-\hat{\lambda}_{t}^{i}dt\right),\label{eq:snyder-density}
\end{equation}
where $\mathcal{L}$ is the state's infinitesimal generator (Kolmogorov's
backward operator), defined as $\mathcal{L}h\left(x\right)=\lim_{\Delta t\to0^{+}}\left(\mathrm{E}\left[h\left(X_{t+\Delta t}\right)|X_{t}=x\right]-h\left(x\right)\right)/\Delta t$,
$\mathcal{L}^{*}$ is $\mathcal{L}$'s adjoint operator (Kolmogorov's
forward operator). The notation $dN_{t}^{i}$ is interpreted as in
(\ref{eq:point-integral}), so this term contributes a jump of size
$p_{t}^{\,\mathcal{N}}\left(x\right)\left(\lambda^{i}\left(x\right)/\hat{\lambda}_{t}^{i}-1\right)$
at a spike of the $i$th neuron. Equation (\ref{eq:snyder-density})
may be written in a notation more familiar for non-stochastic PDEs
using Dirac delta functions,
\[
\frac{\partial}{\partial t}p_{t}^{\,\mathcal{N}}\left(x\right)=\left\{ \mathcal{L}^{*}p_{t}^{\,\mathcal{N}}\right\} \left(x\right)+p_{t}^{\,\mathcal{N}}\left(x\right)\sum_{i}\left(\frac{\lambda^{i}\left(x\right)}{\hat{\lambda}_{t}^{i}}-1\right)\left(\dot{N}_{t}^{i}-\hat{\lambda}_{t}^{i}\right),
\]
where $\dot{N}_{t}^{i}\triangleq\sum_{j}\delta\left(t-t_{j}^{i}\right)$
is the spike train of the $i$th neuron, which is the formal derivative
of the process $N^{i}$. An accessible, albeit non-rigorous, derivation
of (\ref{eq:snyder-density}) via time discretization is found in
\citet[section 2.3]{SusemihlThesis2014}.

The stochastic PDE (\ref{eq:snyder-density}) is non-linear and non-local
(due to the dependence of $\hat{\lambda}_{t}^{i}$ on $p_{t}^{\,\mathcal{N}}$),
and therefore usually intractable. In \citet{RhoSny1977,Susemihl2014}
the authors consider linear dynamics with a Gaussian prior and Gaussian
sensors with centers distributed uniformly over the state space. In
this case, \emph{the posterior is Gaussian}, and (\ref{eq:snyder-density})
leads to closed-form ODEs for its mean and variance. In our more general
setting, we can obtain exact equations for the posterior mean and
variance, as follows.

Let $\mu_{t}\triangleq\E_{t}^{\mathcal{N}}X_{t},\tilde{X}_{t}\triangleq X_{t}-\mu_{t},\Sigma_{t}\triangleq\E_{t}^{\mathcal{N}}[\tilde{X}_{t}\tilde{X}_{t}^{T}]$.
Using (\ref{eq:snyder-density}), along with known results about the
form of the infinitesimal generator $\mathcal{L}_{t}$ for diffusion
processes (e.g. \citet{Oksendal2003}, Theorem 7.3.3), the first two
posterior moments can be shown to obey the following exact equations
(see Appendix \ref{sec:Derivation}):\begin{subequations}\label{eq:finite}

\begin{align}
d\mu_{t} & =\E_{t}^{\mathcal{N}}\left[A\left(X_{t}\right)\right]dt+\sum_{i}\E_{t^{-}}^{\mathcal{N}}\left[\omega_{t^{-}}^{i}X_{t^{-}}\right]\left(dN_{t}^{i}-\hat{\lambda}_{t}^{i}dt\right)\label{eq:mean-finite}\\
d\Sigma_{t} & =\E_{t}^{\mathcal{N}}\left[A\left(X_{t\left(y\right)}\right)\tilde{X}_{t}\transpose+\tilde{X}_{t}A\left(X_{t}\right)\transpose+D\left(X_{t}\right)D\left(X_{t}\right)\transpose\right]dt\nonumber \\
 & \quad+\sum_{i}\E_{t^{-}}^{\mathcal{N}}\left[\omega_{t^{-}}^{i}\tilde{X}_{t^{-}}\tilde{X}_{t^{-}}\transpose\right]\left(dN_{t}^{i}-\hat{\lambda}_{t}^{i}dt\right)\nonumber \\
 & \quad-\sum_{i}\E_{t^{-}}^{\mathcal{N}}\left[\omega_{t^{-}}^{i}X_{t^{-}}\right]\E_{t^{-}}^{\mathcal{N}}\left[\omega_{t^{-}}^{i}X_{t^{-}}\transpose\right]dN_{t}^{i}\label{eq:var-finite}
\end{align}
\end{subequations}where
\[
\omega_{t}^{i}\triangleq\frac{\lambda^{i}\left(X_{t}\right)}{\hat{\lambda}_{t}^{i}}-1,
\]
and the expressions involving $t^{-}$ denote left limits, which are
necessary since the solutions to (\ref{eq:finite}) are discontinuous
at spike times.

In contrast with the more familiar case of linear dynamics with Gaussian
white noise, and the corresponding Kalman-Bucy filter \citet{Maybeck79},
here the posterior variance is random, and is generally not monotonically
decreasing even when estimating a constant state. However, noting
that $\E[dN_{t}^{i}-\hat{\lambda}_{t}^{i}dt]=0$, we may observe from
(\ref{eq:var-finite}) that for a constant state ($A=D=0$), the \emph{expected}
posterior variance $\E\left[\Sigma_{t}\right]$ is decreasing, since
the first two terms in (\ref{eq:var-finite}) vanish.

We will find it useful to rewrite (\ref{eq:finite}) in a different
form, as follows,
\begin{align*}
d\mu_{t} & =d\mu_{t}^{\pi}+d\mu_{t}^{\mathrm{c}}+d\mu_{t}^{N},\\
d\Sigma_{t} & =d\Sigma_{t}^{\pi}+d\Sigma_{t}^{\mathrm{c}}+d\Sigma_{t}^{N},
\end{align*}
where $d\mu_{t}^{\pi},d\Sigma_{t}^{\pi}$ are the \emph{prior terms},\emph{
}corresponding to $\mathcal{L}^{*}p_{t}^{\,\mathcal{N}}\left(x\right)$
in (\ref{eq:snyder-density}), and the remaining terms are divided
into \emph{continuous update terms} $d\mu_{t}^{\mathrm{c}},d\Sigma_{t}^{\mathrm{c}}$
(multiplying $dt$) and \emph{discontinuous update terms} $d\mu_{t}^{N},d\Sigma_{t}^{N}$
(multiplying $dN_{t}^{i}$). Using (\ref{eq:finite}), we find the
exact equations\begin{subequations}\label{eq:exact-finite}
\begin{align}
d\mu_{t}^{\pi} & =\E_{t}^{\mathcal{N}}\left[A\left(X_{t}\right)\right]dt\label{eq:mean-p}\\
d\Sigma_{t}^{\pi} & =\E_{t}^{\mathcal{N}}\left[A\left(X_{t}\right)\tilde{X}_{t}\transpose+\tilde{X}_{t}A\left(X_{t}\right)\transpose+D\left(X_{t}\right)D\left(X_{t}\right)\transpose\right]dt\label{eq:var-p}\\
d\mu_{t}^{\mathrm{c}} & =-\sum_{i}\E_{t}^{\mathcal{N}}\left[\omega_{t}^{i}X_{t}\right]\hat{\lambda}_{i}dt\label{eq:mean-c-finite}\\
d\Sigma_{t}^{\mathrm{c}} & =-\sum_{i}\E_{t}^{\mathcal{N}}\left[\omega_{t}^{i}\tilde{X}_{t}\tilde{X}_{t}\transpose\right]\hat{\lambda}_{i}dt\\
d\mu_{t}^{N} & =\sum_{i}\E_{t^{-}}^{\mathcal{N}}\left[\omega_{t^{-}}^{i}X_{t^{-}}\right]dN_{t}^{i}\\
d\Sigma_{t}^{N} & =\sum_{i}\bigg(\E_{t^{-}}^{\mathcal{N}}\left[\omega_{t^{-}}^{i}\tilde{X}_{t^{-}}\tilde{X}_{t^{-}}\transpose\right]-\E_{t^{-}}^{\mathcal{N}}\left[\omega_{t^{-}}^{i}X_{t^{-}}\right]\E_{t^{-}}^{\mathcal{N}}\left[\omega_{t^{-}}^{i}X_{t^{-}}\transpose\right]\bigg)dN_{t}^{i}\label{eq:var-N-finite}
\end{align}
\end{subequations}The prior terms $d\mu_{t}^{\pi},d\Sigma_{t}^{\pi}$
represent the known dynamics of $X$, and are the same terms appearing
in the Kalman-Bucy filter. These would be the only terms left if no
measurements were available, and would vanish for a static state.
The continuous update terms $d\mu_{t}^{\mathrm{c}},d\Sigma_{t}^{\mathrm{c}}$
represent updates to the posterior between spikes that are not derived
from $X$'s dynamics, and therefore may be interpreted as corresponding
to information obtained from the \emph{absence of spikes}. The discontinuous
update terms $d\mu_{t}^{N},d\Sigma_{t}^{N}$ contribute a change to
the posterior at spike times, depending on the spike's origin $i$,
and thus represent information obtained from the presence of a spike
as well as the parameters of the spiking neuron.

Note that the Gaussian tuning assumption (\ref{eq:gauss-tc}) has
not been used in this section, and equations (\ref{eq:exact-finite})
are valid for any form of $\lambda_{i}$.

\subsubsection{ADF approximation\label{subsec:ADF}}

While equations (\ref{eq:exact-finite}) are exact, they are not practical,
since they require computation of posterior expectations $\E_{t}^{\mathcal{N}}\left[\cdot\right]$.
To bring them to a closed form, we use ADF with an assumed Gaussian
density (see \citet{Opper98} for details). Informally, this may be
envisioned as integrating (\ref{eq:exact-finite}) while replacing
the distribution $p_{t}^{\,\mathcal{N}}$ by its approximating Gaussian
``at each time step''. The approximating Gaussian is obtained by
matching the first two moments of $p_{t}^{\,\mathcal{N}}$ \citet{Opper98}.
Note that the solution of the resulting equations does not in general
match the first two moments of the exact solution, though it may approximate
it. Practically, the ADF approximation amounts to substituting the
normal distribution $\mathcal{N}(\mu_{t},\Sigma_{t})$ for $p_{t}^{\,\mathcal{N}}$
to compute the expectations in (\ref{eq:exact-finite}). This heuristic
may be justified by its relation to a projection method, where right-hand
side of the density PDE is projected onto the tangent space of the
approximating family of densities: the two approaches are equivalent
when the approximating family is exponential \citet{BriHanLeg99}.

If the dynamics are linear, the prior updates (\ref{eq:mean-p})-(\ref{eq:var-p})
are easily computed in closed form after this substitution. Specifically,
for $dX_{t}=AX_{t}dt+DdW_{t}$, the prior updates read\begin{subequations}\label{eq:p-linear}
\begin{align}
d\mu_{t}^{\pi} & =A\mu_{t}\,dt\label{eq:mean-p-linear}\\
d\Sigma_{t}^{\pi} & =\left(A\Sigma_{t}+\Sigma_{t}A+DD\transpose\right)dt,\label{eq:var-p-linear}
\end{align}
\end{subequations}as in the Kalman-Bucy filter \citet{Maybeck79}.
Otherwise, they may be approximated by expanding the non-linear functions
$A(x)$ and $D\left(x\right)D\left(x\right)\transpose$ as power series
and applying the assumed Gaussian density, resulting in tractable
Gaussian integrals. The use of ADF in the prior terms is outside the
scope of this work; see e.g. \citet[Chapter 12]{Maybeck79}.

We therefore turn to the approximation of the non-prior updates (\ref{eq:mean-c-finite})-(\ref{eq:var-N-finite})
in the case of Gaussian tuning (\ref{eq:gauss-tc}). 

Abusing notation, from here on we use $\mu_{t},\Sigma_{t}$, and $p_{t}^{\,\mathcal{N}}\left(x\right)$
to refer to the ADF approximation rather than to the exact values.
Applying the Gaussian ADF approximation $p_{t}^{\,\mathcal{N}}\left(x\right)\approx\mathcal{N}\left(x;\mu_{t},\Sigma_{t}\right)$
in the case of Gaussian tuning functions (\ref{eq:gauss-tc}) yields
the non-prior terms\begin{subequations}\label{eq:filtering-finite}
\begin{align}
d\mu_{t}^{\mathrm{c}} & =\sum_{i}\Sigma_{t}H_{i}\transpose S_{t}^{i}\delta_{t}^{i}\hat{\lambda}_{t}^{i}\,dt,\label{eq:mean-c-gaussian-finite}\\
d\Sigma_{t}^{\mathrm{c}} & =\sum_{i}\Sigma_{t}H_{i}\transpose\left(S_{t}^{i}-S_{t}^{i}\delta_{t}^{i}\left(\delta_{t}^{i}\right)\transpose S_{t}^{i}\right)H_{i}\Sigma_{t}\hat{\lambda}_{t}^{i}\,dt,\label{var-c-gaussian-finite}\\
d\mu_{t}^{N} & =-\sum_{i}\Sigma_{t^{-}}H_{i}\transpose S_{t^{-}}^{i}\delta_{t^{-}}^{i}\,dN_{t}^{i},\label{eq:mean-N-gaussian-finite}\\
d\Sigma_{t}^{N} & =-\sum_{i}\Sigma_{t^{-}}H_{i}\transpose S_{t^{-}}^{i}H_{i}\Sigma_{t^{-}}\,dN_{t}^{i},\label{eq:var-N-gaussian-finite}\\
\hat{\lambda}_{t}^{i} & =h_{i}\sqrt{\frac{\left|S_{t}^{i}\right|}{\left|R_{i}\right|}}\exp\left(-\frac{1}{2}\left\Vert \delta_{t}^{i}\right\Vert _{S_{t}^{i}}^{2}\right),\label{eq:self-rate-gaussian-finite}
\end{align}
where
\begin{align}
\delta_{t}^{i} & \triangleq H_{i}\mu_{t}-\theta_{i},\nonumber \\
S_{t}^{i} & \triangleq\left(R_{i}^{-1}+H\Sigma_{t}H\transpose\right)^{-1}.\label{eq:S-finite}
\end{align}
\end{subequations}These equations are a special case of (\ref{eq:filtering-gaussian}),
which are derived in appendix \ref{sec:Derivation}. The updates for
the posterior precision $\Sigma_{t}^{-1}$ have a simpler form, also
derived in appendix \ref{sec:Derivation}: \begin{subequations}\label{eq:precision-finite}
\begin{align}
d\Sigma_{t}^{-1,\mathrm{c}} & =-\sum_{i}H_{i}\transpose\left(S_{t}^{i}-S_{t}^{i}\delta_{t}^{i}\left(\delta_{t}^{i}\right)\transpose S_{t}^{i}\right)H_{i}\hat{\lambda}_{t}^{i}dt,\label{eq:prec-c-gaussian-finite}\\
d\Sigma_{t}^{-1,N} & =\sum_{i}H_{i}\transpose R_{i}H_{i}dN_{t}^{i}.\label{eq:prec-N-gaussian-finite}
\end{align}
\end{subequations}

In the scalar case $m=n=1$, with $H=1$, $\sigma_{t}^{2}=\Sigma_{t},\alpha_{i}^{2}=R_{i}^{-1}$,
the update equations (\ref{eq:filtering-finite}), (\ref{eq:precision-finite})
read\begin{subequations}\label{eq:filtering-guassian-1d}
\begin{align}
d\mu_{t}^{\mathrm{c}} & =\sum_{i}\frac{\sigma_{t}^{2}}{\sigma_{t}^{2}+\alpha_{i}^{2}}\left(\mu_{t}-\theta_{i}\right)\hat{\lambda}_{t}^{i}\,dN_{t}^{i},\label{eq:mean-c-gaussian-finite-1d}\\
d\sigma_{t}^{2,\mathrm{c}} & =\sum_{i}\frac{\sigma_{t}^{2}}{\sigma_{t}^{2}+\alpha_{i}^{2}}\left(1-\frac{\left(\mu_{t}-\theta_{i}\right)^{2}}{\sigma_{t}^{2}+\alpha_{i}^{2}}\right)\sigma_{t}^{2}\hat{\lambda}_{t}^{i}\,dt,\\
d\sigma_{t}^{-2,\mathrm{c}} & =\sum_{i}\frac{1}{\sigma_{t}^{2}+\alpha_{i}^{2}}\left(\frac{\left(\mu_{t}-\theta_{i}\right)^{2}}{\sigma_{t}^{2}+\alpha_{i}^{2}}-1\right)\hat{\lambda}_{t}^{i}\,dt,\\
d\mu_{t}^{N} & =\sum_{i}\frac{\sigma_{t^{-}}^{2}}{\sigma_{t^{-}}^{2}+\alpha_{i}^{2}}\left(\theta_{i}-\mu_{t^{-}}\right)dN_{t}^{i},\\
d\sigma_{t}^{2,N} & =-\sum_{i}\frac{\sigma_{t^{-}}^{2}}{\sigma_{t^{-}}^{2}+\alpha_{i}^{2}}\sigma_{t^{-}}^{2}\,dN_{t}^{i},\label{eq:var-N-gaussian-finite-1d}\\
d\sigma_{t}^{-2,N} & =\sum_{i}\alpha_{i}^{-2}\,dN_{t}^{i}\\
\hat{\lambda}_{t}^{i} & =h_{i}\sqrt{2\pi\alpha_{i}^{2}}\mathcal{N}\left(\mu_{t};\theta_{i},\sigma_{t}^{2}+\alpha_{i}^{2}\right).\label{eq:self-rate-gaussian-finite-1d}
\end{align}
\end{subequations}Figure \ref{filtering-example} illustrates the
filter (\ref{eq:filtering-guassian-1d}) in a one-dimensional example.

\begin{figure}
\includegraphics[width=1\columnwidth]{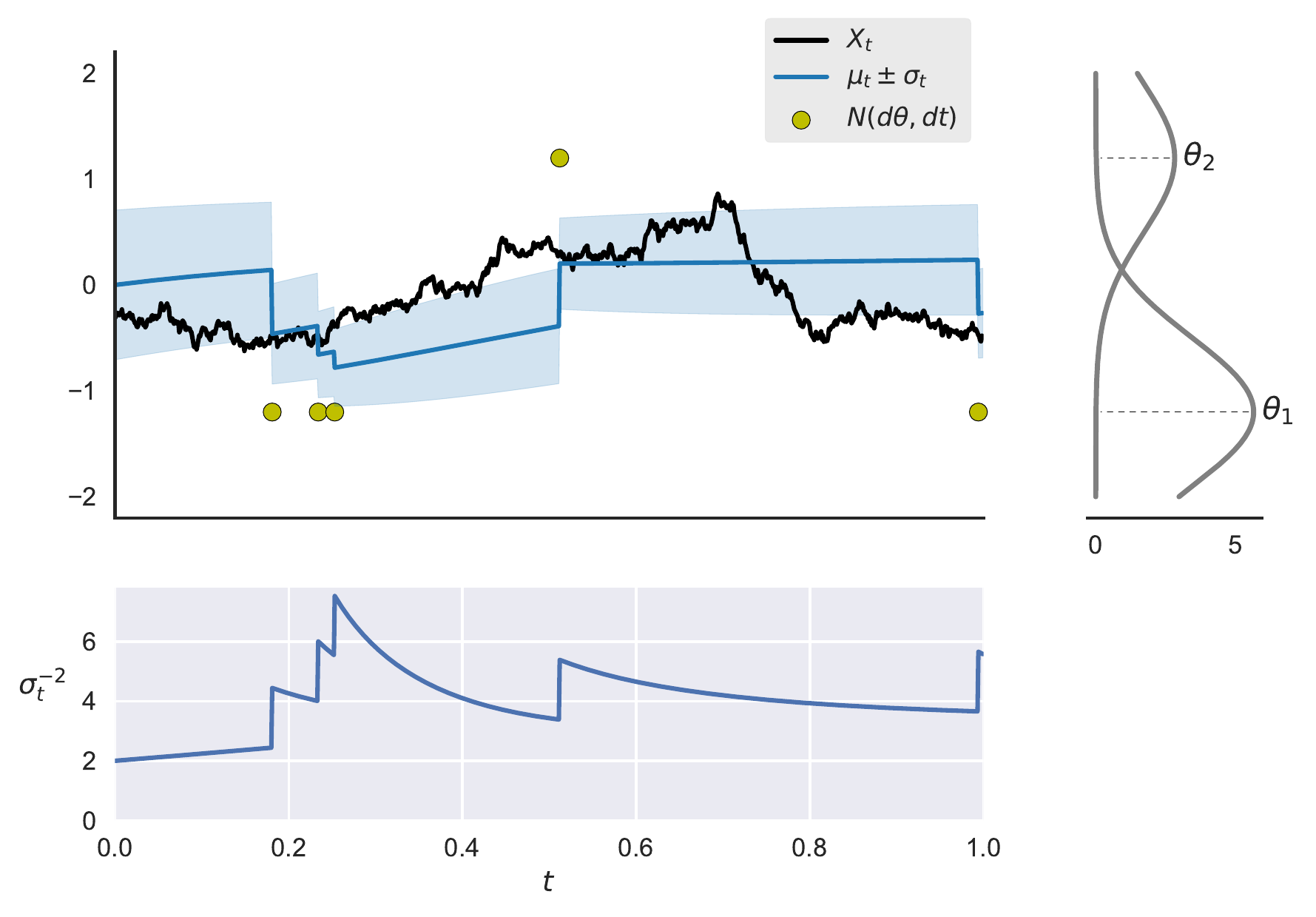}\caption{Filtering example with two sensory neurons. A dynamic state $dX_{t}=-X_{t}dt+dW_{t}$
is observed by two Gaussian neurons (\ref{eq:gauss-tc}), with preferred
stimuli $\theta_{1}=-1.2,\theta_{2}=1.2$, and filtered according
to (\ref{eq:p-linear}),(\ref{eq:filtering-finite}). Each dot correspond
to a spike with the vertical location indicating the neuron's preferred
stimulus $\theta$, equal to $\pm1.2$ in this case. The approximate
posterior mean $\mu_{t}$ is shown in blue, with the posterior standard
deviation $\sigma_{t}$ in shaded areas. The curves to the right of
the graph show the tuning functions of the two neurons. The bottom
graph shows the posterior precision $\sigma_{t}^{-2}$. Neuron parameters
are $\theta_{1}=-1.2,\theta_{2}=1.2,h_{1}=10,h_{2}=5,H_{i}=1,R_{i}^{-1}=\alpha_{i}^{2}=0.5\:\left(i=1,2\right)$.
The process and the filter were both initialized from the steady-state
distribution of the dynamics, which is $\mathcal{N}\left(0,0.5\right)$.
The dynamics were discretized with time step $\Delta t=10^{-3}$. }
\label{filtering-example}
\end{figure}

\subsubsection{Interpretation\label{subsec:Interpretation}}

To gain some insight into the filtering equations, we consider the
discontinuous updates (\ref{eq:mean-N-gaussian-finite})-(\ref{eq:var-N-gaussian-finite})
and continuous updates (\ref{eq:mean-c-gaussian-finite})-(\ref{var-c-gaussian-finite})
in some special cases, in reference to the example presented in Figure
\ref{filtering-example}.

\paragraph{Discontinuous updates}

Consider the case $H_{i}=I$. As seen from the discontinuous update
equations (\ref{eq:mean-N-gaussian-finite})-(\ref{eq:var-N-gaussian-finite}),
when the $i$th neuron spikes, the posterior mean moves towards its
preferred location $\theta_{i}$, and the posterior variance decreases
(in the sense that $\Sigma_{t^{+}}-\Sigma_{t^{-}}$ is negative definite),
as seen in Figure \ref{filtering-example}. Neither update depends
on $h_{i}$. 

For general $H_{i}\in\mathbb{R}^{m\times n}$ of full row rank, let
$H_{i}^{\mathrm{r}}$ be any right inverse of $H_{i}$ and $\bar{\theta}_{i}=H_{i}^{\mathrm{r}}\theta_{i}$.
Note that $H_{i}$ projects the state $X_{t}$ to ``perceptual coordinates''
employed by the $i$th neuron; thus $\bar{\theta}_{i}$ may be interpreted
as the tuning function center in state coordinate, whereas $\theta_{i}$
is in the neuron's perceptual coordinates. We may rewrite (\ref{eq:gauss-tc})
in state coordinates as
\[
\lambda^{i}\left(x\right)=h_{i}\exp\left(-\frac{1}{2}\left\Vert x-\bar{\theta}_{i}\right\Vert _{H_{i}\transpose R_{i}H_{i}}^{2}\right).
\]
Now, the updates for a spike of neuron $i$ at time $t$ can be written
more intuitively (see appendix \ref{sec:Derivation}) as
\begin{align}
\mu_{t^{+}} & =\left(\Sigma_{t^{-}}^{-1}+H_{i}\transpose R_{i}H_{i}\right)^{-1}\left(\Sigma_{t^{-}}^{-1}\mu_{t^{-}}+H_{i}\transpose R_{i}\theta_{i}\right)\label{eq:mean-spike}\\
 & =\left(\Sigma_{t^{-}}^{-1}+H_{i}\transpose R_{i}H_{i}\right)^{-1}\left(\Sigma_{t^{-}}^{-1}\mu_{t^{-}}+H_{i}\transpose R_{i}H_{i}\bar{\theta}_{i}\right)\nonumber \\
\Sigma_{t^{+}}^{-1} & =\Sigma_{t^{-}}^{-1}+H_{i}\transpose R_{i}H_{i}\nonumber 
\end{align}
Thus the new posterior mean is a weighted average of the pre-spike
posterior mean and the preferred stimulus in state coordinates. The
posterior precision $\Sigma_{t}^{-1}$ increases by $H_{i}\transpose R_{i}H_{i}$
which is the tuning function precision matrix in state coordinates.
This may be observed in Figure \ref{filtering-example}, where the
posterior precision increases at each spike time by the fixed amount
$R_{i}=\alpha_{i}^{-2}=2$.

\paragraph{Continuous updates}

The continuous mean update equation (\ref{eq:mean-c-gaussian-finite}),
contributing between spiking events, also admits an intuitive interpretation,
in the case where all neurons share the same shape matrices $H_{i}=H,R_{i}=R$.
In this case, the equation reads
\[
d\mu_{t}^{\mathrm{c}}=\Sigma_{t}H\transpose\left(R^{-1}+H\Sigma_{t}H\transpose\right)^{-1}\left(H\mu_{t}-\sum_{i}\theta_{i}\nu_{t}^{i}\right)\sum_{i}\hat{\lambda}_{t}^{i}dt,
\]
where $\nu_{t}^{i}\triangleq\hat{\lambda}_{t}^{i}/\sum_{j}\hat{\lambda}_{t}^{j}$.
The normalized rates $\nu_{t}^{i}$ may be interpreted heuristically
as the distribution of the next firing neuron's index $i$, provided
the next spike occurs immediately \citet[Section 1, Theorem T15]{Bremaud81}.
Thus, the absence of spikes drives the posterior mean away from the
expected preferred stimulus of the next spiking neuron. The strength
of this effect scales with $\sum_{i}\hat{\lambda}_{t}^{i}$, which
is the total expected rate of spikes given the firing history. This
behavior is qualitatively similar to the result obtained in \citet{BobMeiEld09}
for a finite population of neurons observing a continuous-time finite-state
Markov process, where the posterior probability between spikes concentrates
on states with lower total firing rate.

This behavior may be observed in Figure (\ref{filtering-example}):
when the posterior mean $\mu_{t}$ is near a neuron's preferred stimulus,
it moves away from it between spikes as the next spike is expected
from that neuron. Similarly, despite the symmetry of the two neurons'
preferred stimuli relative to the starting estimate $\mu_{0}$, the
posterior mean shifts at the start of the trial towards the preferred
stimulus of the second neuron, due to its lower firing rate.

The continuous variance update (\ref{eq:mean-c-gaussian-finite})
consists of the difference of two positive semidefinite terms, and
accordingly the posterior variance may increase or decrease between
spikes along various directions. In Figure (\ref{filtering-example}),
the posterior variance decreases before the first spike, and increases
between spikes afterwards.

\subsection{Continuous population approximation}

\subsubsection{Motivation}

The filtering equations (\ref{eq:mean-c-gaussian-finite})-(\ref{eq:S-finite})
implement sensory decoding for non-uniform populations. However, their
applicability to studying neural encoding in large heterogeneous populations
is limited for two closely related reasons. First, the computational
cost of the filter is linear in the number of neurons. Second, the
size of the parameter space describing the population is also linear
in the number of neurons, making optimization of large populations
computationally costly. These traits are shared with other filters
designed for heterogeneous populations, namely \citet{EdenBrown2008,BobMeiEld09,Twum-Danso2001}.

To reduce the parameter space and simplify the filter, we approximate
large neural populations by an infinite continuous population, characterized
by a distribution over neuron parameters. These distributions are
described by few parameters, resulting in a filter whose complexity
does not grow with the population size, and allowing optimization
of population parameters rather than individual tuning functions.

For example, consider a homogeneous population of neurons with preferred
stimuli equally spaced on an interval, as depicted in Figure \ref{continuous-pop}(a).
If the population is large, its firing pattern statistics may be modeled
as an infinite population of neurons, with preferred stimuli uniformly
distributed on the same interval, as in Figure \ref{continuous-pop}(c).
In this continuous population model, each spike is characterized by
the preferred stimulus of the firing neuron -- which is a continuous
variable -- rather than by the neuron's index. Such a population
is parameterized by only two variables representing the endpoints
of the interval, in addition to the tuning function height and width
parameters. There is no need for a parameter representing the density
of neurons on the interval, as scaling the density of neurons is equivalent
to identical scaling of each neuron's maximum firing rate.

\subsubsection{Marked point processes as continuous population models\label{subsec:marked}}

\begin{figure}
\includegraphics[width=1\columnwidth]{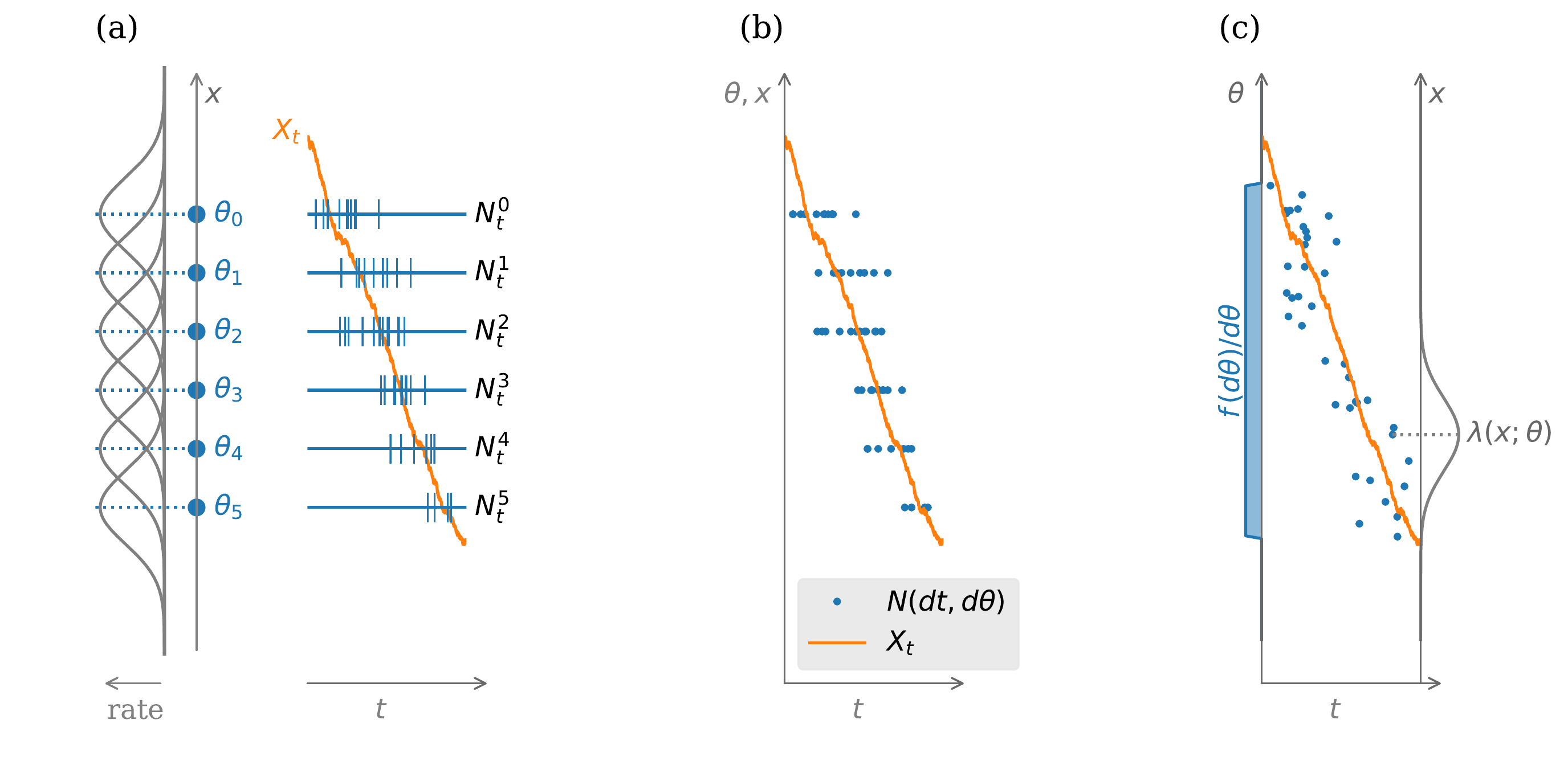}\caption{Finite (discrete) and infinite (continuous) population models. \textbf{(a)}
A finite population of six neurons with preferred stimuli $\theta_{i}$
($i=0,\ldots5$). The tuning functions are depicted on the left. The
spiking pattern of each neurons in response to the stimulus $X_{t}$
(orange line) is represented by a separate point processes $N_{t}^{i}$
(blue vertical ticks). \textbf{(b)} A representation of the same neural
response as a single marked point process. Each blue dot represents
a spike with the vertical location representing its mark $\theta$.
\textbf{(c)} A continuous population model approximating the statistics
of the same population. The approximation would be better for larger
populations. The population density $f\left(d\theta\right)/d\theta$
is depicted to the left of the plot, and the tuning function corresponding
to one of the spikes on the right.}

\label{continuous-pop}
\end{figure}
We now change the mathematical formulation and notation of our model
to accommodate parameterized continuous populations. The new formulation
is more general, and includes finite populations as a special case.
Rather than using a sequence of tuning function parameters --- such
as $(\boldsymbol{y}^{\left(i\right)})_{i=1}^{M}=(h_{i},\theta_{i},H_{i},R_{i})_{i=1}^{M}$
in the case of Gaussian neurons (\ref{eq:gauss-tc}) --- we characterize
the population by a measure $f\left(d\boldsymbol{y}\right)$, where
$f\left(Y\right)$ counts the neurons with parameters in a set $Y$,
up to some multiplicative constant. A continuous measure $f$ may
be used to approximate a large population. Accordingly, we write $\lambda\left(x;\boldsymbol{y}\right)$
for the tuning function of a neuron with parameters \textbf{$\boldsymbol{y}$},
in lieu of the previous notation $\lambda^{i}\left(x\right)$. For
example, the Gaussian case (\ref{eq:gauss-tc}) takes the form
\begin{equation}
\lambda\left(x;\boldsymbol{y}\right)=\lambda\left(x;h,\theta,H,R\right)\triangleq h\exp\left(-\frac{1}{2}\left\Vert Hx-\theta\right\Vert _{R}^{2}\right),\label{eq:gauss-tc-y}
\end{equation}
with $\boldsymbol{y}=\left(h,\theta,H,R\right)$ the parameters of
the Gaussian tuning function.

Similarly, instead of the observation processes $\{N^{i}\}_{i=1}^{M}$
counting the spikes of each neuron, we describe the spikes of all
neurons using a single \emph{marked point process} $N$ \citet{SnyMil91},
which is a random sequence of pairs $\left(t_{k},\boldsymbol{y}_{k}\right)$,
where $t_{k}\in[0,\infty)$ is the time of the $k$th point and $\boldsymbol{y}_{k}\in\boldsymbol{Y}$
its \emph{mark. }In our case, $t_{k}$ is the $k$th spike time, and
$\boldsymbol{y}_{k}\in\mathbf{Y}$ are the parameters of the spiking
neuron. Alternatively, $N$ may be described as a random discrete
measure, where $N\left(\left[s,t\right]\times Y\right)$ is the number
of spikes in the time interval $\left[s,t\right]$ from neurons with
parameters in the set $Y$. In line with the discrete measure view,
we write, for an arbitrary function $h$,
\[
\int_{\left[s,t\right]\times Y}h\left(\tau,\boldsymbol{y}\right)N\left(d\tau,d\boldsymbol{y}\right)=\sum_{k}1\left\{ s\leq t_{k}\leq t,\boldsymbol{y}_{k}\in Y\right\} h\left(t_{k},\boldsymbol{y}_{k}\right),
\]
which is the ordinary Lebesgue integral of $h$ with respect to the
discrete measure $N$. We use the notation $N_{t}\left(Y\right)\triangleq N\left([0,t]\times Y\right)$,
and when $Y=\mathbf{Y}$ we omit it and write $N_{t}$ for the total
number of spikes up to time $t$. As before, $\mathcal{N}_{t}$ denotes
the history up to time $t$ -- here including both spike times and
marks.

Figure \ref{continuous-pop} illustrates how the activity of a neural
population may be represented as a marked point process, and how the
firing statistics are approximated by a continuous population. In
this case, a homogeneous population of neurons with equally-spaced
preferred stimuli is approximated by a continuous population with
uniformly distributed preferred stimuli.

To characterize the statistics of $N$, we first consider the finite
population case. In this case, $f=\sum_{i}\delta_{\boldsymbol{y}_{i}}$
where $\delta_{\boldsymbol{y}_{i}}$ is the point mass at $\boldsymbol{y}_{i}$,
and the rate of points with marks in a set $Y\subseteq\mathbf{Y}$
at time $t$ (conditioned on $\mathcal{N}_{t},X_{\left[0,t\right]}$)
is
\[
\sum_{i:y_{i}\in Y}\lambda\left(X_{t};\boldsymbol{y}_{i}\right)=\int_{Y}\lambda\left(X_{t};\boldsymbol{y}\right)f\left(d\boldsymbol{y}\right).
\]
In the case of a general population distribution $f$, we similarly
take the integral $\int_{Y}\lambda\left(X_{t};\boldsymbol{y}\right)f\left(d\boldsymbol{y}\right)$
as the rate (or \emph{intensity}) of points with marks in $Y$ at
time $t$ conditioned on $(\mathcal{N}_{t},X_{\left[0,t\right]})$.
The random measure $\lambda\left(X_{t};\boldsymbol{y}\right)f\left(d\boldsymbol{y}\right)$
appearing in this integral is termed the \emph{intensity kernel} of
the marked point process $N$ with respect to the history $(\mathcal{N}_{t},X_{\left[0,t\right]})$
(e.g., \citet{Bremaud81}, Chapter VIII). The dynamics of $N$ may
be described heuristically by means of the intensity kernel as 
\begin{equation}
\E\left[N\left(dt,d\boldsymbol{y}\right)|X_{\left[0,t\right]},\mathcal{N}_{t}\right]=\lambda\left(X_{t};\boldsymbol{y}\right)f\left(d\boldsymbol{y}\right)dt.\label{eq:rate-heuristic}
\end{equation}
The intensity kernel with respect to $\mathcal{N}_{t}$ alone is given
by 
\begin{equation}
\E\left[N\left(dt,d\boldsymbol{y}\right)|\mathcal{N}_{t}\right]=\hat{\lambda}_{t}\left(\boldsymbol{y}\right)f\left(d\boldsymbol{y}\right)dt,\label{eq:rate-heuristic-1}
\end{equation}
where
\[
\hat{\lambda}_{t}\left(\boldsymbol{y}\right)\triangleq\E\left[\lambda\left(X_{t};\boldsymbol{y}\right)|\mathcal{N}_{t}\right]=\E_{t}^{\mathcal{N}}\left[\lambda\left(X_{t};\boldsymbol{y}\right)\right]
\]
We denote the rate of the unmarked process $N_{t}$ with respect to
$\mathcal{N}_{t}$ (i.e., the total posterior expected firing rate)
by
\[
\hat{\lambda}_{t}^{f}\triangleq\int\hat{\lambda}_{t}\left(\boldsymbol{y}\right)f\left(d\boldsymbol{y}\right)
\]

\subsubsection{Filtering\label{subsec:cont-filtering}}

Assume Gaussian tuning functions (\ref{eq:gauss-tc-y}). Writing
$\boldsymbol{y}\triangleq\left(h,\theta,H,R\right)$, the filtering
equations (\ref{eq:mean-c-gaussian-finite})-(\ref{eq:self-rate-gaussian-finite})
take the form \begin{subequations}\label{eq:filtering-gaussian}

\begin{align}
d\mu_{t}^{\mathrm{c}} & =\int_{\mathbf{Y}}\Sigma_{t}H\transpose S_{t}^{H,R}\delta_{t}^{H,\theta}\hat{\lambda}_{t}\left(\boldsymbol{y}\right)f\left(d\boldsymbol{y}\right)dt\label{eq:mean-c-gaussian}\\
d\Sigma_{t}^{\mathrm{c}} & =\int_{\mathbf{Y}}\Sigma_{t}H\transpose\left(S_{t}^{H,R}-S_{t}^{H,R}\delta_{t}^{H,\theta}(\delta_{t}^{H,\theta})\transpose S_{t}^{H,R}\right)H\Sigma_{t}\hat{\lambda}_{t}\left(\boldsymbol{y}\right)f\left(d\boldsymbol{y}\right)dt\label{eq:var-c-gaussian}\\
d\mu_{t}^{N} & =\int_{\mathbf{Y}}\Sigma_{t^{-}}H\transpose S_{t^{-}}^{H,R}\delta_{t^{-}}^{H,\theta}N\left(dt,d\boldsymbol{y}\right)\label{eq:mean-N-gaussian}\\
d\Sigma_{t}^{N} & =-\int_{\mathbf{Y}}\Sigma_{t^{-}}H\transpose S_{t^{-}}^{H,R}H\Sigma_{t^{-}}N\left(dt,d\boldsymbol{y}\right)\label{eq:var-N-gaussian}\\
\delta_{t}^{H,\theta} & \triangleq H\mu_{t}-\theta,\label{eq:delta}\\
S_{t}^{H,R} & \triangleq\left(R^{-1}+H\transpose\Sigma_{t}H\right)^{-1},\label{eq:S}\\
\hat{\lambda}_{t}\left(\boldsymbol{y}\right) & \triangleq\hat{\lambda}_{t}\left(h,\theta,H,R\right)=h\sqrt{\frac{\left|S_{t}^{H,R}\right|}{\left|R\right|}}\exp\left(-\frac{1}{2}\left\Vert \theta-H\mu_{t}\right\Vert _{S_{t}^{H,R}}^{2}\right),\label{eq:self-rate-gaussian}
\end{align}
and the posterior precision updates (\ref{eq:prec-c-gaussian-finite})-(\ref{eq:prec-N-gaussian-finite})
become
\begin{align}
d\Sigma_{t}^{-1,\mathrm{c}} & =-\int H\transpose\left(S_{t}^{H,R}-S_{t}^{H,R}\delta_{t}^{H,\theta}(\delta_{t}^{H,\theta})\transpose S_{t}^{H,R}\right)H\hat{\lambda}_{t}\left(\boldsymbol{y}\right)f\left(d\boldsymbol{y}\right)dt\label{eq:prec-c-gaussian}\\
d\Sigma_{t}^{-1,N} & =\int_{\mathbf{Y}}H\transpose RH\,N\left(dt,d\boldsymbol{y}\right)\label{eq:prec-N-gaussian}
\end{align}
\end{subequations}The derivation of these equations, as well as filtering
equations for special cases considered in this section, are found
in appendix \ref{sec:Derivation}.

Using (\ref{eq:self-rate-gaussian}), the continuous update equations
(\ref{eq:mean-c-gaussian})-(\ref{eq:var-c-gaussian}) may be evaluated
in closed form for some specific forms of the population distribution
$f$. Note that the discontinuous update equations (\ref{eq:mean-N-gaussian})-(\ref{eq:var-N-gaussian})
do not depend on $f$, and are already in closed form. We now consider
several population distributions where the continuous updates may
be brought to closed form.

\paragraph{Single neuron}

The result for a single neuron with parameters $h,\theta,H,R$ is
trivial to obtain from (\ref{eq:mean-c-gaussian})-(\ref{eq:var-c-gaussian}),
yielding\begin{subequations}\label{eq:c-single}
\begin{align}
d\mu_{t}^{\mathrm{c}} & =\Sigma_{t}H\transpose S_{t}^{H,R}\left(H\mu_{t}-\theta\right)\hat{\lambda}_{t}^{f}dt,\label{eq:mean-c-single}\\
d\Sigma_{t}^{\mathrm{c}} & =\Sigma_{t^{-}}H\transpose\left(S_{t}^{H,R}-S_{t}^{H,R}\left(H\mu_{t}-\theta\right)\left(H\mu_{t}-\theta\right)\transpose S_{t}^{H,R}\right)\hat{\lambda}_{t}^{f}H\Sigma_{t^{-}}dt,\label{eq:var-c-single}
\end{align}
\end{subequations}where $\hat{\lambda}_{t}^{f}=\hat{\lambda}_{t}\left(\boldsymbol{y}\right)$
as given by (\ref{eq:self-rate-gaussian}), and $S_{t}^{H,R}$ is
defined in (\ref{eq:S}). 

\paragraph{Uniform population}

Here all neurons share the same height $h$ and shape matrices $H,R$,
whereas the location parameter $\theta$ covers $\mathbb{R}^{m}$
uniformly, i.e. $f\left(dh',d\theta,dH',dR'\right)=\delta_{h}\left(dh'\right)\delta_{H}\left(dH'\right)\delta_{R}\left(dR'\right)d\theta$,
where $\delta_{x}$ is a Dirac measure at $x$, i.e. $\delta_{x}\left(A\right)=1\{x\in A\}$,
and $d\theta$ indicates the Lebesgue measure in the parameter $\theta$.
A straightforward calculation from (\ref{eq:mean-c-gaussian})-(\ref{eq:var-c-gaussian})
and (\ref{eq:self-rate-gaussian}) yields
\begin{align}
d\mu_{t}^{\mathrm{c}} & =0,\quad d\Sigma_{t}^{\mathrm{c}}=0,\label{eq:mean-var-uniform}
\end{align}
in agreement with the (exact) result obtained in the uniform coding
setting of \citet{RhoSny1977}, where the filtering equations only
include the prior term and the discontinuous update term.

\paragraph{Gaussian population}

As in the uniform population case, we assume all neurons share the
same height $h$ and shape matrices $H,R$, and differ only in the
location parameter $\theta$. Abusing notation slightly, we write
$f\left(dh',d\theta,dH',dR'\right)=\delta_{h}\left(dh'\right)\delta_{H}\left(dH'\right)\delta_{R}\left(dR'\right)f\left(d\theta\right)$
where the preferred stimuli are normally distributed,
\begin{align}
f\left(d\theta\right) & =\mathcal{N}\left(\theta;c,\Sigma_{\mathrm{pop}}\right)d\theta\nonumber \\
 & =\left(2\pi\right)^{-n/2}\left|\Sigma_{\mathrm{pop}}\right|^{-1/2}\exp\left(-\frac{1}{2}\left\Vert \theta-c\right\Vert _{\Sigma_{\mathrm{pop}}^{-1}}^{2}\right)d\theta,\label{eq:gaussian-f}
\end{align}
for fixed $c\in\mathbb{R}^{m}$, and positive definite $\Sigma_{\mathrm{pop}}$.

We take $f$ to be normalized, since any scaling of $f$ may be included
in the coefficient $h$ in (\ref{eq:gauss-tc-y}), resulting in the
same point process. Thus, when used to approximate a large population,
the coefficient $h$ would be proportional to the number of neurons.

The continuous updates for this case read\begin{subequations}\label{eq:c-gauss-gauss}
\begin{align}
d\mu_{t}^{\mathrm{c}} & =\Sigma_{t}H\transpose Z_{t}^{H,R}\left(H\mu_{t}-c\right)\hat{\lambda}_{t}^{f}dt,\label{eq:mean-c-gauss-gauss}\\
d\Sigma_{t}^{\mathrm{c}} & =\Sigma_{t}H\transpose\Big(Z_{t}^{H,R}-Z_{t}^{H,R}\left(H\mu_{t}-c\right)\left(H\mu_{t}-c\right)\transpose Z_{t}^{H,R}\Big)H\Sigma_{t}\hat{\lambda}_{t}^{f}dt,\label{eq:var-c-gauss-gauss}
\end{align}
where
\begin{align}
Z_{t}^{H,R} & \triangleq\left(\Sigma_{\mathrm{pop}}+R^{-1}+H\Sigma_{t}H\transpose\right)^{-1},\label{eq:Z^HR}\\
\hat{\lambda}_{t}^{f} & \triangleq\int\hat{\lambda}_{t}\left(\boldsymbol{y}\right)f\left(d\boldsymbol{y}\right)=h\sqrt{\frac{\left|Z_{t}^{H,R}\right|}{\left|R\right|}}\exp\left(-\frac{1}{2}\left\Vert H\mu_{t}-c\right\Vert _{Z_{t}^{H,R}}^{2}\right).\label{eq:total-rate-gauss-gauss}
\end{align}
\end{subequations}These updates generalize the single-neuron updates
(\ref{eq:c-single}), with the population center $c$ taking the place
of the location parameter $\theta$, and $Z_{t}^{H,R}$ substituting
$S_{t}^{H,R}$. The single-neuron case is obtained when $\Sigma_{\mathrm{pop}}=0$.

It is illustrative to consider these equations in the scalar case
$m=n=1$, with $H=1$. Letting $\sigma_{t}^{2}=\Sigma_{t},\alpha^{2}=R^{-1},\sigma_{\mathrm{pop}}^{2}=\Sigma_{\mathrm{pop}}$
yields\begin{subequations}\label{eq:c-gauss-gauss-1d}
\begin{align}
d\mu_{t}^{\mathrm{c}} & =\frac{\sigma_{t}^{2}}{\sigma_{t}^{2}+\alpha^{2}+\sigma_{\mathrm{pop}}^{2}}\left(\mu_{t}-c\right)\hat{\lambda}_{t}^{f}dt,\label{eq:mean-c-gauss-gauss-1d}\\
d\sigma_{t}^{2,\mathrm{c}} & =\frac{\sigma_{t}^{2}}{\sigma_{t}^{2}+\alpha^{2}+\sigma_{\mathrm{pop}}^{2}}\left(1-\frac{\left(\mu_{t}-c\right)^{2}}{\sigma_{t}^{2}+\alpha^{2}+\sigma_{\mathrm{pop}}^{2}}\right)\sigma_{t}^{2}\hat{\lambda}_{t}^{f}dt,\label{eq:var-c-gauss-gauss-1d}
\end{align}
where 
\[
\hat{\lambda}_{t}^{f}=h\sqrt{2\pi\alpha^{2}}\mathcal{N}\left(\mu_{t};c,\sigma_{t}^{2}+\alpha^{2}+\sigma_{\mathrm{pop}}^{2}\right).
\]
\end{subequations}

Figure \ref{gaussian} demonstrates the continuous update terms (\ref{eq:c-gauss-gauss-1d})
as a function of the current mean estimate $\mu_{t}$, for various
values of the population variance $\sigma_{\mathrm{pop}}^{2}$, including
the case of a single neuron, $\sigma_{\mathrm{pop}}^{2}=0$. The continuous
update term $d\mu_{t}^{\mathrm{c}}$ pushes the posterior mean $\mu_{t}$
away from the population center $c$ in the absence of spikes. This
effect weakens as $\left|\mu_{t}-c\right|$ grows due to the factor
$\hat{\lambda}_{t}^{f}$, consistent with the idea that far from $c$,
the lack of events is less surprising, hence less informative. The
continuous variance update term $d\sigma_{t}^{2,\mathrm{c}}$ increases
the variance when $\mu_{t}$ is near $\theta$, otherwise decreases
it. This stands in contrast with the Kalman-Bucy filter, where the
posterior variance cannot increase when estimating a static state.

\begin{figure*}
\subfloat[Gaussian population (\ref{eq:c-gauss-gauss-1d}). Parameters are $c=0$
(population centered at the origin), $\alpha^{2}=0.25,h=1,\sigma_{t}^{2}=1$.
The case $\sigma_{\mathrm{pop}}^{2}=0$ corresponds to a single sensory
neuron.]{\includegraphics[width=0.45\columnwidth]{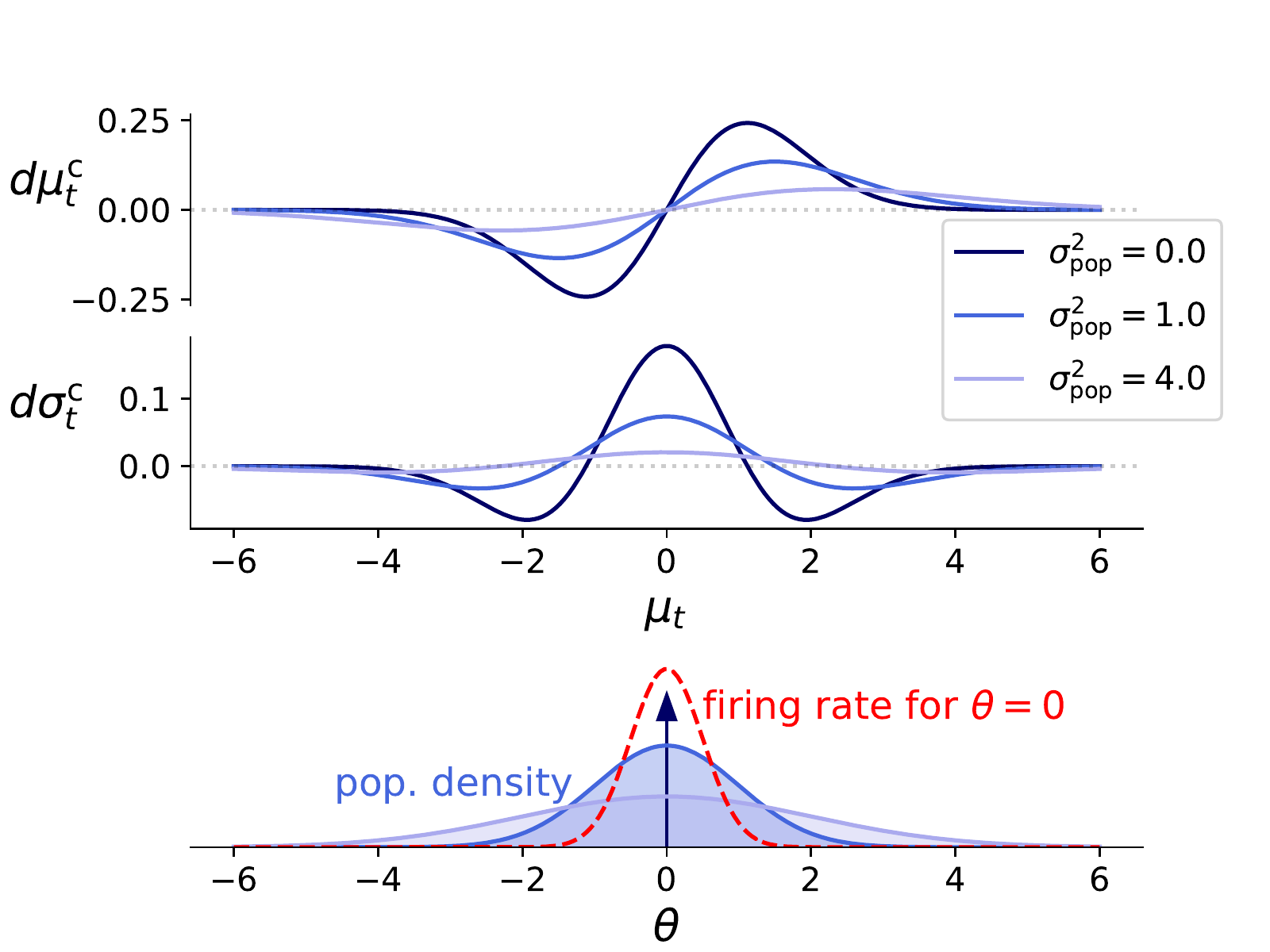}\label{gaussian}}\hfill{}\subfloat[{Uniform population on an interval (\ref{eq:c-interval}). Parameters
are $\left[a,b\right]=\left[-1,1\right],h=1,\sigma_{t}^{2}=0.01$.}]{\includegraphics[width=0.45\columnwidth]{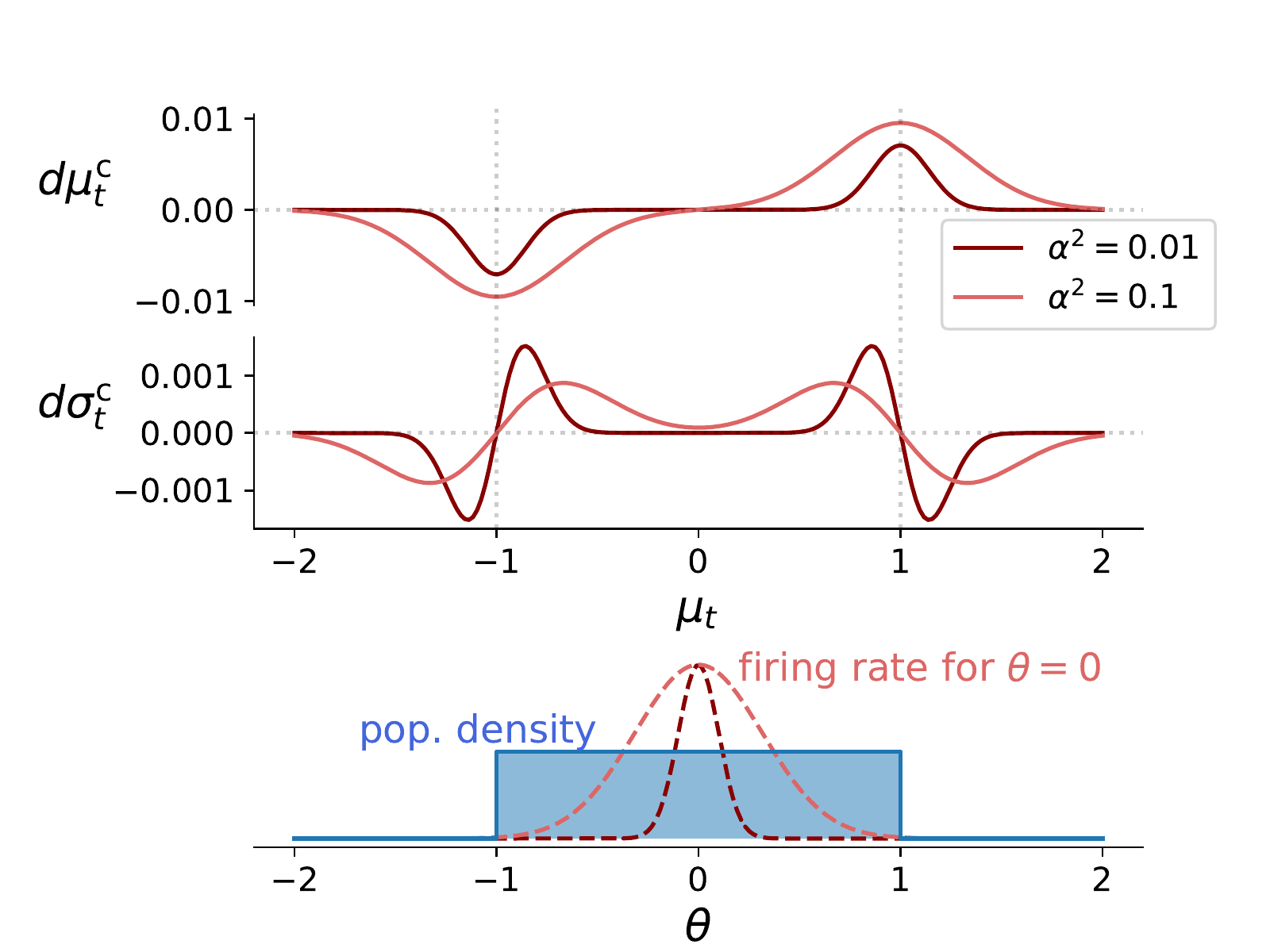}

\label{interval}}\caption{Continuous update terms as a function of the current posterior mean
estimate, for a 1-d state observed through a population of Gaussian
neurons (\ref{eq:gauss-tc}). The population density is Gaussian on
the left plot, and uniform on the interval $\left[-1,1\right]$ on
the right plot. The bottom plots shows the population density $f\left(d\theta\right)/d\theta$
and a tuning function $\lambda\left(x;\theta\right)$ for $\theta=0$.}
\label{continuous-terms}
\end{figure*}

\paragraph{Uniform population on an interval}

In this case we assume a scalar state, $n=m=1$, and 
\begin{align}
f\left(d\theta\right) & =1_{\left[a,b\right]}\left(\theta\right)d\theta,\label{eq:interval}
\end{align}
where similarly to the Gaussian population case, $h$ and $R$ are
fixed. Unlike the Gaussian case, here we find it more convenient not
to normalize the distribution. Since the state is assumed to be scalar,
let $\sigma_{t}^{2}=\Sigma_{t},\alpha^{2}=R^{-1}$. The continuous
updates for this case are\begin{subequations}\label{eq:c-interval}
\begin{align}
d\mu_{t}^{\mathrm{c}} & =h\sqrt{2\pi\alpha^{2}}\sqrt{\frac{\sigma_{t}^{2}}{\sigma_{t}^{2}+\alpha^{2}}}\left(\phi\left(b'_{t}\right)-\phi\left(a'_{t}\right)\right)\sigma_{t}dt\label{eq:mean-c-interval}\\
d\sigma_{t}^{2,\mathrm{c}} & =h\sqrt{2\pi\alpha^{2}}\frac{\sigma_{t}^{2}}{\sigma_{t}^{2}+\alpha^{2}}\left(b'_{t}\phi\left(b'_{t}\right)-a'_{t}\phi\left(a'_{t}\right)\right)\sigma_{t}^{2}dt\label{eq:var-c-interval}
\end{align}
where
\[
a'_{t}\triangleq\frac{a-\mu_{t}}{\sqrt{\sigma_{t}^{2}+\alpha^{2}}}\,,\quad b'_{t}\triangleq\frac{b-\mu_{t}}{\sqrt{\sigma_{t}^{2}+\alpha^{2}}}\,,\quad\phi\left(x\right)\triangleq\mathcal{N}\left(x;0,1\right),
\]
and the posterior rates are given by
\begin{align*}
\hat{\lambda}_{t}\left(\theta\right) & =h\sqrt{2\pi\alpha^{2}}\mathcal{N}\left(\theta;\mu_{t},\alpha^{2}+\sigma_{t}^{2}\right),\\
\hat{\lambda}_{t}^{f} & =h\sqrt{2\pi\alpha^{2}}\int_{a'_{t}}^{b'_{t}}\phi.
\end{align*}
\end{subequations}

Figure \ref{interval} demonstrates the continuous update terms (\ref{eq:c-interval})
as a function of the current mean estimate $\mu_{t}$. When the mean
estimate is around an endpoint of the interval, the mean update $\mu_{t}^{\mathrm{c}}$
pushes the posterior mean outside the interval in the absence of spikes.
The posterior variance $\sigma_{t}^{2}$ decreases outside the interval,
where the absence of spikes is expected, and increases inside the
interval, where it is unexpected\footnote{This holds only approximately, when the tuning width is not too large
relative to the size of the interval. For wider tuning functions the
behavior becomes similar to the single sensor case.}. When the posterior mean is not near the interval endpoints, the
updates are near zero, consistently with the uniform population case
(\ref{eq:mean-var-uniform}).

\paragraph{Finite mixtures}

Note that the continuous updates (\ref{eq:mean-c-gaussian})-(\ref{eq:var-c-gaussian})
are linear in $f$. Accordingly, if $f\left(dy\right)=\sum_{i}\alpha_{i}f_{i}\left(dy\right),$
where each $f_{i}$ is of one of the above forms, the updates are
obtained by the appropriate weighted sums of the filters derived above
for the various special forms of $f_{i}$. This form is quite general:
it includes populations where $\theta$ is distributed according to
a Gaussian mixture, as well as heterogeneous populations with finitely
many different values of the shape matrices $H,R$\@. The resulting
filter includes a term for each component of the mixture.

\section{Numerical evaluation}

\begin{figure*}
\captionsetup{position=top}\subfloat[High firing rate: $h=1000$]{\includegraphics[width=0.5\columnwidth]{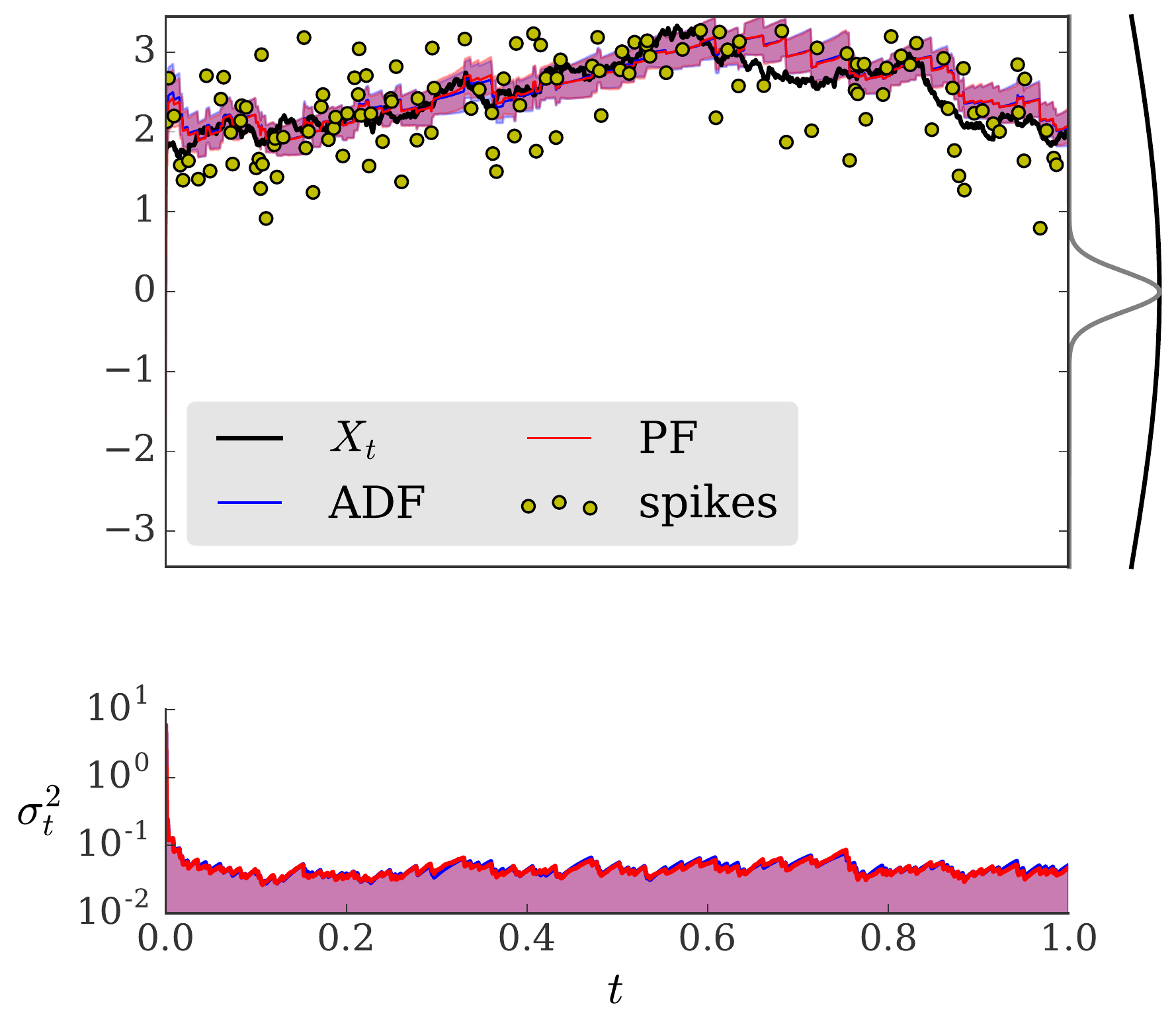}\label{high-rate}}\subfloat[Low firing rate: $h=2$]{\includegraphics[width=0.5\columnwidth]{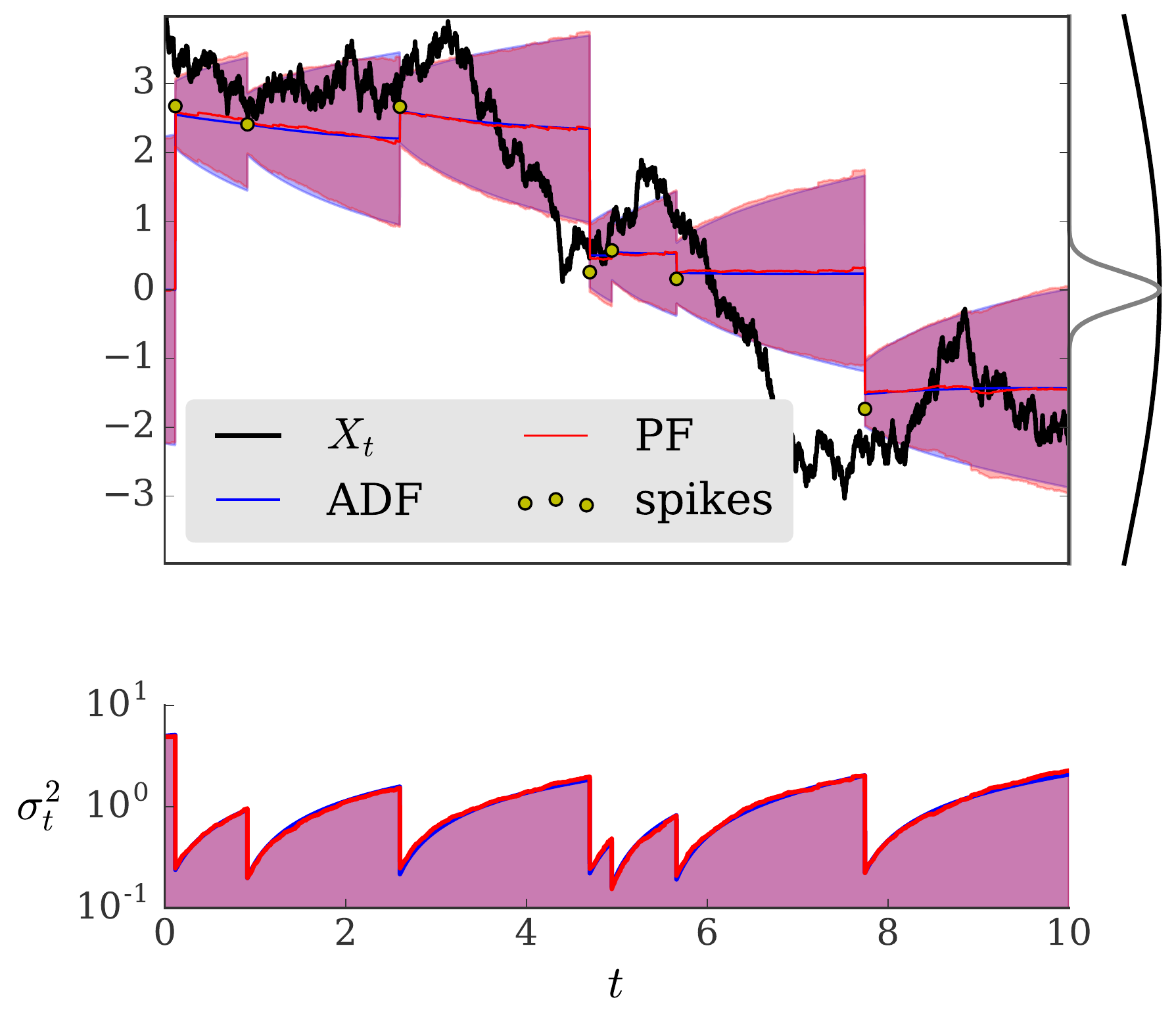}\label{low-rate}}\caption{Two examples of a linear one-dimensional process observed through
a Gaussian population (\ref{eq:gaussian-f}) of Gaussian neurons (\ref{eq:gauss-tc}),
filtered using the ADF approximation (namely, equations (\ref{eq:p-linear}),(\ref{eq:mean-c-gaussian})-(\ref{eq:var-c-gaussian}),(\ref{eq:c-gauss-gauss})),
and using a particle filter. Each dot correspond to a spike with the
vertical location indicating the neuron's preferred stimulus $\theta$.
The approximate posterior means obtained from ADF and particle filtering
are shown in blue and red respectively, with the corresponding posterior
standard deviations in shaded areas of the respective colors. The
curves to the right of the graph show the preferred stimulus density
(black), and a neuron's tuning function (\ref{eq:gauss-tc}) centered
at $\theta=0$ (gray), arbitrarily normalized to the same height for
visualization. The bottom graph shows the posterior variance. Parameters
used in both examples: $a=-0.1,d=1,H=1,\sigma_{\mathrm{pop}}^{2}=4,R^{-1}=0.25,c=0,\mu_{0}=0,\sigma_{0}^{2}=1$
(note the different extent of the time axes). The observed processes
were initialized from their steady-state distribution. The dynamics
were discretized with time step $\Delta t=10^{-3}$. The particle
filter uses 1000 particles with systematic resampling (see, e.g.,
\citet{DouJoh09}) at each time step.}

\label{1d-filtering}
\end{figure*}
Since the filter (\ref{eq:filtering-gaussian}) is based on an assumed
density approximation, its results may be inexact. We tested the accuracy
of the filter in the Gaussian population case (\ref{eq:c-gauss-gauss}),
by numerical comparison with Particle Filtering (PF) \citet{DouJoh09}.

Figure \ref{1d-filtering} shows two examples of filtering a one-dimensional
process observed through a Gaussian population (\ref{eq:gaussian-f})
of Gaussian neurons (\ref{eq:gauss-tc}), using both the ADF approximation
(\ref{eq:c-gauss-gauss}) and a Particle Filter (PF) for comparison.
See the figure caption for precise details. Figure \ref{particle-error-ks-1d}
shows the distribution of approximation errors and the deviation of
the posterior from Gaussian. The approximation errors plotted are
the relative error in the mean estimate $\epsilon_{\mu}\triangleq\left(\mu_{\mathrm{ADF}}-\mu_{\mathrm{PF}}\right)/\sigma_{\mathrm{PF}}$,
and the error in the posterior standard deviation estimate $\epsilon_{\sigma}\triangleq\left(\sigma_{\mathrm{ADF}}-\sigma_{\mathrm{PF}}\right)/\sigma_{\mathrm{PF}}$,
where $\mu_{\mathrm{ADF}},\mu_{\mathrm{PF}},\sigma_{\mathrm{ADF}},\sigma_{\mathrm{PF}}$
are, respectively, the posterior mean obtained from ADF and PF, and
the posterior standard deviation obtained from ADF and PF. The deviation
of the posterior distribution from Gaussian is quantified using the
Kolmogorov-Smirnov (KS) statistic $\sup_{x}\left|F\left(x\right)-G\left(x\right)\right|$
where $F$ is the particle distribution cdf and $G$ is the cdf of
a Gaussian matching $F$’s first two moments. For comparison, the
orange lines in Figure \ref{particle-error-ks-1d} show the distribution
of this KS statistic under the hypothesis that the particles are drawn
independently from a Gaussian, which is known as the Lilliefors distribution
(see \citet{Lilliefors1967}). As seen in the figure, the Gaussian
posterior assumption underlying the ADF approximation is quite accurate
despite the fact that the population is non-uniform. Accordingly,
approximation errors are typically of a few percent (see Table \ref{errors-particle}).

\begin{figure}
\begin{centering}
\captionsetup{position=top}\subfloat[High firing rate: $h=1000$]{\includegraphics[width=0.5\columnwidth]{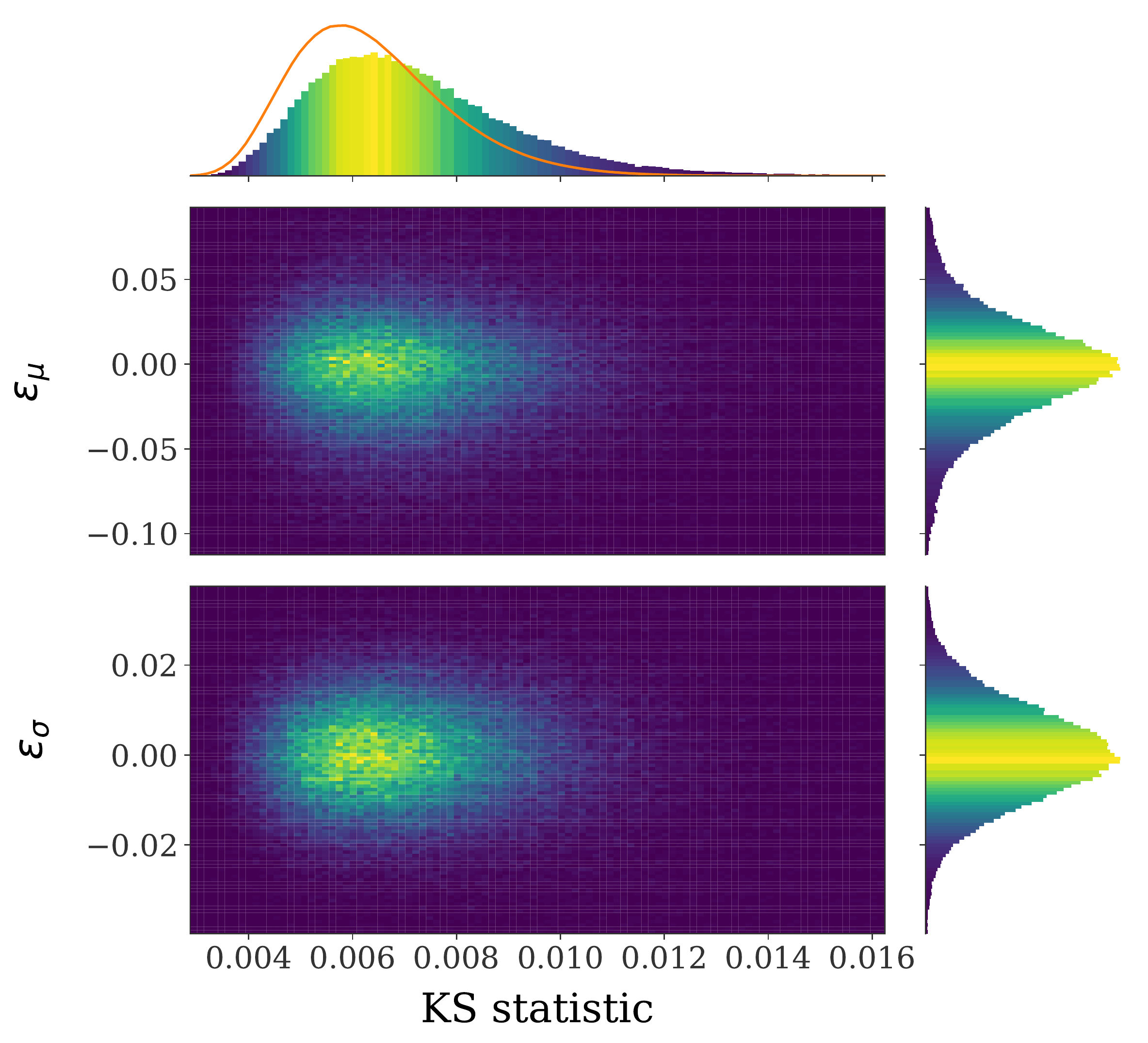}}\subfloat[Low firing rate: $h=2$]{\includegraphics[width=0.5\columnwidth]{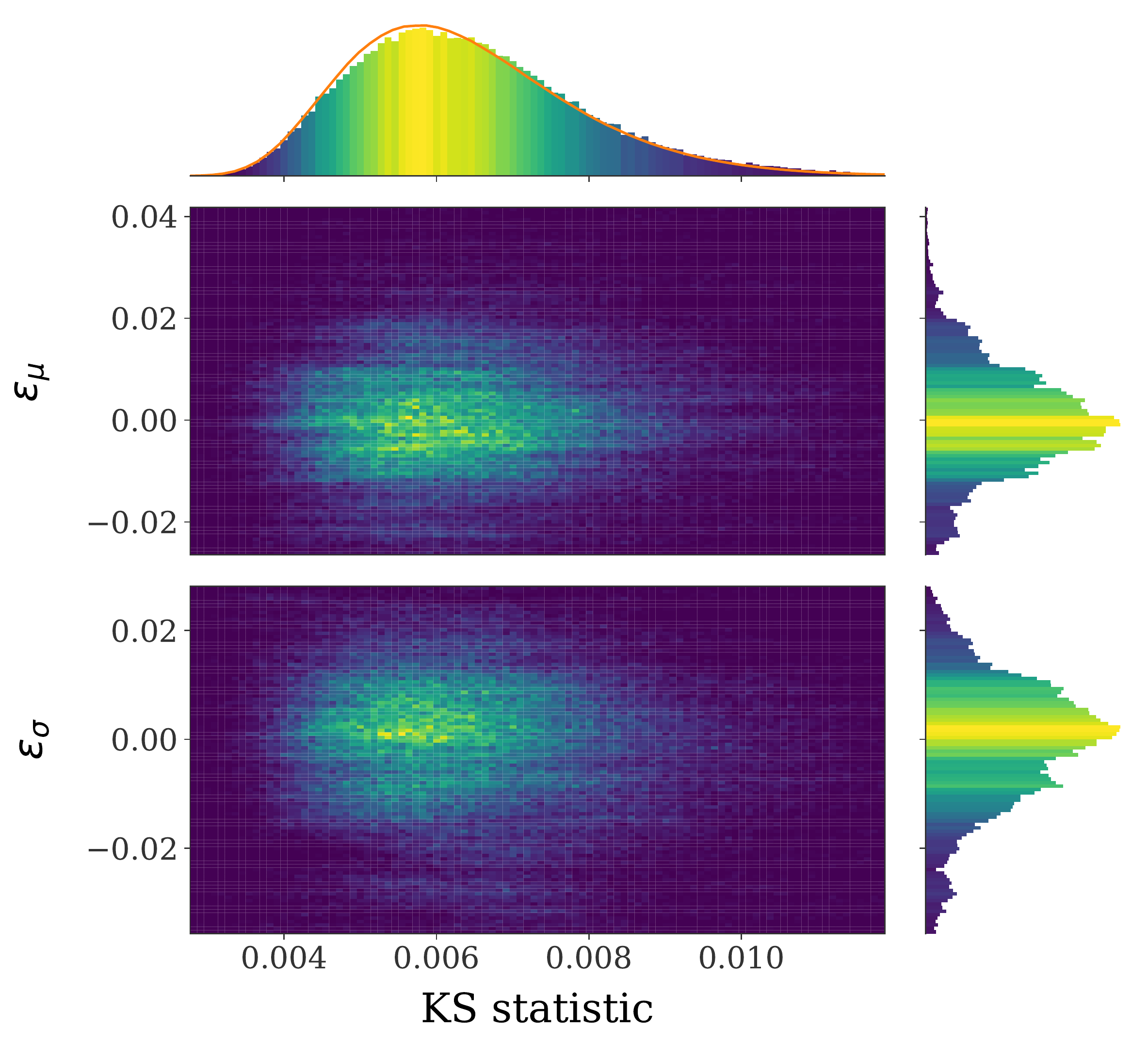}}
\par\end{centering}
\caption{Two-dimensional histograms of approximation errors (relative to particle
filter) vs. KS statistic of particle distribution, with color indicating
count of time steps falling within each two-dimensional bin. The KS
statistic is plotted against the estimation errors $\epsilon_{\mu}$,
$\epsilon_{\sigma}$ (see main text). For clearer visualization, trials
were omitted from the 2d histograms where the respective errors are
below their 0.5-percentile value or over their 99.5-percentile value,
as well as trials where the KS statistic is over its 99.5-percentile
value. The one-dimensional histograms above and to the right of each
plot show the distribution of the KS statistic, of $\epsilon_{\mu}$
and of $\epsilon_{\sigma}$. Orange lines represent the theoretical
distribution of the KS statistic for Gaussian data (Lilliefors distribution),
obtained via Monte Carlo simulation. Data was collected from 100 trials
of 1000 time steps with step size $\Delta t=10^{-3}$, using 10,000
particles with systematic resampling at each time step.}
\label{particle-error-ks-1d}
\end{figure}
\begin{figure}
\captionsetup{position=top}\subfloat[Position]{\includegraphics[height=0.42\columnwidth]{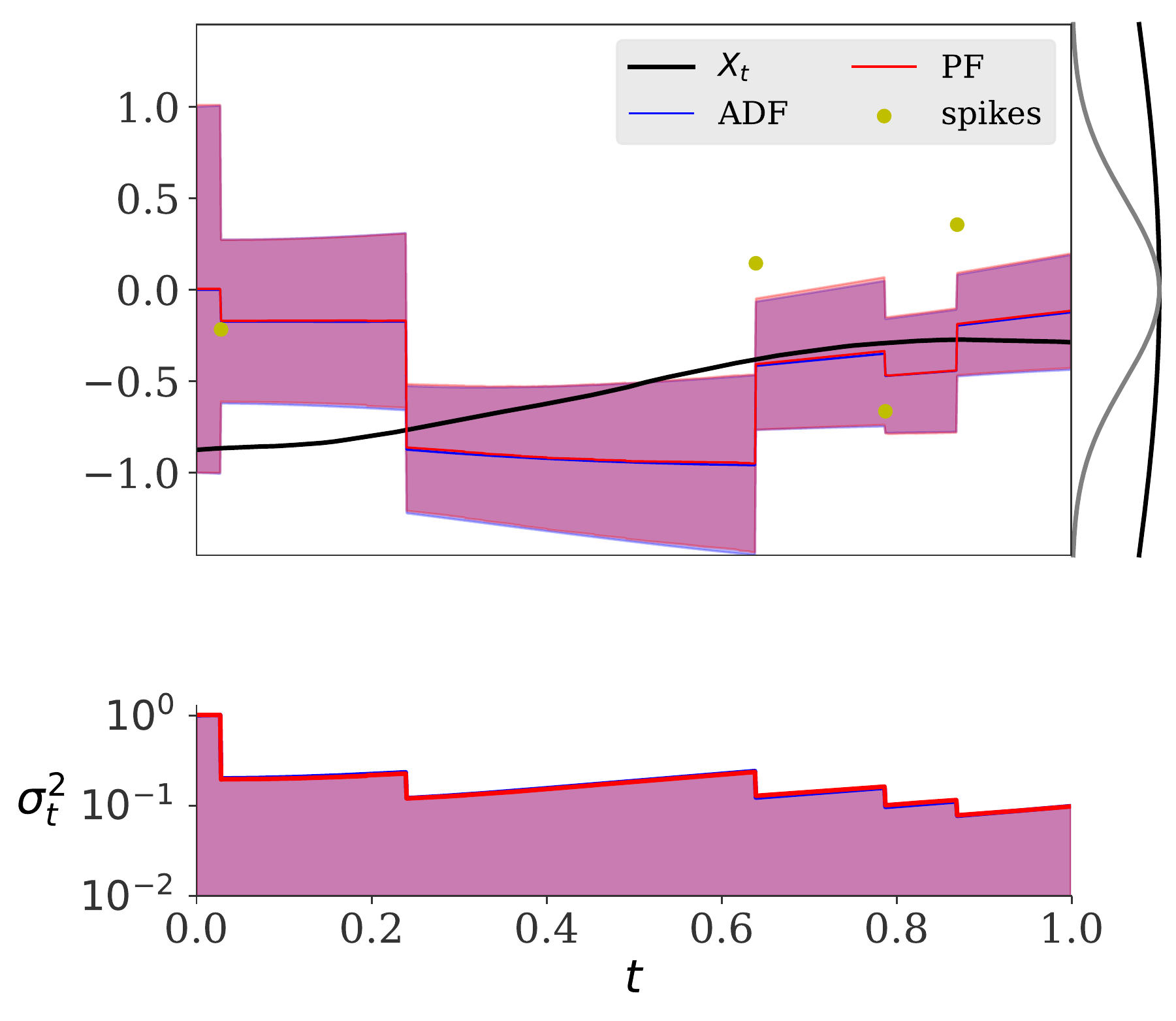}

}\hfill{}\subfloat[Velocity]{\includegraphics[height=0.42\columnwidth]{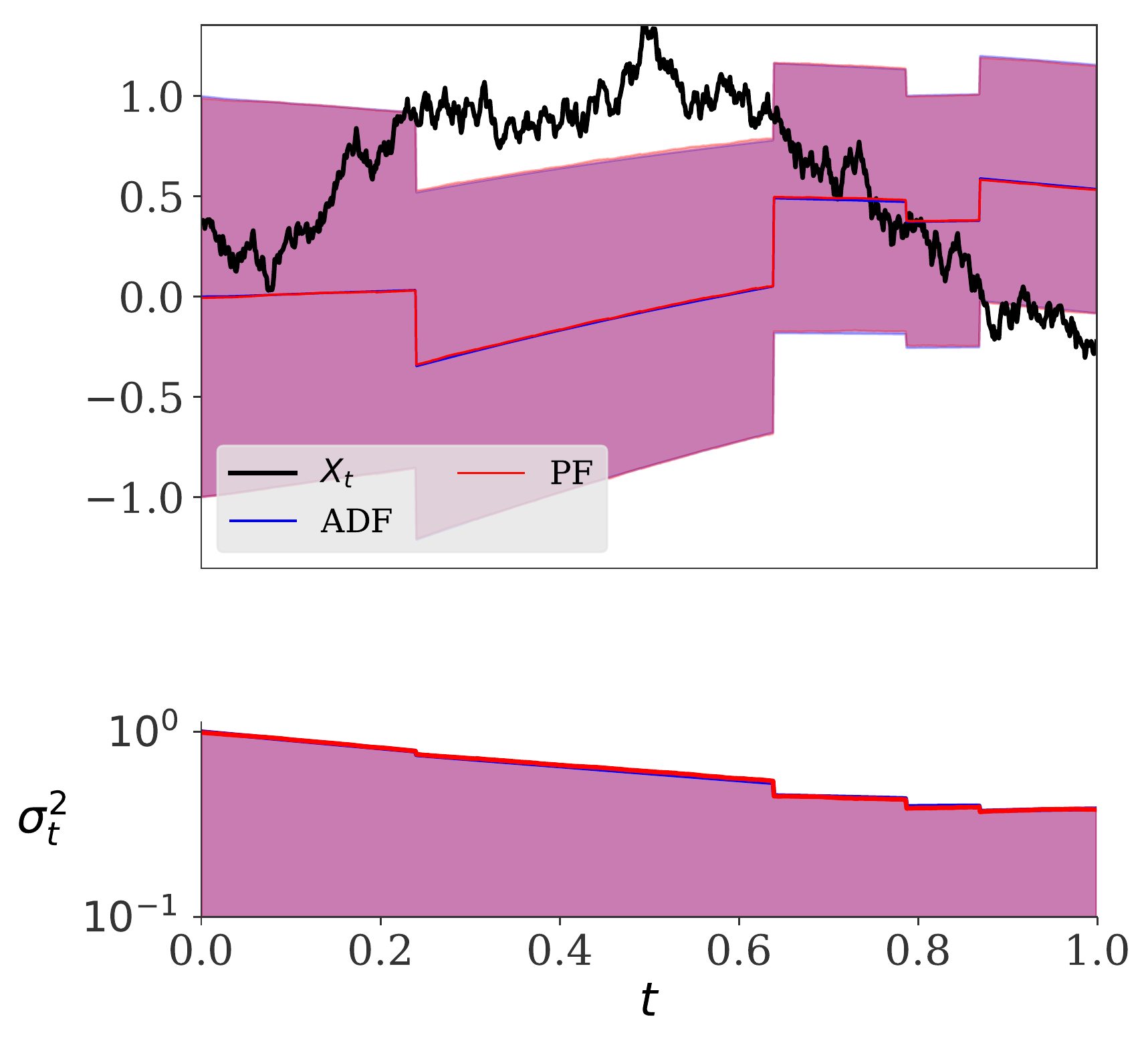}

}\caption{An example two-dimensional process (\ref{eq:2d-particle}) observed
through a Gaussian population (\ref{eq:gaussian-f}) of Gaussian neurons
(\ref{eq:gauss-tc}) (see Figure \ref{1d-filtering}). The sensory
population is Gaussian (\ref{eq:gauss-tc}),(\ref{eq:gaussian-f}),
with $H=\left[1,0\right]$, so only the particle position is directly
observed. Other sensory parameters are $\alpha^{2}=R^{-1}=0.25,h=10,\Sigma_{\mathrm{pop}}=4$.
The process was initialized from $\mathcal{N}\left([0,0],I_{2}\right)$
where $I_{2}$ is the $2\times2$ identity matrix. Other details are
as in Figure \ref{1d-filtering}.}

\label{2d-filtering}
\end{figure}
\begin{figure}
\begin{centering}
\captionsetup{position=top}\subfloat[Position]{

\includegraphics[width=0.5\columnwidth]{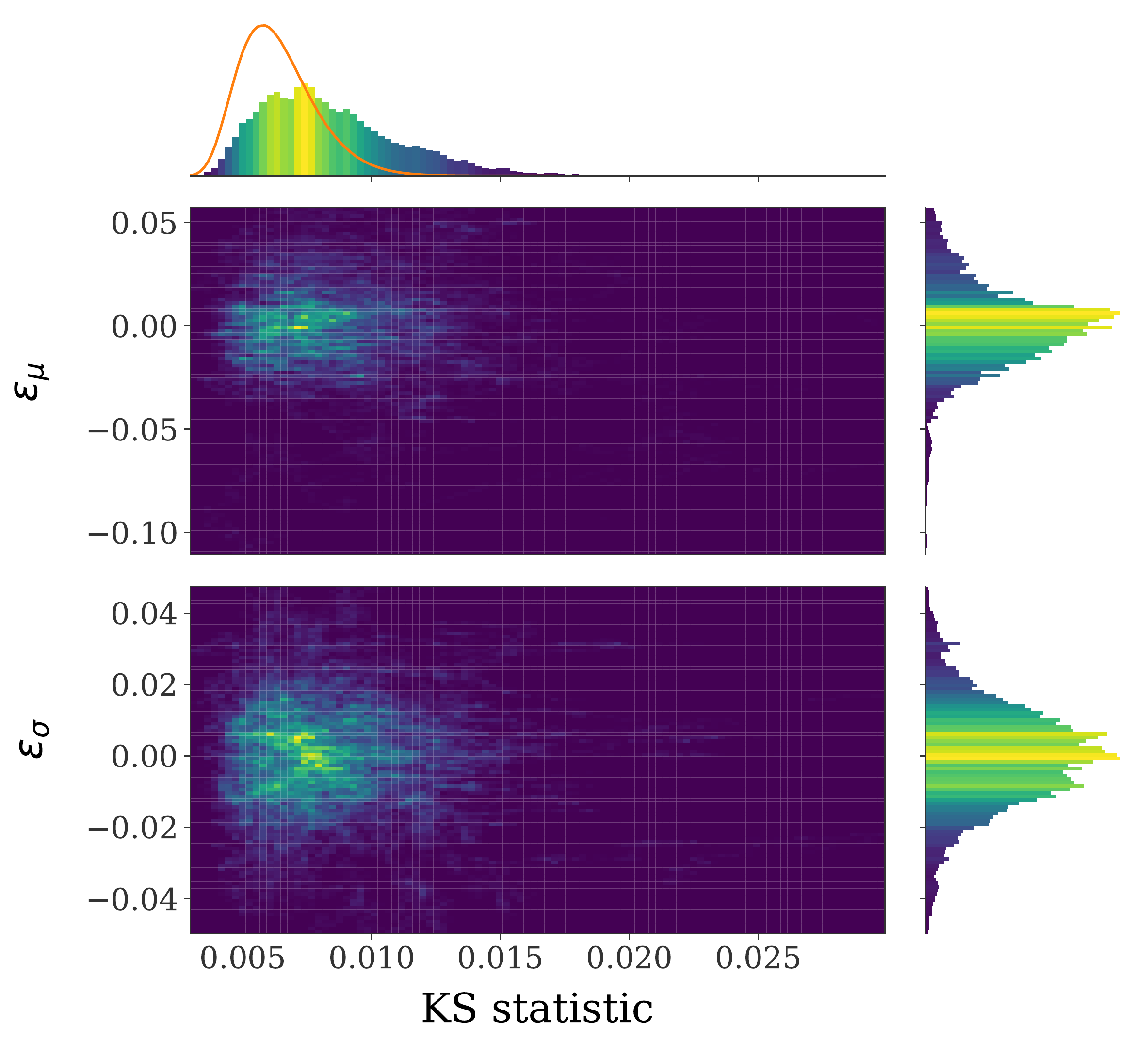}}\subfloat[Velocity]{\includegraphics[width=0.5\columnwidth]{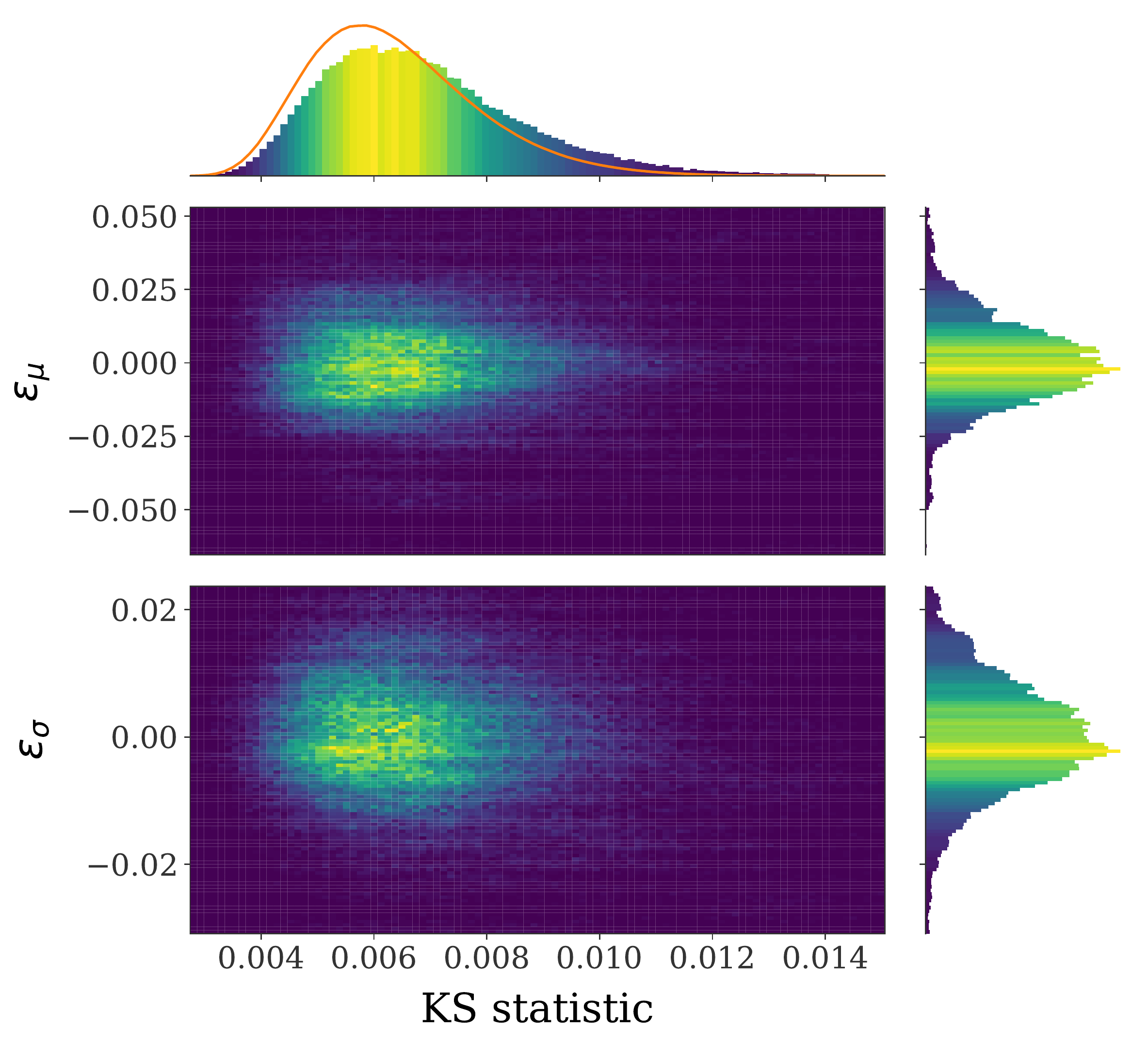}}
\par\end{centering}
\caption{Two-dimensional histograms of approximation errors (relative to particle
filter) vs. KS statistic of particle distribution in a two-dimensional
example (\ref{eq:2d-particle}). Details are as in Figure (\ref{particle-error-ks-1d}).}
\label{particle-error-ks-2d}
\end{figure}
Figure \ref{2d-filtering} shows an example of filtering a two-dimensional
process with dynamics
\begin{equation}
dX_{t}=\left(\begin{array}{cc}
0 & 1\\
0 & -0.1
\end{array}\right)X_{t}dt+\left(\begin{array}{c}
0\\
1
\end{array}\right)dW_{t},\label{eq:2d-particle}
\end{equation}
which may be interpreted as the position and velocity of a particle
subject to friction proportional to its velocity as well as ``Gaussian
white noise'' external force. In this example, only the position
is directly observed by the neural population. Additional details
are given in the figure caption. The distribution of approximation
errors and the KS statistic in this two-dimensional setting is shown
in Figure \ref{particle-error-ks-2d}. The approximation errors plotted
are $\epsilon_{\mu}$ and $\epsilon_{\sigma}$ as defined above; both
these errors and the KS statistic are computed separately for each
dimension.

Statistics of the estimation error distribution for these examples
are provided in Table \ref{errors-particle}. 

\begin{table}
\vspace{-10bp}
\caption{Approximation errors relative to PF in the examples of Figures \ref{1d-filtering}
and \ref{2d-filtering}.}
\vspace{-15bp}

\subfloat[1d example (Figure \ref{1d-filtering})]{

\begin{tabular}{cccccc}
\toprule 
 & \multicolumn{2}{c}{$h=1000$} &  & \multicolumn{2}{c}{$h=2$}\tabularnewline
\cmidrule{2-3} \cmidrule{5-6} 
 & $\epsilon_{\mu}$ & $\epsilon_{\sigma}$ &  & $\epsilon_{\mu}$ & $\epsilon_{\sigma}$\tabularnewline
\midrule
median & -0.00272 & $1.29\times10^{-4}$ &  & $-2.84\times10^{-4}$ & $2.96\times10^{-4}$\tabularnewline
\midrule 
5th perc. & -0.0601 & -0.0185 &  & -0.0184 & -0.0245\tabularnewline
\midrule 
95th perc. & 0.0482 & 0.0192 &  & 0.0186 & 0.0178\tabularnewline
\midrule 
mean & -0.00415 & $1.41\times10^{-4}$ &  & $3.34\times10^{-4}$ & $-9.35\times10^{-4}$\tabularnewline
\midrule 
std. dev. & 0.0345 & 0.0126 &  & 0.0119 & 0.0122\tabularnewline
\midrule 
med. abs. value & 0.0188 & 0.00722 &  & 0.00662 & 0.00766\tabularnewline
\midrule 
mean abs. value & 0.0251 & 0.00919 &  & 0.0086 & 0.00942\tabularnewline
\bottomrule
\end{tabular}}

\subfloat[2d example (Figure \ref{2d-filtering})]{%
\begin{tabular}{cccccc}
\toprule 
 & \multicolumn{2}{c}{position} &  & \multicolumn{2}{c}{velocity}\tabularnewline
\cmidrule{2-3} \cmidrule{5-6} 
 & $\epsilon_{\mu}$ & $\epsilon_{\sigma}$ &  & $\epsilon_{\mu}$ & $\epsilon_{\sigma}$\tabularnewline
\midrule 
median & $-1.32\times10^{-4}$ & $2.28\times10^{-4}$ &  & $-3.37\times10^{-4}$ & $-1.53\times10^{-4}$\tabularnewline
\midrule 
5th perc. & -0.0337 & -0.0253 &  & -0.0234 & -0.0148\tabularnewline
\midrule 
95th perc. & 0.0361 & 0.0257 &  & 0.0258 & 0.0154\tabularnewline
\midrule 
mean & -0.00101 & $2.95\times10^{-5}$ &  & $1.51\times10^{-5}$ & $2.95\times10^{-5}$\tabularnewline
\midrule 
std. dev. & 0.0236 & 0.0157 &  & 0.0169 & 0.00922\tabularnewline
\midrule 
med. abs. value & 0.0115 & 0.00920 &  & 0.00908 & 0.00564\tabularnewline
\midrule 
mean abs. value & 0.0163 & 0.0118 &  & 0.0121 & 0.00711\tabularnewline
\bottomrule
\end{tabular}}

\label{errors-particle}
\end{table}

\section{Encoding\label{sec:Encoding}}

We demonstrate the use of the Assumed Density Filter in determining
optimal encoding strategies, i.e., selecting the optimal population
parameters $\phi$ (see Section \ref{sec:Problem-Overview}). To illustrate
the use of ADF for the encoding problem, we consider two simple examples.
We also use the first example as a test for the filter's robustness.
We will study optimal encoding issues in more detail in a subsequent
paper.

\subsection{Optimal encoding depends on prior variance}

Previous work using a finite neuron population and a Fisher information-based
criterion \citet{HarMcAlp04} has suggested that the optimal distribution
of preferred stimuli depends on the prior variance. When it is small
relative to the tuning width, optimal encoding is achieved by placing
all preferred stimuli at a fixed distance from the prior mean. On
the other hand, when the prior variance is large relative to the tuning
width, optimal encoding is uniform (see figure 2 in \citet{HarMcAlp04}).
These results are consistent with biological observations reported
in \citet{Brand2002} concerning the encoding of aural stimuli.

Similar results are obtained with our model, as shown in Figure \ref{harper}.
Whereas \citet{HarMcAlp04} implicitly assumed a static state in the
computation of Fisher information, we use a time-varying scalar state.
The state obeys the dynamics
\begin{equation}
dX_{t}=aX_{t}dt+s\,dW_{t}\quad\left(a<0\right),\label{eq:1d-diffusion}
\end{equation}
and is observed through a Gaussian population (\ref{eq:gaussian-f})
and filtered using the ADF approximation. In this case, optimal encoding
is interpreted as the simultaneous optimization of the population
center $c$ and the population variance $\sigma_{\mathrm{pop}}^{2}$.
The process is initialized so that it has a constant prior distribution,
its variance given by $s^{2}/\left(2\left|a\right|\right)$. In Figure
\ref{harper} (left), the steady-state prior distribution is narrow
relative to the tuning width, leading to an optimal population with
a narrow population distribution far from the origin. In Figure \ref{harper}
(right), the prior is wide relative to the tuning width, leading to
an optimal population with variance that roughly matches the prior
variance.

Note that we are optimizing only the two parameters $c,\sigma_{\mathrm{pop}}^{2}$
rather than each preferred stimulus as in \citet{HarMcAlp04}. This
mitigates the considerable additional computational cost due to simulating
the decoding process, rather than computing Fisher information. This
direct simulation, though more computationally expensive, offers two
advantages over the Fisher information-based method which is used
in \citet{HarMcAlp04} and which is prevalent in computational neuroscience.
First, the simple computation of Fisher information from tuning functions,
commonly used in the neuroscience literature, is based on the assumption
of a static state, whereas our method can be applied in a fully dynamic
context, including the presence of observation-dependent feedback.
Second, simulations of the decoding process allow for the minimization
of arbitrary criteria, including the direct minimization of posterior
variance or Mean Square Error (MSE). As discussed in the introduction,
Fisher information may be a poor proxy for the MSE for finite decoding
times.

\begin{figure*}
\includegraphics[width=0.45\columnwidth]{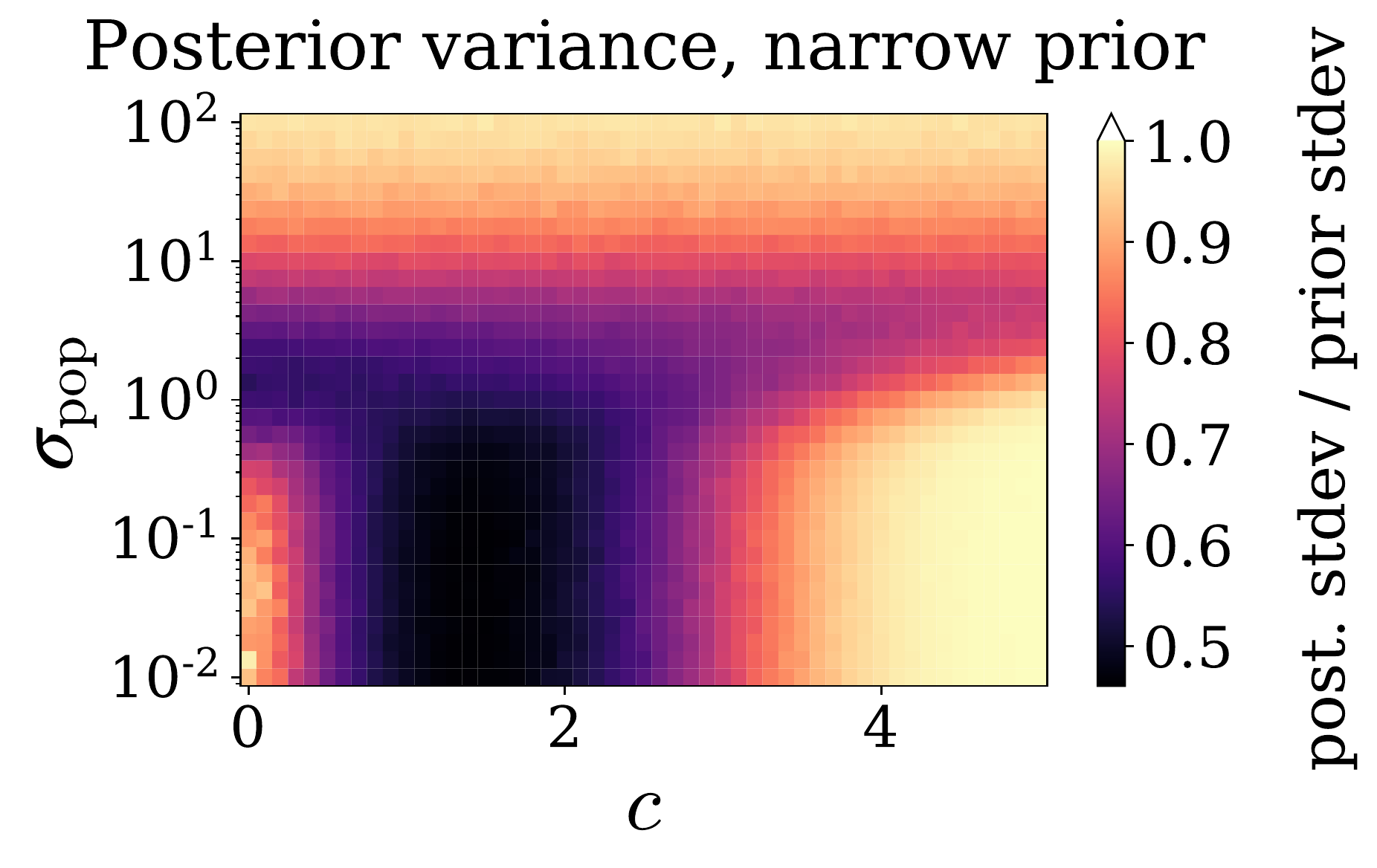}\hfill{}\includegraphics[width=0.45\columnwidth]{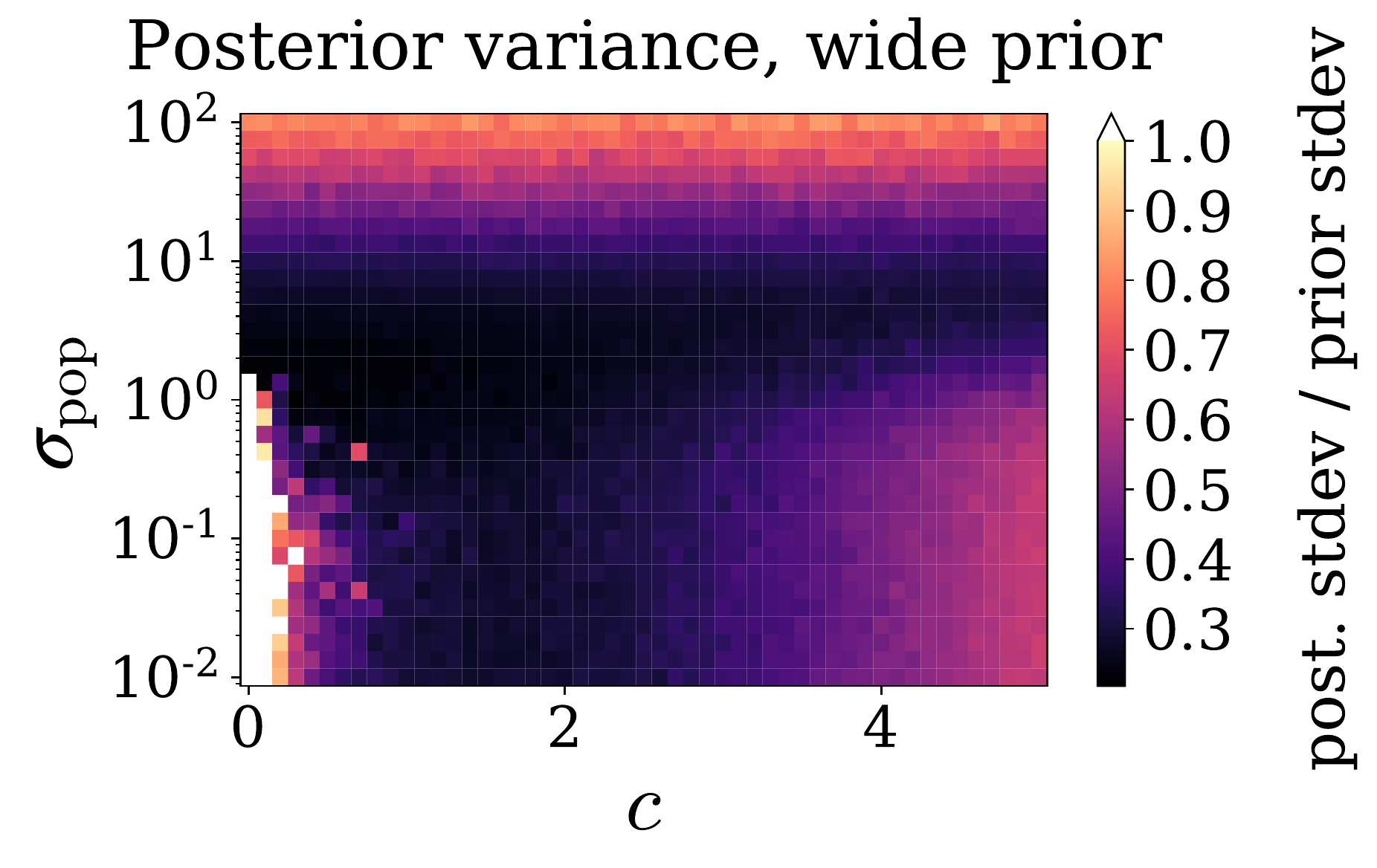}\caption{Optimal population distribution depends on prior variance relative
to tuning function width. A scalar state with dynamics $dX_{t}=aX_{t}+s\,dW_{t}$
($a=-0.05$) is filtered with tuning parameters $h=50,\alpha=1$ and
preferred stimulus density $\mathcal{N}(c,\sigma_{\mathrm{pop}}^{2})$.
The process is initialized from its steady state distribution, $\protect\mc N\left(0,s^{2}/\left(-2a\right)\right)$.
Both graphs show the posterior standard deviation derived from the
ADF approximation, relative to the prior standard deviation $\sigma_{\mathrm{p}}$.
In the left graph, $s=0.1$ so that the prior variance is $0.1$,
whereas on the right, $s=0.5$, so that the prior variance is 2.5.
In both cases the filter is initialized with the correct prior, and
the posterior variance is averaged over the time interval $\left[1,2\right]$
and across 1000 trials for each data point. Only non-negative values
of $c$ were simulated, but note that the process is symmetric about
zero, so that the full plots would also be symmetric about $c=0$.
The areas colored in white in the right plot correspond to parameters
where the computed posterior variance exceeded the prior variance.
This is due to poor performance of the ADF approximation for these
parameter values, in cases where no spikes occur and the true posterior
becomes bimodal.}
\label{harper}
\end{figure*}
We also test the robustness of the filter to inaccuracies in model
parameters or observation/encoding parameters in this problem. We
use inaccurate values for the model parameter $a$ in (\ref{eq:1d-diffusion})
and the observation/encoding parameter $h$ in (\ref{eq:gauss-tc-y})
in the filter. Specifically, we multiply or divide each of the two
parameters by a factor of $5/4$ in the filter, while the dynamics
and the observing neural population remain unaltered. The results
remain qualitatively similar, as seen in Figure (\ref{robustness}).

\begin{figure}
\includegraphics[width=1\textwidth]{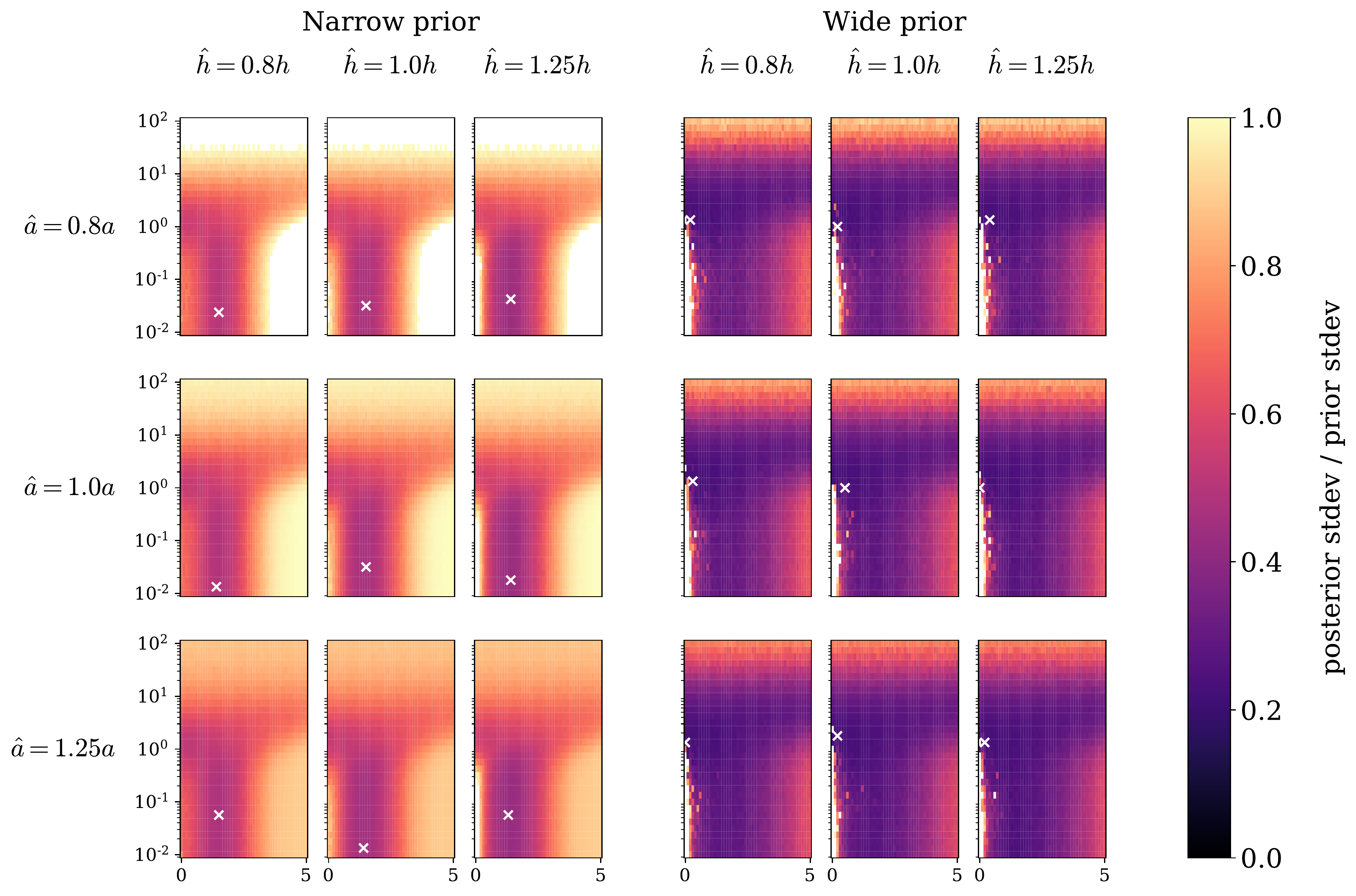}\caption{Robustness of the encoding result of figure (\ref{harper}) to decoding
with an inaccurate model of the dynamics or of the observation/encoding
process. As in Figure (\ref{harper}), color indicates the posterior
standard deviation derived from the ADF approximation, relative to
the prior standard deviation. The values of the dynamics parameter
$a$ and the population parameter $h$ used by the filter, respectively
denoted $\hat{a},\hat{h}$, are varied while the dynamics and observation
processes remain the same. Each row corresponds to a value of $\hat{a}$
and each column to a value of $\hat{h}$. White crosses denote the
values where the posterior standard deviation is minimized.}
\label{robustness}
\end{figure}

\subsection{Adaptation of homogeneous population to stimulus statistics}

In \citet{Benucci2013}, the tuning of cat primary visual cortex (V1)
neurons to oriented gratings was measured after adapting either to
a uniform distribution of orientations, or to a ``biased'' distribution
with one orientation being more common. The population's preferred
stimuli are roughly uniformly distributed, and adapt to the prior
statistics through change in amplitude. The adapted population was
observed to have decreased firing rate near the common orientation,
so that the mean firing rate is constant across the population.

We present a simplified model where optimal encoding exhibits a constant
mean firing rate across a neural population. A random state $X$ is
drawn from a ``biased'' prior distribution which is a mixture of
two uniform distributions on intervals sharing an endpoint (see Figure
\ref{benucci-setup}),
\begin{equation}
p\left(x\right)=\sum_{i=1}^{2}k_{i}\frac{1\left\{ a_{i}\leq x\leq b_{i}\right\} }{b_{i}-a_{i}},\quad k_{1}+k_{2}=1,\quad b_{1}=a_{2}.\label{eq:benucci-prior}
\end{equation}
The neural population consists of Gaussian neurons (\ref{eq:gauss-tc-y})
with preferred locations uniformly distribution on the interval $\left[a_{1},b_{2}\right]$.
However, they are allowed to adapt to different tuning amplitude $h_{1},h_{2}$
in each of the sub-intervals $\left[a_{1},b_{1}\right],\left[a_{2},b_{2}\right]$
respectively. Thus, the population distribution is given by a mixture
of two uniform components,
\begin{align}
f\left(d\boldsymbol{y}\right) & =f\left(dh,d\theta,dH',dR'\right)\nonumber \\
 & =\delta_{H}\left(dH'\right)\delta_{R}\left(dR'\right)\sum_{i}\delta_{h_{i}}\left(dh\right)f_{i}\left(d\theta\right),\nonumber \\
f_{i}\left(d\theta\right) & \triangleq1\left\{ a_{i}\leq\theta\leq b_{i}\right\} d\theta,\label{eq:benucci-pop}
\end{align}
where $H=1$. We optimize the parameters $h_{1},h_{2}$ to minimize
accumulated decoding MSE over a finite decoding interval under a constraint
on the \emph{total} firing rate of the population,
\begin{align}
\min_{h_{1},h_{2}}\quad & \E\left[\int_{0}^{T}\left(X-\hat{X}_{t}\right)^{2}dt\right]\nonumber \\
\mathrm{s.t.}\quad & r\triangleq\E\left[\int\lambda\left(X;y\right)f\left(d\boldsymbol{y}\right)\right]\leq\bar{r},\label{eq:opt}
\end{align}
where $\hat{X}_{t}=\E\left[X|\mathcal{N}_{t}\right]$ is the MMSE
estimate of $X$. The total firing rate $r$ may be obtained from
(\ref{eq:gauss-tc-y}), (\ref{eq:benucci-prior}), (\ref{eq:benucci-pop}),
yielding
\begin{align}
r & =r_{1}+r_{2},\nonumber \\
r_{i} & =\sum_{j=1}^{2}k_{j}\int_{a_{j}}^{b_{j}}\frac{dx}{b_{j}-a_{j}}\int_{a_{i}}^{b_{i}}d\theta\,h_{i}\exp\left(-\frac{\left(x-\theta\right)^{2}}{2\alpha^{2}}\right)\nonumber \\
 & =\sqrt{2\pi}\alpha^{2}\sum_{j=1}^{2}k_{j}\frac{h_{i}}{\left(b_{j}-a_{j}\right)}\nonumber \\
 & \quad\times\left[\Phi_{1}\left(\frac{a_{j}-b_{i}}{\alpha}\right)+\Phi_{1}\left(\frac{b_{j}-a_{i}}{\alpha}\right)-\Phi_{1}\left(\frac{a_{j}-a_{i}}{\alpha}\right)-\Phi_{1}\left(\frac{b_{j}-b_{i}}{\alpha}\right)\right],\label{eq:benucci-E-rate}
\end{align}
where $r_{i}$ is the total firing rate in the $i$th sub-population,
$\alpha^{2}=R$, $\Phi_{1}\left(x\right)\triangleq\int_{-\infty}^{x}\Phi=x\Phi\left(x\right)+\phi\left(x\right),$
and $\phi,\Phi$ are respectively the pdf and cdf of the standard
Gaussian distribution. In this continuous population approximation,
the firing rate of a \emph{single }neuron with preferred stimulus
$\theta$ is proportional to $\E\left[\lambda\left(X;\theta\right)\right]$.
We use narrow tuning ($\alpha=0.1$), so that the rate $\E\left[\lambda\left(X;\theta\right)\right]$
is nearly constant within each sub-population, justifying the approximation
\begin{equation}
\E\left[\lambda\left(X;\theta\right)\right]\approx\frac{1}{b_{i}-a_{i}}\int_{a_{i}}^{b_{i}}\E\left[\lambda\left(X;\theta\right)\right]d\theta=\frac{r_{i}}{b_{i}-a_{i}}\quad\left(a_{i}\leq\theta\leq b_{i}\right)\label{eq:benucci-single-rate}
\end{equation}

To solve this optimization problem, we approximate $\hat{X}_{T}$
by the ADF filter output $\mu_{T}$. We further assume that the MSE
is a decreasing function of the total firing rate $r$, so that the
solution is obtained for $r=\bar{r}$. Although we have no proof of
this claim, it appears reasonable, since the posterior variance decreases
at each spike (see \subsecref{Interpretation}). The stronger constraint
$r=\bar{r}$ leaves a single degree of freedom, the amplitude ratio
$h_{1}/h_{2}$. In Figure \ref{benucci-mse} we evaluate $\E\left[\int\left(X-\mu_{t}\right)^{2}dt\right]$
using Monte Carlo simulation for various values of the ratio $h_{1}/h_{2}$,
as well as location of the interval endpoint $b_{2}$. Although the
optimization is constrained only by the total firing rate, the minimal
MSE is obtained near the solution of
\[
\frac{r_{1}}{b_{1}-a_{1}}=\frac{r_{2}}{b_{2}-a_{2}},
\]
where firing rates are equalized across the population.

\begin{figure}
\subfloat[]{\centering{}\label{benucci-setup}\includegraphics[height=0.3\columnwidth]{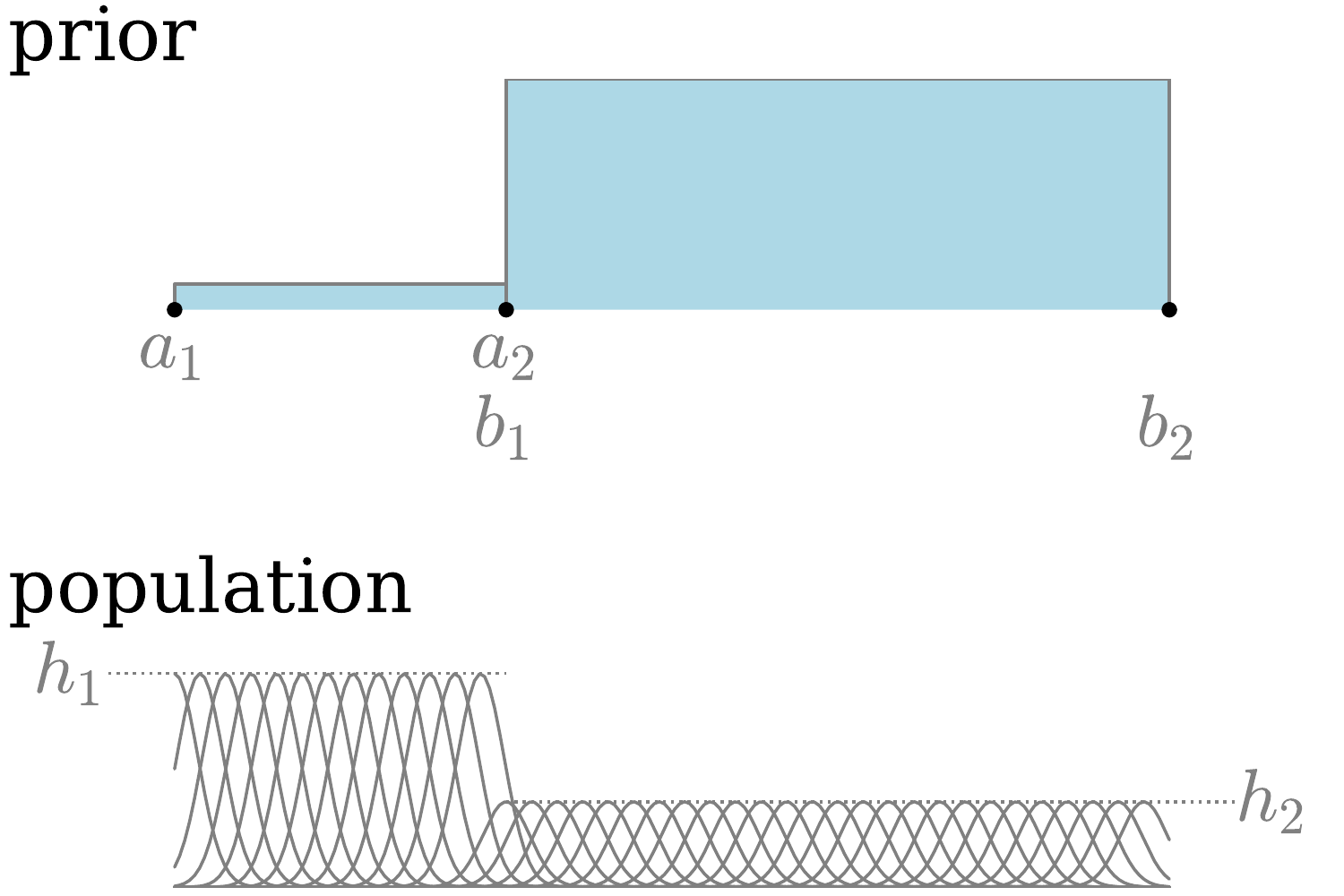}}\hfill{}\subfloat[]{\centering{}\label{benucci-mse}\includegraphics[height=0.3\columnwidth]{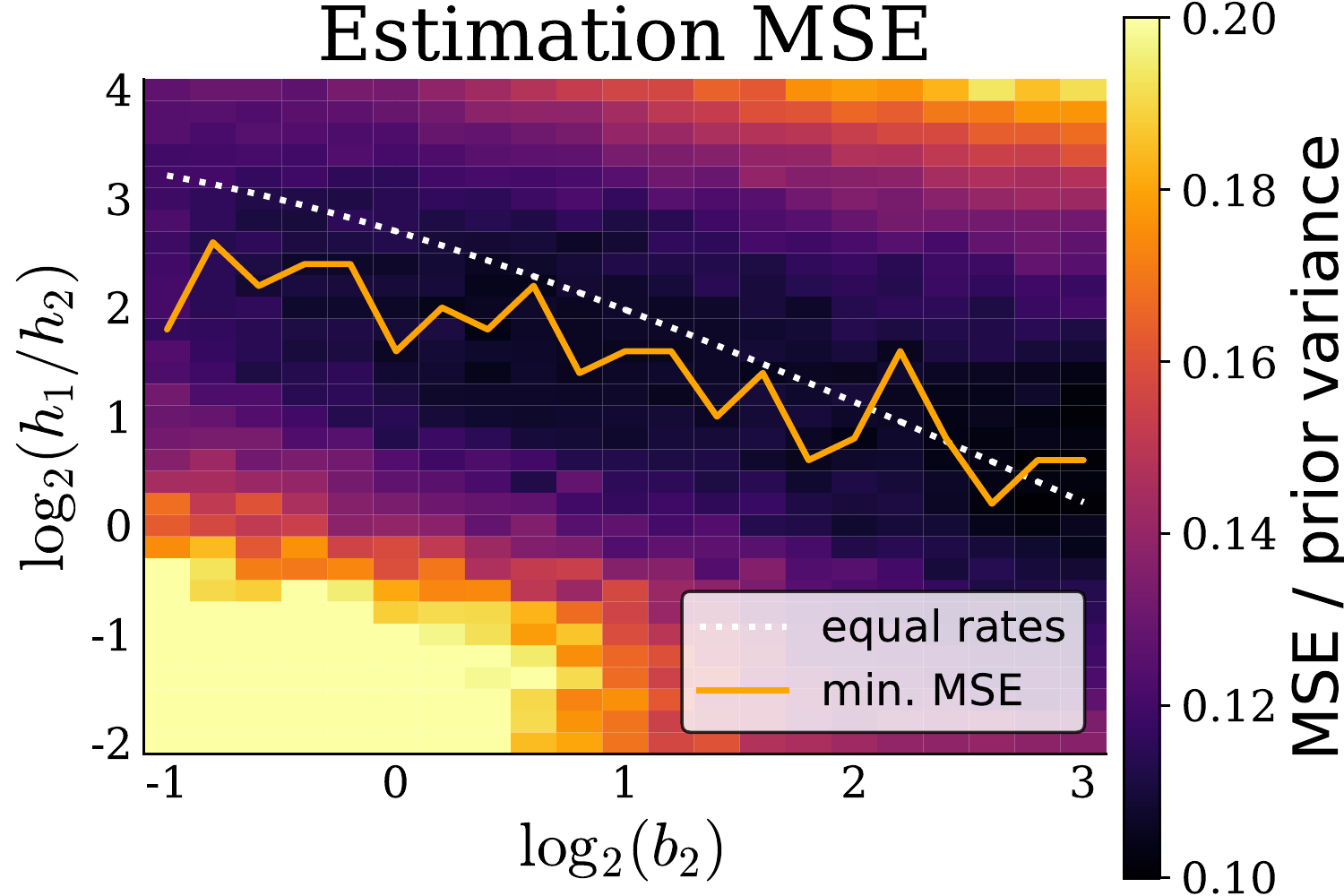}}

\caption{Optimal population in the mixture model equalizes firing rates of
different neurons. \textbf{(a) }prior and neural population \textbf{(b)}
estimation MSE as a function of height ratio $h_{1}/h_{2}$ and interval
endpoint $b_{2}$. Parameters are $k_{1}=0.1,k_{2}=0.9$, $a_{1}=-1,b_{1}=a_{2}=0,T=1,\bar{r}=10,\alpha^{2}=0.01,H=1$.
The dotted white lines indicates the ratio where firing rates are
equalized between the two sub-populations according to the approximation
\eqref{benucci-single-rate}, and the orange line the ratio where
the MSE is minimized. The MSE was computed by averaging across 10,000
trials for each data point.}
\end{figure}

\section{Comparison to Previous Work\label{sec:PreviousWorks}}

\subsection{Neural Decoding\label{subsec:Decoding-Previous}}

\textcolor{black}{Table \ref{lit-comparison} provides a concise comparison
of our setting and results to previous works on optimal neural decoding.
As noted in the introduction, we focus our comparison on analytically
expressible continuous-time filters. }

\begin{sidewaystable}
\caption{Summary of the setting and results of several previous works based
on continuous-time filtering of point process observations, in comparison
to the current work. The \emph{complexity }column lists the asymptotic
computational complexity of each time step in a discrete-time implementation
of the filter, as a function of population size $n$ and number of
spikes $N_{t}$. A complexity of $\infty$ denotes an infinite-dimensional
filter.}
\begin{threeparttable}
\begin{centering}
\begin{tabular}{ccccccccc}
\toprule 
\multirow{2}{*}{Ref.} & \multirow{2}{*}{Dynamics} & \multicolumn{2}{c}{Neural code} &  & \multicolumn{4}{c}{Decoding}\tabularnewline
\cmidrule{3-4} \cmidrule{6-9} 
 &  & rates & population &  &  & exact &  & complexity\tabularnewline
\midrule 
\citet{Snyder1972}\tnote{a} &  &  &  &  &  &  &  & \tabularnewline
\& \citet{Segall1975-filtering} & diffusion & any & finite\tnote{b} &  &  & \Checkmark{} &  & $\infty$\tabularnewline
\citet{Snyder1972}\tnote{a} & diffusion & any & finite\tnote{b} &  &  & - &  & $n$\tabularnewline
\citet{RhoSny1977} & lin. G. diff. & G. & uniform &  &  & \Checkmark{} &  & $1$\tabularnewline
\citet{Komaee2010} & lin. G. diff. & any & uniform &  &  & - &  & $1$\tabularnewline
\citet{Huys2007} & 1d G. process & G. & uniform &  &  & \Checkmark{} &  & $N_{t}^{3}$ \tnote{c}\tabularnewline
\citet{EdenBrown2008} & lin. G. diff. & any & finite &  &  & - &  & $n$\tabularnewline
\citet{SusMeiOpp11,SusMeiOpp13} & 1d G. Mat. & G. & uniform &  &  & - &  & $1$\tabularnewline
\citet{BobMeiEld09} &  &  &  &  &  &  &  & \tabularnewline
\& \citet{Twum-Danso2001} & CTMC & any & finite &  &  & \Checkmark{} &  & $n$\tabularnewline
This work & diffusion & G. & non-uniform &  &  & - &  & $1$ \tnote{d}\tabularnewline
\bottomrule
\end{tabular}
\par\end{centering}
\begin{tablenotes}
\item[a]{\cite{Snyder1972}} includes both an exact PDE for the posterior statistics, and an approximate solution.
\item[b]The setting of {\cite{Snyder1972}} and {\cite{Segall1975-filtering}} is a single observation point process, but the result is readily extended to a finite population by summation.
\item[c]This filter is non-recursive, and its complexity grows with the history of spikes. The exponent 3 is related to inversion of an $N_t \times N_t$ matrix, which in principle can be performed with lower complexity.
\item[d]Constant complexity for distributions over preferred stimuli described in section \ref{subsec:cont-filtering} (including the uniform coding case). For heterogenous mixture populations, the complexity is linear in the number of mixture components.
\end{tablenotes}
\end{threeparttable}

\medskip{}

abbreviations: \textbf{G.}	Gaussian, \ \textbf{lin.~G.~diff.}	linear
Gaussian diffusion, \ \textbf{1d~G.~Mat.} 1-dimensional Gaussian
process with Matern kernel auto-correlation, \ \textbf{CTMC} Continuous-time
Markov chain.

\label{lit-comparison}

\end{sidewaystable}
Few previous analytically oriented works studied neural decoding for
dynamic stimuli, or without the uniform coding assumption. The filtering
problem for a dynamic state with a general (non-uniform) finite population
is tractable when the the state space is finite; in this case the
posterior is finite-dimensional and obeys SDEs derived in \citet{Bremaud81}.
These results have been presented in a neuroscience context in \citet{Twum-Danso2001}
and \citet{BobMeiEld09}. 

We are aware of a single previous work \citet{EdenBrown2008} deriving
a closed-form filter for \emph{non-uniform }coding of a diffusion
process in \emph{continuous time}. The setting of \citet{EdenBrown2008}
is a finite populations with arbitrary tuning functions, and the derivation
uses an approximation similar to the one used in the current work.
Our work differs from \citet{EdenBrown2008} most notably by the use
of a parameterized ``infinite'' population. In this sense, the setting
of \citet{EdenBrown2008} is less general, corresponding to the case
of finite population described in section \ref{subsec:FinitePopulation}.
On the other hand, \citet{EdenBrown2008} deals with a more general
form for the tuning functions. Although \citet{EdenBrown2008} also
approximates the posterior using a Gaussian distribution, it is not
equivalent to our filter. The two filters are compared in detail in
appendix \ref{sec:EB-appendix}.

\subsection{Neural Encoding\label{subsec:PreviousWork-Encoding}}

As the current work is concerned with efficient neural decoding as
a tool for studying neural encoding in non-uniform populations, we
briefly mention some works that study neural encoding analytically.
As noted above, encoding will be studied in more detail in a subsequent
paper. Many previous works used Fisher information to study non-uniform
coding of static stimuli. As mentioned in the introduction, Fisher
information does not necessarily provide a good estimate for possible
decoding performance, and optimizing it may yield qualitatively misleading
results (see \citet{Bethge2002,YaeMei10,Pilarski2015}). \textcolor{black}{However,
we note one such work,} \citet{GanSim14}, which is particularly relevant
in comparison to the current work, as it similarly studies large parameterized
populations. The neural populations studied in \citet{GanSim14} are
characterized by two functions of the stimulus: a ``density'' function,
which described the local density of neurons; and a ``gain'' function,
which modulates each neuron's gain according to its preferred location.
The density function also distorts tuning functions so that tuning
width is approximately inversely proportional to neuron density, resulting
in approximately uniform coding when the gain function is constant.
On the other hand, the introduction of the gain function produces
a violation of uniform coding. This population is optimized for Fisher
information and several related measures. For each of these measures,
the optimal population density is shown to be monotonic with the prior,
i.e., more neurons should be assigned to more probable states. This
is in contrast to the results we present in section (\ref{sec:Encoding})
(e.g. Figure \ref{harper}, left), where the optimal population distribution
is shifted relative to the prior distribution when the prior is narrow.
This discrepancy may be attributed to the limitations of Fisher information
in predicting actual decoding error, or to the coupling of neuron
density and tuning width used in \citet{GanSim14} to facilitate the
derivation of closed-form solutions.

Several previous works attempted direct minimization of decoding MSE
rather than Fisher information. In \citet{YaeMei10}, an explicit
expression for the MSE is derived in the static case with uniform
coding, and is used to characterize optimal tuning function width
and its relation to coding time. More recently, \citet{WanStoLee2016}
studied $L_{p}$-based loss measures and presented exact results for
optimal tuning functions in the case of univariate static signals,
for single neurons and homogeneous populations. In \citet{SusMeiOpp11,SusMeiOpp13},
a mean-field approximation is suggested to allow efficient evaluation
of the MSE in a dynamics setting.

\section*{Acknowledgements}

The work of YH and RM is partially supported by grant No. 451/17 from
the Israel Science Foundation, and by the Ollendorff center of the
Viterbi Faculty of Electrical Engineering at the Technion.

\bibliographystyle{apacite}
\bibliography{nc_adf}

\begin{thebibliography}{}

\bibitem [\protect \citeauthoryear {%
Benucci%
, Saleem%
\BCBL {}\ \BBA {} Carandini%
}{%
Benucci%
\ \protect \BOthers {.}}{%
{\protect \APACyear {2013}}%
}]{%
Benucci2013}
\APACinsertmetastar {%
Benucci2013}%
\begin{APACrefauthors}%
Benucci, A.%
, Saleem, A\BPBI B.%
\BCBL {}\ \BBA {} Carandini, M.%
\end{APACrefauthors}%
\unskip\
\newblock
\APACrefYearMonthDay{2013}{}{}.
\newblock
{\BBOQ}\APACrefatitle {Adaptation maintains population homeostasis in primary
  visual cortex.} {Adaptation maintains population homeostasis in primary
  visual cortex.}{\BBCQ}
\newblock
\APACjournalVolNumPages{Nature neuroscience}{16}{6}{724--9}.
\newblock
\begin{APACrefDOI} \doi{10.1038/nn.3382} \end{APACrefDOI}
\PrintBackRefs{\CurrentBib}

\bibitem [\protect \citeauthoryear {%
Bethge%
, Rotermund%
\BCBL {}\ \BBA {} Pawelzik%
}{%
Bethge%
\ \protect \BOthers {.}}{%
{\protect \APACyear {2002}}%
}]{%
Bethge2002}
\APACinsertmetastar {%
Bethge2002}%
\begin{APACrefauthors}%
Bethge, M.%
, Rotermund, D.%
\BCBL {}\ \BBA {} Pawelzik, K.%
\end{APACrefauthors}%
\unskip\
\newblock
\APACrefYearMonthDay{2002}{Oct}{}.
\newblock
{\BBOQ}\APACrefatitle {Optimal short-term population coding: when Fisher
  information fails.} {Optimal short-term population coding: when fisher
  information fails.}{\BBCQ}
\newblock
\APACjournalVolNumPages{Neural Comput}{14}{10}{2317--2351}.
\newblock
\begin{APACrefDOI} \doi{10.1162/08997660260293247} \end{APACrefDOI}
\PrintBackRefs{\CurrentBib}

\bibitem [\protect \citeauthoryear {%
Bobrowski%
, Meir%
\BCBL {}\ \BBA {} Eldar%
}{%
Bobrowski%
\ \protect \BOthers {.}}{%
{\protect \APACyear {2009}}%
}]{%
BobMeiEld09}
\APACinsertmetastar {%
BobMeiEld09}%
\begin{APACrefauthors}%
Bobrowski, O.%
, Meir, R.%
\BCBL {}\ \BBA {} Eldar, Y.%
\end{APACrefauthors}%
\unskip\
\newblock
\APACrefYearMonthDay{2009}{May}{}.
\newblock
{\BBOQ}\APACrefatitle {Bayesian filtering in spiking neural networks: noise,
  adaptation, and multisensory integration.} {Bayesian filtering in spiking
  neural networks: noise, adaptation, and multisensory integration.}{\BBCQ}
\newblock
\APACjournalVolNumPages{Neural Comput}{21}{5}{1277--1320}.
\newblock
\begin{APACrefDOI} \doi{10.1162/neco.2008.01-08-692} \end{APACrefDOI}
\PrintBackRefs{\CurrentBib}

\bibitem [\protect \citeauthoryear {%
Brand%
, Behrend%
, Marquardt%
, McAlpine%
\BCBL {}\ \BBA {} Grothe%
}{%
Brand%
\ \protect \BOthers {.}}{%
{\protect \APACyear {2002}}%
}]{%
Brand2002}
\APACinsertmetastar {%
Brand2002}%
\begin{APACrefauthors}%
Brand, A.%
, Behrend, O.%
, Marquardt, T.%
, McAlpine, D.%
\BCBL {}\ \BBA {} Grothe, B.%
\end{APACrefauthors}%
\unskip\
\newblock
\APACrefYearMonthDay{2002}{}{}.
\newblock
{\BBOQ}\APACrefatitle {{Precise inhibition is essential for microsecond
  interaural time difference coding.}} {{Precise inhibition is essential for
  microsecond interaural time difference coding.}}{\BBCQ}
\newblock
\APACjournalVolNumPages{Nature}{417}{6888}{543--547}.
\newblock
\begin{APACrefDOI} \doi{10.1038/417543a} \end{APACrefDOI}
\PrintBackRefs{\CurrentBib}

\bibitem [\protect \citeauthoryear {%
Br\'emaud%
}{%
Br\'emaud%
}{%
{\protect \APACyear {1981}}%
}]{%
Bremaud81}
\APACinsertmetastar {%
Bremaud81}%
\begin{APACrefauthors}%
Br\'emaud, P.%
\end{APACrefauthors}%
\unskip\
\newblock
\APACrefYear{1981}.
\newblock
\APACrefbtitle {Point Processes and Queues: Martingale Dynamics} {Point
  processes and queues: Martingale dynamics}.
\newblock
\APACaddressPublisher{}{Springer, New York}.
\PrintBackRefs{\CurrentBib}

\bibitem [\protect \citeauthoryear {%
Brigo%
, Hanzon%
, Le~Gland%
\BCBL {}\ \protect \BOthers {.}}{%
Brigo%
\ \protect \BOthers {.}}{%
{\protect \APACyear {1999}}%
}]{%
BriHanLeg99}
\APACinsertmetastar {%
BriHanLeg99}%
\begin{APACrefauthors}%
Brigo, D.%
, Hanzon, B.%
, Le~Gland, F.%
\BCBL {}\ \BOthersPeriod {.}\end{APACrefauthors}%
\unskip\
\newblock
\APACrefYearMonthDay{1999}{}{}.
\newblock
{\BBOQ}\APACrefatitle {Approximate nonlinear filtering by projection on
  exponential manifolds of densities} {Approximate nonlinear filtering by
  projection on exponential manifolds of densities}.{\BBCQ}
\newblock
\APACjournalVolNumPages{Bernoulli}{5}{3}{495--534}.
\PrintBackRefs{\CurrentBib}

\bibitem [\protect \citeauthoryear {%
Brown%
, Barbieri%
, Ventura%
, Kass%
\BCBL {}\ \BBA {} Frank%
}{%
Brown%
\ \protect \BOthers {.}}{%
{\protect \APACyear {2002}}%
}]{%
brown2002time}
\APACinsertmetastar {%
brown2002time}%
\begin{APACrefauthors}%
Brown, E\BPBI N.%
, Barbieri, R.%
, Ventura, V.%
, Kass, R\BPBI E.%
\BCBL {}\ \BBA {} Frank, L\BPBI M.%
\end{APACrefauthors}%
\unskip\
\newblock
\APACrefYearMonthDay{2002}{}{}.
\newblock
{\BBOQ}\APACrefatitle {The time-rescaling theorem and its application to neural
  spike train data analysis} {The time-rescaling theorem and its application to
  neural spike train data analysis}.{\BBCQ}
\newblock
\APACjournalVolNumPages{Neural computation}{14}{2}{325--346}.
\PrintBackRefs{\CurrentBib}

\bibitem [\protect \citeauthoryear {%
Chelaru%
\ \BBA {} Dragoi%
}{%
Chelaru%
\ \BBA {} Dragoi%
}{%
{\protect \APACyear {2008}}%
}]{%
CheDra08}
\APACinsertmetastar {%
CheDra08}%
\begin{APACrefauthors}%
Chelaru, M.%
\BCBT {}\ \BBA {} Dragoi, V.%
\end{APACrefauthors}%
\unskip\
\newblock
\APACrefYearMonthDay{2008}{}{}.
\newblock
{\BBOQ}\APACrefatitle {Efficient coding in heterogeneous neuronal populations}
  {Efficient coding in heterogeneous neuronal populations}.{\BBCQ}
\newblock
\APACjournalVolNumPages{Proceedings of the National Academy of
  Sciences}{105}{42}{16344--16349}.
\PrintBackRefs{\CurrentBib}

\bibitem [\protect \citeauthoryear {%
Dayan%
\ \BBA {} Abbott%
}{%
Dayan%
\ \BBA {} Abbott%
}{%
{\protect \APACyear {2005}}%
}]{%
DayAbb05}
\APACinsertmetastar {%
DayAbb05}%
\begin{APACrefauthors}%
Dayan, P.%
\BCBT {}\ \BBA {} Abbott, L.%
\end{APACrefauthors}%
\unskip\
\newblock
\APACrefYear{2005}.
\newblock
\APACrefbtitle {Theoretical Neuroscience: Computational and Mathematical
  Modeling of Neural Systems} {Theoretical neuroscience: Computational and
  mathematical modeling of neural systems}.
\newblock
\APACaddressPublisher{}{MIT Press}.
\PrintBackRefs{\CurrentBib}

\bibitem [\protect \citeauthoryear {%
Doucet%
\ \BBA {} Johansen%
}{%
Doucet%
\ \BBA {} Johansen%
}{%
{\protect \APACyear {2009}}%
}]{%
DouJoh09}
\APACinsertmetastar {%
DouJoh09}%
\begin{APACrefauthors}%
Doucet, A.%
\BCBT {}\ \BBA {} Johansen, A.%
\end{APACrefauthors}%
\unskip\
\newblock
\APACrefYearMonthDay{2009}{}{}.
\newblock
{\BBOQ}\APACrefatitle {A tutorial on particle filtering and smoothing: fifteen
  years later} {A tutorial on particle filtering and smoothing: fifteen years
  later}.{\BBCQ}
\newblock
\BIn{} D.~Crisan\ \BBA {} B.~Rozovskii\ (\BEDS), \APACrefbtitle {Handbook of
  Nonlinear Filtering} {Handbook of nonlinear filtering}\ (\BPGS\ 656--704).
\newblock
\APACaddressPublisher{}{Oxford, UK: Oxford University Press}.
\PrintBackRefs{\CurrentBib}

\bibitem [\protect \citeauthoryear {%
Ecker%
, Berens%
, Tolias%
\BCBL {}\ \BBA {} Bethge%
}{%
Ecker%
\ \protect \BOthers {.}}{%
{\protect \APACyear {2011}}%
}]{%
EckBerTol11}
\APACinsertmetastar {%
EckBerTol11}%
\begin{APACrefauthors}%
Ecker, A.%
, Berens, P.%
, Tolias, A.%
\BCBL {}\ \BBA {} Bethge, M.%
\end{APACrefauthors}%
\unskip\
\newblock
\APACrefYearMonthDay{2011}{}{}.
\newblock
{\BBOQ}\APACrefatitle {The effect of noise correlations in populations of
  diversely tuned neurons} {The effect of noise correlations in populations of
  diversely tuned neurons}.{\BBCQ}
\newblock
\APACjournalVolNumPages{Journal of Neuroscience}{31}{40}{14272--14283}.
\PrintBackRefs{\CurrentBib}

\bibitem [\protect \citeauthoryear {%
Eden%
\ \BBA {} Brown%
}{%
Eden%
\ \BBA {} Brown%
}{%
{\protect \APACyear {2008}}%
}]{%
EdenBrown2008}
\APACinsertmetastar {%
EdenBrown2008}%
\begin{APACrefauthors}%
Eden, U\BPBI T.%
\BCBT {}\ \BBA {} Brown, E\BPBI N.%
\end{APACrefauthors}%
\unskip\
\newblock
\APACrefYearMonthDay{2008}{}{}.
\newblock
{\BBOQ}\APACrefatitle {{Continuous-Time Filters for State Estimation From Point
  Process Models of Neural Data.}} {{Continuous-Time Filters for State
  Estimation From Point Process Models of Neural Data.}}{\BBCQ}
\newblock
\APACjournalVolNumPages{Statistica Sinica}{18}{4}{1293--1310}.
\newblock
\begin{APACrefURL}
  \url{http://www.pubmedcentral.nih.gov/articlerender.fcgi?artid=3208353{\&}tool=pmcentrez{\&}rendertype=abstract}
  \end{APACrefURL}
\newblock
\begin{APACrefDOI} \doi{10.1016/j.biotechadv.2011.08.021.Secreted}
  \end{APACrefDOI}
\PrintBackRefs{\CurrentBib}

\bibitem [\protect \citeauthoryear {%
Eden%
, Frank%
, Barbieri%
, Solo%
\BCBL {}\ \BBA {} Brown%
}{%
Eden%
\ \protect \BOthers {.}}{%
{\protect \APACyear {2004}}%
}]{%
Eden2004}
\APACinsertmetastar {%
Eden2004}%
\begin{APACrefauthors}%
Eden, U\BPBI T.%
, Frank, L\BPBI M.%
, Barbieri, R.%
, Solo, V.%
\BCBL {}\ \BBA {} Brown, E\BPBI N.%
\end{APACrefauthors}%
\unskip\
\newblock
\APACrefYearMonthDay{2004}{}{}.
\newblock
{\BBOQ}\APACrefatitle {Dynamic Analysis of Neural Encoding by Point Process
  Adaptive Filtering} {Dynamic analysis of neural encoding by point process
  adaptive filtering}.{\BBCQ}
\newblock
\APACjournalVolNumPages{Neural Computation}{}{}{}.
\PrintBackRefs{\CurrentBib}

\bibitem [\protect \citeauthoryear {%
Efron%
\ \BBA {} Tibshirani%
}{%
Efron%
\ \BBA {} Tibshirani%
}{%
{\protect \APACyear {1994}}%
}]{%
Efron1994introduction}
\APACinsertmetastar {%
Efron1994introduction}%
\begin{APACrefauthors}%
Efron, B.%
\BCBT {}\ \BBA {} Tibshirani, R\BPBI J.%
\end{APACrefauthors}%
\unskip\
\newblock
\APACrefYear{1994}.
\newblock
\APACrefbtitle {An introduction to the bootstrap} {An introduction to the
  bootstrap}.
\newblock
\APACaddressPublisher{}{CRC press}.
\PrintBackRefs{\CurrentBib}

\bibitem [\protect \citeauthoryear {%
Frey%
\ \BBA {} Runggaldier%
}{%
Frey%
\ \BBA {} Runggaldier%
}{%
{\protect \APACyear {2001}}%
}]{%
Frey2001}
\APACinsertmetastar {%
Frey2001}%
\begin{APACrefauthors}%
Frey, R.%
\BCBT {}\ \BBA {} Runggaldier, W\BPBI J.%
\end{APACrefauthors}%
\unskip\
\newblock
\APACrefYearMonthDay{2001}{}{}.
\newblock
{\BBOQ}\APACrefatitle {A nonlinear filtering approach to volatility estimation
  with a view towards high frequency data} {A nonlinear filtering approach to
  volatility estimation with a view towards high frequency data}.{\BBCQ}
\newblock
\APACjournalVolNumPages{International Journal of Theoretical and Applied
  Finance}{4}{02}{199--210}.
\PrintBackRefs{\CurrentBib}

\bibitem [\protect \citeauthoryear {%
Ganguli%
\ \BBA {} Simoncelli%
}{%
Ganguli%
\ \BBA {} Simoncelli%
}{%
{\protect \APACyear {2014}}%
}]{%
GanSim14}
\APACinsertmetastar {%
GanSim14}%
\begin{APACrefauthors}%
Ganguli, D.%
\BCBT {}\ \BBA {} Simoncelli, E.%
\end{APACrefauthors}%
\unskip\
\newblock
\APACrefYearMonthDay{2014}{}{}.
\newblock
{\BBOQ}\APACrefatitle {Efficient sensory encoding and Bayesian inference with
  heterogeneous neural populations.} {Efficient sensory encoding and bayesian
  inference with heterogeneous neural populations.}{\BBCQ}
\newblock
\APACjournalVolNumPages{Neural Comput}{26}{10}{2103--2134}.
\newblock
\begin{APACrefDOI} \doi{10.1162/NECO_a_00638} \end{APACrefDOI}
\PrintBackRefs{\CurrentBib}

\bibitem [\protect \citeauthoryear {%
Gill%
\ \BBA {} Levit%
}{%
Gill%
\ \BBA {} Levit%
}{%
{\protect \APACyear {1995}}%
}]{%
GillLevit95}
\APACinsertmetastar {%
GillLevit95}%
\begin{APACrefauthors}%
Gill, R\BPBI D.%
\BCBT {}\ \BBA {} Levit, B\BPBI Y.%
\end{APACrefauthors}%
\unskip\
\newblock
\APACrefYearMonthDay{1995}{}{}.
\newblock
{\BBOQ}\APACrefatitle {Applications of the van Trees inequality: a Bayesian
  Cram{\'e}r-Rao bound} {Applications of the van trees inequality: a bayesian
  cram{\'e}r-rao bound}.{\BBCQ}
\newblock
\APACjournalVolNumPages{Bernoulli}{}{}{59--79}.
\PrintBackRefs{\CurrentBib}

\bibitem [\protect \citeauthoryear {%
Harel%
, Meir%
\BCBL {}\ \BBA {} Opper%
}{%
Harel%
\ \protect \BOthers {.}}{%
{\protect \APACyear {2015}}%
}]{%
Harel2015}
\APACinsertmetastar {%
Harel2015}%
\begin{APACrefauthors}%
Harel, Y.%
, Meir, R.%
\BCBL {}\ \BBA {} Opper, M.%
\end{APACrefauthors}%
\unskip\
\newblock
\APACrefYearMonthDay{2015}{}{}.
\newblock
{\BBOQ}\APACrefatitle {A Tractable Approximation to Optimal Point Process
  Filtering: Application to Neural Encoding} {A tractable approximation to
  optimal point process filtering: Application to neural encoding}.{\BBCQ}
\newblock
\BIn{} \APACrefbtitle {Advances in Neural Information Processing Systems}
  {Advances in neural information processing systems}\ (\BPGS\ 1603--1611).
\PrintBackRefs{\CurrentBib}

\bibitem [\protect \citeauthoryear {%
Harper%
\ \BBA {} McAlpine%
}{%
Harper%
\ \BBA {} McAlpine%
}{%
{\protect \APACyear {2004}}%
}]{%
HarMcAlp04}
\APACinsertmetastar {%
HarMcAlp04}%
\begin{APACrefauthors}%
Harper, N.%
\BCBT {}\ \BBA {} McAlpine, D.%
\end{APACrefauthors}%
\unskip\
\newblock
\APACrefYearMonthDay{2004}{Aug}{}.
\newblock
{\BBOQ}\APACrefatitle {Optimal neural population coding of an auditory spatial
  cue.} {Optimal neural population coding of an auditory spatial cue.}{\BBCQ}
\newblock
\APACjournalVolNumPages{Nature}{430}{7000}{682--686}.
\newblock
\APACrefnote{n1397b}
\newblock
\begin{APACrefDOI} \doi{10.1038/nature02768} \end{APACrefDOI}
\PrintBackRefs{\CurrentBib}

\bibitem [\protect \citeauthoryear {%
Henderson%
\ \BBA {} Searle%
}{%
Henderson%
\ \BBA {} Searle%
}{%
{\protect \APACyear {1981}}%
}]{%
Henderson1981}
\APACinsertmetastar {%
Henderson1981}%
\begin{APACrefauthors}%
Henderson, H.%
\BCBT {}\ \BBA {} Searle, S.%
\end{APACrefauthors}%
\unskip\
\newblock
\APACrefYearMonthDay{1981}{}{}.
\newblock
{\BBOQ}\APACrefatitle {On deriving the inverse of a sum of matrices} {On
  deriving the inverse of a sum of matrices}.{\BBCQ}
\newblock
\APACjournalVolNumPages{Siam Review}{23}{1}{53--60}.
\PrintBackRefs{\CurrentBib}

\bibitem [\protect \citeauthoryear {%
Huys%
, Zemel%
, Natarajan%
\BCBL {}\ \BBA {} Dayan%
}{%
Huys%
\ \protect \BOthers {.}}{%
{\protect \APACyear {2007}}%
}]{%
Huys2007}
\APACinsertmetastar {%
Huys2007}%
\begin{APACrefauthors}%
Huys, Q\BPBI J\BPBI M.%
, Zemel, R\BPBI S.%
, Natarajan, R.%
\BCBL {}\ \BBA {} Dayan, P.%
\end{APACrefauthors}%
\unskip\
\newblock
\APACrefYearMonthDay{2007}{feb}{}.
\newblock
{\BBOQ}\APACrefatitle {{Fast population coding.}} {{Fast population
  coding.}}{\BBCQ}
\newblock
\APACjournalVolNumPages{Neural computation}{19}{2}{404--41}.
\newblock
\begin{APACrefDOI} \doi{10.1162/neco.2007.19.2.404} \end{APACrefDOI}
\PrintBackRefs{\CurrentBib}

\bibitem [\protect \citeauthoryear {%
Komaee%
}{%
Komaee%
}{%
{\protect \APACyear {2010}}%
}]{%
Komaee2010}
\APACinsertmetastar {%
Komaee2010}%
\begin{APACrefauthors}%
Komaee, A.%
\end{APACrefauthors}%
\unskip\
\newblock
\APACrefYearMonthDay{2010}{}{}.
\newblock
{\BBOQ}\APACrefatitle {{State estimation from space-time point process
  observations with an application in optical beam tracking}} {{State
  estimation from space-time point process observations with an application in
  optical beam tracking}}.{\BBCQ}
\newblock
\APACjournalVolNumPages{2010 44th Annual Conference on Information Sciences and
  Systems, CISS 2010}{}{}{}.
\newblock
\begin{APACrefDOI} \doi{10.1109/CISS.2010.5464949} \end{APACrefDOI}
\PrintBackRefs{\CurrentBib}

\bibitem [\protect \citeauthoryear {%
Lilliefors%
}{%
Lilliefors%
}{%
{\protect \APACyear {1967}}%
}]{%
Lilliefors1967}
\APACinsertmetastar {%
Lilliefors1967}%
\begin{APACrefauthors}%
Lilliefors, H\BPBI W.%
\end{APACrefauthors}%
\unskip\
\newblock
\APACrefYearMonthDay{1967}{}{}.
\newblock
{\BBOQ}\APACrefatitle {On the Kolmogorov-Smirnov test for normality with mean
  and variance unknown} {On the kolmogorov-smirnov test for normality with mean
  and variance unknown}.{\BBCQ}
\newblock
\APACjournalVolNumPages{Journal of the American statistical
  Association}{62}{318}{399--402}.
\PrintBackRefs{\CurrentBib}

\bibitem [\protect \citeauthoryear {%
Macke%
, Buesing%
\BCBL {}\ \BBA {} Sahani%
}{%
Macke%
\ \protect \BOthers {.}}{%
{\protect \APACyear {2015}}%
}]{%
Macke2015}
\APACinsertmetastar {%
Macke2015}%
\begin{APACrefauthors}%
Macke, J\BPBI H.%
, Buesing, L.%
\BCBL {}\ \BBA {} Sahani, M.%
\end{APACrefauthors}%
\unskip\
\newblock
\APACrefYearMonthDay{2015}{}{}.
\newblock
{\BBOQ}\APACrefatitle {Estimating State and Parameters in State Space Models of
  Spike Trains} {Estimating state and parameters in state space models of spike
  trains}.{\BBCQ}
\newblock
\APACjournalVolNumPages{Advanced State Space Methods for Neural and Clinical
  Data}{}{}{137}.
\PrintBackRefs{\CurrentBib}

\bibitem [\protect \citeauthoryear {%
Maybeck%
}{%
Maybeck%
}{%
{\protect \APACyear {1979}}%
}]{%
Maybeck79}
\APACinsertmetastar {%
Maybeck79}%
\begin{APACrefauthors}%
Maybeck, P.%
\end{APACrefauthors}%
\unskip\
\newblock
\APACrefYear{1979}.
\newblock
\APACrefbtitle {Stochastic Models, Estimation, and Control} {Stochastic models,
  estimation, and control}.
\newblock
\APACaddressPublisher{}{Academic Press}.
\PrintBackRefs{\CurrentBib}

\bibitem [\protect \citeauthoryear {%
Minka%
}{%
Minka%
}{%
{\protect \APACyear {2001}}%
}]{%
Minka01}
\APACinsertmetastar {%
Minka01}%
\begin{APACrefauthors}%
Minka, T.%
\end{APACrefauthors}%
\unskip\
\newblock
\APACrefYearMonthDay{2001}{}{}.
\newblock
{\BBOQ}\APACrefatitle {Expectation propagation for approximate Bayesian
  inference} {Expectation propagation for approximate bayesian
  inference}.{\BBCQ}
\newblock
\BIn{} \APACrefbtitle {Proceedings of the Seventeenth conference on Uncertainty
  in artificial intelligence} {Proceedings of the seventeenth conference on
  uncertainty in artificial intelligence}\ (\BPGS\ 362--369).
\PrintBackRefs{\CurrentBib}

\bibitem [\protect \citeauthoryear {%
{\O}ksendal%
}{%
{\O}ksendal%
}{%
{\protect \APACyear {2003}}%
}]{%
Oksendal2003}
\APACinsertmetastar {%
Oksendal2003}%
\begin{APACrefauthors}%
{\O}ksendal, B.%
\end{APACrefauthors}%
\unskip\
\newblock
\APACrefYear{2003}.
\newblock
\APACrefbtitle {Stochastic {Differential} {Equations}} {Stochastic
  {Differential} {Equations}}.
\newblock
\APACaddressPublisher{}{Springer}.
\PrintBackRefs{\CurrentBib}

\bibitem [\protect \citeauthoryear {%
Opper%
}{%
Opper%
}{%
{\protect \APACyear {1998}}%
}]{%
Opper98}
\APACinsertmetastar {%
Opper98}%
\begin{APACrefauthors}%
Opper, M.%
\end{APACrefauthors}%
\unskip\
\newblock
\APACrefYearMonthDay{1998}{}{}.
\newblock
{\BBOQ}\APACrefatitle {A {B}ayesian Approach to Online Learning} {A {B}ayesian
  approach to online learning}.{\BBCQ}
\newblock
\BIn{} D.~Saad\ (\BED), \APACrefbtitle {Online Learning in Neural Networks}
  {Online learning in neural networks}\ (\BPGS\ 363--378).
\newblock
\APACaddressPublisher{}{Cambridge university press}.
\PrintBackRefs{\CurrentBib}

\bibitem [\protect \citeauthoryear {%
Pilarski%
\ \BBA {} Pokora%
}{%
Pilarski%
\ \BBA {} Pokora%
}{%
{\protect \APACyear {2015}}%
}]{%
Pilarski2015}
\APACinsertmetastar {%
Pilarski2015}%
\begin{APACrefauthors}%
Pilarski, S.%
\BCBT {}\ \BBA {} Pokora, O.%
\end{APACrefauthors}%
\unskip\
\newblock
\APACrefYearMonthDay{2015}{}{}.
\newblock
{\BBOQ}\APACrefatitle {{On the Cram{\'{e}}r-Rao bound applicability and the
  role of Fisher information in computational neuroscience}} {{On the
  Cram{\'{e}}r-Rao bound applicability and the role of Fisher information in
  computational neuroscience}}.{\BBCQ}
\newblock
\APACjournalVolNumPages{BioSystems}{136}{}{11--22}.
\newblock
\begin{APACrefURL} \url{http://dx.doi.org/10.1016/j.biosystems.2015.07.009}
  \end{APACrefURL}
\newblock
\begin{APACrefDOI} \doi{10.1016/j.biosystems.2015.07.009} \end{APACrefDOI}
\PrintBackRefs{\CurrentBib}

\bibitem [\protect \citeauthoryear {%
Pillow%
\ \protect \BOthers {.}}{%
Pillow%
\ \protect \BOthers {.}}{%
{\protect \APACyear {2008}}%
}]{%
Pillow2008}
\APACinsertmetastar {%
Pillow2008}%
\begin{APACrefauthors}%
Pillow, J\BPBI W.%
, Shlens, J.%
, Paninski, L.%
, Sher, A.%
, Litke, A\BPBI M.%
, Chichilnisky, E.%
\BCBL {}\ \BBA {} Simoncelli, E\BPBI P.%
\end{APACrefauthors}%
\unskip\
\newblock
\APACrefYearMonthDay{2008}{}{}.
\newblock
{\BBOQ}\APACrefatitle {Spatio-temporal correlations and visual signalling in a
  complete neuronal population} {Spatio-temporal correlations and visual
  signalling in a complete neuronal population}.{\BBCQ}
\newblock
\APACjournalVolNumPages{Nature}{454}{7207}{995--999}.
\PrintBackRefs{\CurrentBib}

\bibitem [\protect \citeauthoryear {%
Radhakrishna~Rao%
}{%
Radhakrishna~Rao%
}{%
{\protect \APACyear {1945}}%
}]{%
Rao1945}
\APACinsertmetastar {%
Rao1945}%
\begin{APACrefauthors}%
Radhakrishna~Rao, C.%
\end{APACrefauthors}%
\unskip\
\newblock
\APACrefYearMonthDay{1945}{}{}.
\newblock
{\BBOQ}\APACrefatitle {Information and the accuracy attainable in the
  estimation of statistical parameters} {Information and the accuracy
  attainable in the estimation of statistical parameters}.{\BBCQ}
\newblock
\APACjournalVolNumPages{Bull. Calcutta Math. Soc.}{37}{}{81--91}.
\PrintBackRefs{\CurrentBib}

\bibitem [\protect \citeauthoryear {%
Rhodes%
\ \BBA {} Snyder%
}{%
Rhodes%
\ \BBA {} Snyder%
}{%
{\protect \APACyear {1977}}%
}]{%
RhoSny1977}
\APACinsertmetastar {%
RhoSny1977}%
\begin{APACrefauthors}%
Rhodes, I.%
\BCBT {}\ \BBA {} Snyder, D.%
\end{APACrefauthors}%
\unskip\
\newblock
\APACrefYearMonthDay{1977}{}{}.
\newblock
{\BBOQ}\APACrefatitle {{Estimation and control performance for space-time
  point-process observations}} {{Estimation and control performance for
  space-time point-process observations}}.{\BBCQ}
\newblock
\APACjournalVolNumPages{IEEE Transactions on Automatic
  Control}{22}{3}{338--346}.
\newblock
\begin{APACrefDOI} \doi{10.1109/TAC.1977.1101534} \end{APACrefDOI}
\PrintBackRefs{\CurrentBib}

\bibitem [\protect \citeauthoryear {%
Segall%
}{%
Segall%
}{%
{\protect \APACyear {1976}}%
}]{%
Segall1976}
\APACinsertmetastar {%
Segall1976}%
\begin{APACrefauthors}%
Segall, A.%
\end{APACrefauthors}%
\unskip\
\newblock
\APACrefYearMonthDay{1976}{}{}.
\newblock
{\BBOQ}\APACrefatitle {Recursive estimation from discrete-time point processes}
  {Recursive estimation from discrete-time point processes}.{\BBCQ}
\newblock
\APACjournalVolNumPages{Information Theory, IEEE Transactions
  on}{22}{4}{422--431}.
\PrintBackRefs{\CurrentBib}

\bibitem [\protect \citeauthoryear {%
Segall%
}{%
Segall%
}{%
{\protect \APACyear {1978}}%
}]{%
Segall1978}
\APACinsertmetastar {%
Segall1978}%
\begin{APACrefauthors}%
Segall, A.%
\end{APACrefauthors}%
\unskip\
\newblock
\APACrefYearMonthDay{1978}{}{}.
\newblock
{\BBOQ}\APACrefatitle {Centralized and decentralized control schemes for
  {Gauss}-{Poisson} processes} {Centralized and decentralized control schemes
  for {Gauss}-{Poisson} processes}.{\BBCQ}
\newblock
\APACjournalVolNumPages{IEEE Transactions on Automatic Control}{23}{1}{47--57}.
\PrintBackRefs{\CurrentBib}

\bibitem [\protect \citeauthoryear {%
Segall%
, Davis%
\BCBL {}\ \BBA {} Kailath%
}{%
Segall%
\ \protect \BOthers {.}}{%
{\protect \APACyear {1975}}%
}]{%
Segall1975-filtering}
\APACinsertmetastar {%
Segall1975-filtering}%
\begin{APACrefauthors}%
Segall, A.%
, Davis, M\BPBI H.%
\BCBL {}\ \BBA {} Kailath, T.%
\end{APACrefauthors}%
\unskip\
\newblock
\APACrefYearMonthDay{1975}{}{}.
\newblock
{\BBOQ}\APACrefatitle {Nonlinear filtering with counting observations}
  {Nonlinear filtering with counting observations}.{\BBCQ}
\newblock
\APACjournalVolNumPages{Information Theory, IEEE Transactions
  on}{21}{2}{143--149}.
\PrintBackRefs{\CurrentBib}

\bibitem [\protect \citeauthoryear {%
Segall%
\ \BBA {} Kailath%
}{%
Segall%
\ \BBA {} Kailath%
}{%
{\protect \APACyear {1975}}%
}]{%
Segall1975-modelling}
\APACinsertmetastar {%
Segall1975-modelling}%
\begin{APACrefauthors}%
Segall, A.%
\BCBT {}\ \BBA {} Kailath, T.%
\end{APACrefauthors}%
\unskip\
\newblock
\APACrefYearMonthDay{1975}{}{}.
\newblock
{\BBOQ}\APACrefatitle {The modeling of randomly modulated jump processes} {The
  modeling of randomly modulated jump processes}.{\BBCQ}
\newblock
\APACjournalVolNumPages{Information Theory, IEEE Transactions
  on}{21}{2}{135--143}.
\PrintBackRefs{\CurrentBib}

\bibitem [\protect \citeauthoryear {%
Snyder%
}{%
Snyder%
}{%
{\protect \APACyear {1972}}%
}]{%
Snyder1972}
\APACinsertmetastar {%
Snyder1972}%
\begin{APACrefauthors}%
Snyder, D.%
\end{APACrefauthors}%
\unskip\
\newblock
\APACrefYearMonthDay{1972}{{\APACmonth{01}}}{}.
\newblock
{\BBOQ}\APACrefatitle {{Filtering and detection for doubly stochastic Poisson
  processes}} {{Filtering and detection for doubly stochastic Poisson
  processes}}.{\BBCQ}
\newblock
\APACjournalVolNumPages{IEEE Transactions on Information
  Theory}{18}{1}{91--102}.
\newblock
\begin{APACrefDOI} \doi{10.1109/TIT.1972.1054756} \end{APACrefDOI}
\PrintBackRefs{\CurrentBib}

\bibitem [\protect \citeauthoryear {%
Snyder%
\ \BBA {} Miller%
}{%
Snyder%
\ \BBA {} Miller%
}{%
{\protect \APACyear {1991}}%
}]{%
SnyMil91}
\APACinsertmetastar {%
SnyMil91}%
\begin{APACrefauthors}%
Snyder, D.%
\BCBT {}\ \BBA {} Miller, M.%
\end{APACrefauthors}%
\unskip\
\newblock
\APACrefYear{1991}.
\newblock
\APACrefbtitle {Random Point Processes in Time and Space} {Random point
  processes in time and space}\ (\PrintOrdinal{second edition}\ \BEd).
\newblock
\APACaddressPublisher{}{Springer}.
\PrintBackRefs{\CurrentBib}

\bibitem [\protect \citeauthoryear {%
Snyder%
, Rhodes%
\BCBL {}\ \BBA {} Hoversten%
}{%
Snyder%
\ \protect \BOthers {.}}{%
{\protect \APACyear {1977}}%
}]{%
Snyder1977}
\APACinsertmetastar {%
Snyder1977}%
\begin{APACrefauthors}%
Snyder, D.%
, Rhodes, I.%
\BCBL {}\ \BBA {} Hoversten, E.%
\end{APACrefauthors}%
\unskip\
\newblock
\APACrefYearMonthDay{1977}{}{}.
\newblock
{\BBOQ}\APACrefatitle {A separation theorem for stochastic control problems
  with point-process observations} {A separation theorem for stochastic control
  problems with point-process observations}.{\BBCQ}
\newblock
\APACjournalVolNumPages{Automatica}{13}{1}{85--87}.
\PrintBackRefs{\CurrentBib}

\bibitem [\protect \citeauthoryear {%
Solo%
}{%
Solo%
}{%
{\protect \APACyear {2000}}%
}]{%
Solo2000}
\APACinsertmetastar {%
Solo2000}%
\begin{APACrefauthors}%
Solo, V.%
\end{APACrefauthors}%
\unskip\
\newblock
\APACrefYearMonthDay{2000}{}{}.
\newblock
{\BBOQ}\APACrefatitle {{'Unobserved' Monte Carlo method for identification of
  partially observed nonlinear state space systems , Part II: Counting Process
  Observations.}} {{'Unobserved' Monte Carlo method for identification of
  partially observed nonlinear state space systems , Part II: Counting Process
  Observations.}}{\BBCQ}
\newblock
\BIn{} \APACrefbtitle {Proceedings of the 39th IEEE Conference on Decision and
  Control} {Proceedings of the 39th ieee conference on decision and control}\
  (\BPGS\ 3331--3336).
\PrintBackRefs{\CurrentBib}

\bibitem [\protect \citeauthoryear {%
Susemihl%
}{%
Susemihl%
}{%
{\protect \APACyear {2014}}%
}]{%
SusemihlThesis2014}
\APACinsertmetastar {%
SusemihlThesis2014}%
\begin{APACrefauthors}%
Susemihl, A.%
\end{APACrefauthors}%
\unskip\
\newblock
\APACrefYear{2014}.
\unskip\
\newblock
\APACrefbtitle {Optimal Population Coding of Dynamic Stimuli} {Optimal
  population coding of dynamic stimuli}\ \APACtypeAddressSchool {\BUPhD}{}{}.
\unskip\
\newblock
\APACaddressSchool {}{Technical University Berlin}.
\PrintBackRefs{\CurrentBib}

\bibitem [\protect \citeauthoryear {%
Susemihl%
, Meir%
\BCBL {}\ \BBA {} Opper%
}{%
Susemihl%
\ \protect \BOthers {.}}{%
{\protect \APACyear {2011}}%
}]{%
SusMeiOpp11}
\APACinsertmetastar {%
SusMeiOpp11}%
\begin{APACrefauthors}%
Susemihl, A.%
, Meir, R.%
\BCBL {}\ \BBA {} Opper, M.%
\end{APACrefauthors}%
\unskip\
\newblock
\APACrefYearMonthDay{2011}{}{}.
\newblock
{\BBOQ}\APACrefatitle {Analytical Results for the Error in Filtering of
  Gaussian Processes} {Analytical results for the error in filtering of
  gaussian processes}.{\BBCQ}
\newblock
\BIn{} J.~Shawe-Taylor, R.~Zemel, P.~Bartlett, F.~Pereira\BCBL {}\ \BBA {}
  K.~Weinberger\ (\BEDS), \APACrefbtitle {Advances in Neural Information
  Processing Systems 24} {Advances in neural information processing systems
  24}\ (\BPGS\ 2303--2311).
\PrintBackRefs{\CurrentBib}

\bibitem [\protect \citeauthoryear {%
Susemihl%
, Meir%
\BCBL {}\ \BBA {} Opper%
}{%
Susemihl%
\ \protect \BOthers {.}}{%
{\protect \APACyear {2013}}%
}]{%
SusMeiOpp13}
\APACinsertmetastar {%
SusMeiOpp13}%
\begin{APACrefauthors}%
Susemihl, A.%
, Meir, R.%
\BCBL {}\ \BBA {} Opper, M.%
\end{APACrefauthors}%
\unskip\
\newblock
\APACrefYearMonthDay{2013}{}{}.
\newblock
{\BBOQ}\APACrefatitle {Dynamic state estimation based on Poisson spike
  trains---towards a theory of optimal encoding} {Dynamic state estimation
  based on poisson spike trains---towards a theory of optimal encoding}.{\BBCQ}
\newblock
\APACjournalVolNumPages{Journal of Statistical Mechanics: Theory and
  Experiment}{2013}{03}{P03009}.
\PrintBackRefs{\CurrentBib}

\bibitem [\protect \citeauthoryear {%
Susemihl%
, Meir%
\BCBL {}\ \BBA {} Opper%
}{%
Susemihl%
\ \protect \BOthers {.}}{%
{\protect \APACyear {2014}}%
}]{%
Susemihl2014}
\APACinsertmetastar {%
Susemihl2014}%
\begin{APACrefauthors}%
Susemihl, A.%
, Meir, R.%
\BCBL {}\ \BBA {} Opper, M.%
\end{APACrefauthors}%
\unskip\
\newblock
\APACrefYearMonthDay{2014}{}{}.
\newblock
{\BBOQ}\APACrefatitle {{Optimal Neural Codes for Control and Estimation}}
  {{Optimal Neural Codes for Control and Estimation}}.{\BBCQ}
\newblock
\APACjournalVolNumPages{Advances in Neural Information Processing
  Systems}{}{}{1--9}.
\PrintBackRefs{\CurrentBib}

\bibitem [\protect \citeauthoryear {%
Twum-Danso%
\ \BBA {} Brockett%
}{%
Twum-Danso%
\ \BBA {} Brockett%
}{%
{\protect \APACyear {2001}}%
}]{%
Twum-Danso2001}
\APACinsertmetastar {%
Twum-Danso2001}%
\begin{APACrefauthors}%
Twum-Danso, N.%
\BCBT {}\ \BBA {} Brockett, R.%
\end{APACrefauthors}%
\unskip\
\newblock
\APACrefYearMonthDay{2001}{}{}.
\newblock
{\BBOQ}\APACrefatitle {{Trajectory estimation from place cell data}}
  {{Trajectory estimation from place cell data}}.{\BBCQ}
\newblock
\APACjournalVolNumPages{Neural Networks}{14}{6-7}{835--844}.
\newblock
\begin{APACrefDOI} \doi{10.1016/S0893-6080(01)00079-X} \end{APACrefDOI}
\PrintBackRefs{\CurrentBib}

\bibitem [\protect \citeauthoryear {%
Wang%
, Stocker%
\BCBL {}\ \BBA {} Lee%
}{%
Wang%
\ \protect \BOthers {.}}{%
{\protect \APACyear {2016}}%
}]{%
WanStoLee2016}
\APACinsertmetastar {%
WanStoLee2016}%
\begin{APACrefauthors}%
Wang, Z.%
, Stocker, A\BPBI A.%
\BCBL {}\ \BBA {} Lee, D\BPBI D.%
\end{APACrefauthors}%
\unskip\
\newblock
\APACrefYearMonthDay{2016}{}{}.
\newblock
{\BBOQ}\APACrefatitle {Efficient neural codes that minimize lp reconstruction
  error} {Efficient neural codes that minimize lp reconstruction error}.{\BBCQ}
\newblock
\APACjournalVolNumPages{Neural computation}{}{}{}.
\PrintBackRefs{\CurrentBib}

\bibitem [\protect \citeauthoryear {%
Yaeli%
\ \BBA {} Meir%
}{%
Yaeli%
\ \BBA {} Meir%
}{%
{\protect \APACyear {2010}}%
}]{%
YaeMei10}
\APACinsertmetastar {%
YaeMei10}%
\begin{APACrefauthors}%
Yaeli, S.%
\BCBT {}\ \BBA {} Meir, R.%
\end{APACrefauthors}%
\unskip\
\newblock
\APACrefYearMonthDay{2010}{}{}.
\newblock
{\BBOQ}\APACrefatitle {Error-based analysis of optimal tuning functions
  explains phenomena observed in sensory neurons.} {Error-based analysis of
  optimal tuning functions explains phenomena observed in sensory
  neurons.}{\BBCQ}
\newblock
\APACjournalVolNumPages{Front Comput Neurosci}{4}{}{130}.
\newblock
\begin{APACrefDOI} \doi{10.3389/fncom.2010.00130} \end{APACrefDOI}
\PrintBackRefs{\CurrentBib}

\end{thebibliography}

\appendix

\section{Derivation of filtering equations\label{sec:Derivation}}

\subsection{Setting and notation}

In the main text, we have presented our model in an open-loop setting,
where the process $X$ is passively observed. Here we consider a more
general setting, incorporating a control process $U_{t}$, so the
dynamics are
\begin{equation}
dX_{t}=\left(A\left(X_{t}\right)+B\left(U_{t}\right)\right)dt+D\left(X_{t}\right)dW_{t},\label{eq:closed_loop_dyn}
\end{equation}
where, in general, $U_{t}$ is a function of $\mathcal{N}_{t}$.

We denote by $p_{t}^{\,\mathcal{N}}\left(\cdot\right)$ the posterior
density -- that is, the density of $X_{t}$ given $\mathcal{N}_{t}$,
and by $\E_{t}^{\mathcal{N}}\left[\cdot\right]$ the posterior expectation
-- that is, expectation conditioned on $\mathcal{N}_{t}$. 

\subsection{The Innovation Measure}

We derive the filtering equations in terms of the \emph{innovation
measure} of the marked point process, which is a random measure closely
related to the notion of \emph{innovation process }associated with
unmarked point processes or diffusion processes. The role of the innovation
measure in filtering marked point processes is analogous to the role
of the innovation process in Kalman filtering. 

In section (\ref{subsec:model-finite}) we characterized the intensity
(or rate) of each point process in a finite population using the asymptotic
behavior of jump probabilities in small intervals. An alternative
definition is the following\footnote{A more detailed discussion of this definition and of innovation processes
is available in \citet{Segall1975-modelling}}: Consider an unmarked point process $N_{t}$ with $\mc N_{t}$ denoting
its history up to time $t$. Given some history $\mathcal{F}_{t}$
containing $\mc N_{t}$ (e.g. $\mc F_{t}$ might include the observation
process $N$ and its driving state process $X$), the process $\lambda_{t}^{\mc F}$
is called the intensity of $N$ relative to $\mc F$ when $\lambda_{t}$
is $\mc F_{t}$-measurable\footnote{i.e., it is strictly a function of the history $\mc F_{t}$.}
and 
\[
I_{t}^{\mc F}\triangleq N_{t}-\int_{0}^{t}\lambda_{s}^{\mc F}ds
\]
is an $\mc F_{t}$-martingale, meaning 
\[
\E\left[I_{t}^{\mathcal{F}}|\mathcal{F}_{s}\right]=I_{s}^{\mathcal{F}},
\]
or equivalently,
\[
\E\left[N_{t}|\mathcal{F}_{s}\right]=N_{s}+\E\left[\int_{s}^{t}\lambda^{\mc F}\Big|\mc{\mc F}_{s}\right].
\]
Heuristically we may write this relation as
\[
\E\left(dN_{t}|\mc F_{t}\right)=\lambda_{t}^{\mc F}dt.
\]
When $\mc F_{t}=\mc N_{t}$, the process $I_{t}\triangleq I_{t}^{\mc N}$
is called the \emph{innovation process}. We may write
\[
dN_{t}=\lambda_{t}^{\mc N}dt+dI_{t}=\E\left[dN_{t}|\mc N_{t}\right]+dI_{t},
\]
so the innovation increments $dI_{t}$ represent the ``unexpected''
part of the increments $dN_{t}$ after observation of $\mc N_{t}$.
The innovation process may be similarly defined for other stochastic
processes (such as discrete-time or diffusion processes), and it plays
an important role in the Kalman and the Kalman-Bucy filters. It plays
an analogous role in the filtering of point processes, as seen below.

As discussed in section (\ref{subsec:marked}), in the continuous
population model we characterize the observation process $N$ by its
intensity kernel relative to the history $\left(\mc N_{t},X_{\left[0,t\right]}\right)$,
\[
\E\left[N\left(dt,d\boldsymbol{y}\right)|X_{\left[0,t\right]},\mathcal{N}_{t}\right]=\lambda\left(X_{t};\boldsymbol{y}\right)f\left(d\boldsymbol{y}\right)dt.
\]
This equation is heuristic notation for the statement that the rate
of $N_{t}\left(Y\right)$ relative to $\left(\mc N_{t},X_{\left[0,t\right]}\right)$
is $\int_{Y}\lambda\left(X_{t};\boldsymbol{y}\right)f\left(d\boldsymbol{y}\right)$,
for any measurable $Y\subset\boldsymbol{Y}$. The rate relative to
the spiking history $\mc N_{t}$ is obtained by marginalizing over
$X_{t}$ (see \citet{Segall1975-modelling}, Theorem 2), yielding
the rate $\int_{\boldsymbol{Y}}\hat{\lambda}\left(\boldsymbol{y}\right)f\left(d\boldsymbol{y}\right)$
where 
\[
\hat{\lambda}\left(\boldsymbol{y}\right)=\E_{t}^{\mc N}\left[\lambda\left(X_{t};\boldsymbol{y}\right)\right].
\]
Therefore the innovation process of $N_{t}\left(Y\right)$ is $I_{t}\left(Y\right)=N_{t}\left(Y\right)-\int_{0}^{t}\int_{Y}\hat{\lambda}_{s}\left(y\right)f\left(d\boldsymbol{y}\right)ds$,
and accordingly we define the \emph{innovation measure} $I$ to be
the random measure
\begin{equation}
I\left(dt,d\boldsymbol{y}\right)\triangleq N\left(dt,dy\right)-\hat{\lambda}_{t}\left(\boldsymbol{y}\right)f\left(d\boldsymbol{y}\right)dt,\label{eq:innovation}
\end{equation}
so that the innovation process may be expressed as
\[
I_{t}\left(Y\right)=\int_{0}^{t}\int_{\boldsymbol{Y}}I\left(ds,dy\right).
\]
The martingale property of $I_{t}\left(Y\right)$ implies that the
innovation measure satisfies
\[
\E_{t}^{\mc N}\left[\int_{Y}I\left(dt,dy\right)\right]=0,
\]
for all measurable $Y\subset\boldsymbol{Y}$.

\subsection{Exact filtering equations}

Define
\[
\omega_{t}\left(x;\boldsymbol{y}\right)\triangleq\frac{\lambda\left(x;\boldsymbol{y}\right)}{\hat{\lambda}_{t}\left(\boldsymbol{y}\right)}-1,\quad\omega_{t}^{\boldsymbol{y}}\triangleq\omega_{t}\left(X_{t};\boldsymbol{y}\right).
\]
The stochastic PDE (\ref{eq:snyder-density}) is extended in \citet{RhoSny1977}
for the case of marked point process observation in the presence of
feedback\footnote{The setting considered in \citet{RhoSny1977} assumes linear dynamics
and an inifinite uniform population. However, these assumption are
only relevant to establish other proposition in that paper. The proof
of equation (\ref{eq:snyder-density}) still holds as is in our more
general setting.}: in this case, the posterior density $p_{t}^{\,\mathcal{N}}$ obeys
the equation

\begin{equation}
dp_{t}^{\,\mathcal{N}}\left(x\right)=\left\{ \mathcal{L}_{t}^{*}p_{t}^{\,\mathcal{N}}\right\} \left(x\right)dt+p_{t^{-}}^{\,\mathcal{N}}\left(x\right)\int_{y\in\mathbf{Y}}\omega_{t}\left(x;\boldsymbol{y}\right)I\left(dt,d\boldsymbol{y}\right),\label{eq:rhodes-snyder-cont}
\end{equation}
Here $\mathcal{L}_{t}$ is the posterior infinitesimal generator,
defined with an additional conditioning on $\mathcal{N}_{t}$,
\[
\mathcal{L}_{t}h\left(x\right)=\lim_{\Delta t\to0^{+}}\frac{\mathrm{E}\left[h\left(X_{t+\Delta t}\right)|X_{t}=x,\mathcal{N}_{t}\right]-h\left(x\right)}{\Delta t},
\]
and $\mathcal{L}_{t}^{*}$ is its adjoint. Note that in this closed-loop
setting, the infinitesimal generator is itself a random operator,
due to its dependence on past observations through the control law,
and that $N_{t}$ is no longer a doubly-stochastic Poisson process.

As in section (\ref{subsec:exact}), we use the notations
\[
\mu_{t}\triangleq\E_{t}^{\mathcal{N}}X_{t},\quad\tilde{X}_{t}\triangleq X_{t}-\mu_{t},\quad\Sigma_{t}\triangleq\E_{t}^{\mathcal{N}}\left[\tilde{X}_{t}\tilde{X}_{t}\transpose\right].
\]
We derive of the following equations for the first two posterior moments,
which generalize (\ref{eq:finite}) to marked point processes, and
for the presence of feedback in the state dynamics,\begin{subequations}\label{eq:exact-innovation}
\begin{align}
d\mu_{t} & =\left(\E_{t}^{\mathcal{N}}\left[A\left(X_{t}\right)\right]+B\left(U_{t}\right)\right)dt+\int_{\mathbf{Y}}\E_{t^{-}}^{\mathcal{N}}\left[\omega_{t^{-}}^{\boldsymbol{y}}X_{t^{-}}\right]I\left(dt,d\boldsymbol{y}\right)\label{eq:mean-I}\\
d\Sigma_{t} & =\E_{t}^{\mathcal{N}}\left[A\left(X_{t}\right)\tilde{X}_{t}^{T}+\tilde{X}_{t}A\left(X_{t}\right)^{T}+D\left(X_{t}\right)D\left(X_{t}\right)^{T}\right]dt\nonumber \\
 & \quad+\int_{\mathbf{Y}}\E_{t^{-}}^{\mathcal{N}}\left[\omega_{t^{-}}^{\boldsymbol{y}}\tilde{X}_{t^{-}}\tilde{X}_{t^{-}}^{T}\right]I\left(dt,d\boldsymbol{y}\right)\nonumber \\
 & \quad-\int_{\mathbf{Y}}\E_{t^{-}}^{\mathcal{N}}\left[\omega_{t^{-}}^{\boldsymbol{y}}X_{t^{-}}\right]\E_{t^{-}}^{\mathcal{N}}\left[\omega_{t^{-}}^{\boldsymbol{y}}X_{t^{-}}^{T}\right]N\left(dt,d\boldsymbol{y}\right).\label{eq:var-I}
\end{align}
\end{subequations}A rigorous derivation of (\ref{eq:mean-I}) under
more general conditions is found in \citet{Segall1975-filtering},
from which (\ref{eq:var-I}) may be derived by considering the dynamics
of the process $X_{t}X_{t}\transpose$. Here we provide a more heuristic
derivation based on (\ref{eq:rhodes-snyder-cont}).

Compare the mean update equation (\ref{eq:mean-I}) to the Kalman-Bucy
filter, which gives the posterior moments for a diffusion process
$X$ with noisy observations $Y$ of the form
\begin{align*}
dX_{t} & =A_{t}dX_{t}+D_{t}dW_{t}\\
dY_{t} & =H_{t}X_{t}dt+dV_{t},
\end{align*}
where $W,V$ are independent standard Wiener processes. The Kalman-Bucy
filter reads
\begin{align*}
d\mu_{t} & =A_{t}\mu_{t}dt+\Sigma_{t}H_{t}\transpose\left(dY_{t}-H_{t}\mu_{t}dt\right),\\
\dot{\Sigma}_{t} & =A_{t}\Sigma_{t}+\Sigma_{t}A_{t}+D_{t}D_{t}\transpose-\Sigma_{t}H_{t}\transpose H_{t}\Sigma_{t}.
\end{align*}
Here, the term $dY_{t}-H_{t}\mu_{t}dt$ appearing in the first equation
is an increment of the innovation process $Y_{t}-\int_{0}^{t}H_{s}\mu_{s}ds$.

For a sufficiently well-behaved function $h$, we find, using (\ref{eq:rhodes-snyder-cont})
and the definition of operator adjoint,

\begin{align}
d\E_{t}^{\mc N}\left[h\left(X_{t}\right)\right] & =\int dx\,h\left(x\right)\left[\left\{ \mathcal{L}_{t}^{*}p_{t}^{\,\mathcal{N}}\right\} \left(x\right)dt+p_{t^{-}}^{\,\mathcal{N}}\left(x\right)\int_{\boldsymbol{y}\in\boldsymbol{Y}}\omega_{t-}\left(x;\boldsymbol{y}\right)I\left(dt,d\boldsymbol{y}\right)\right]\nonumber \\
 & =\int dx\,\left[p_{t}^{\,\mathcal{N}}\left(x\right)\left\{ \mathcal{L}_{t}h\right\} \left(x\right)dt+p_{t^{-}}^{\,\mathcal{N}}\left(x\right)h\left(x\right)\int_{\boldsymbol{y}\in\boldsymbol{Y}}\omega_{t-}\left(x;\boldsymbol{y}\right)I\left(dt,d\boldsymbol{y}\right)\right]\nonumber \\
 & =\E_{t}^{\mc N}\left[\left\{ \mathcal{L}_{t}h\right\} \left(X_{t}\right)\right]dt+\int_{\boldsymbol{y}\in\boldsymbol{Y}}\E_{t^{-}}^{\mc N}\left[h\left(X_{t^{-}}\right)\omega_{t^{-}}^{\boldsymbol{y}}\right]I\left(dt,d\boldsymbol{y}\right)\label{eq:inc-E-h}
\end{align}
Assuming the state evolves as in (\ref{eq:closed_loop_dyn}), the
(closed loop) infinitesimal generator is
\[
\mathcal{L}_{t}h\left(x\right)=\left(A\left(x\right)+B\left(U_{t}\right)\right)\transpose\nabla h\left(x\right)+\frac{1}{2}\mathrm{Tr}\left[\nabla^{2}h\left(x\right)D\left(x\right)D\left(x\right)\transpose\right],
\]
which, when specialized to the functions $h_{i}\left(x\right)=x^{i}$
and $h_{ij}\left(x\right)=x^{i}x^{j}$, where $x^{i}$ is the $i$th
component of $x$, reads
\begin{align}
\mathcal{L}_{t}h_{i}\left(x\right) & =\left(A\left(x\right)+B\left(U_{t}\right)\right)^{i}\label{eq:gen-xi}\\
\mathcal{L}h_{ij}\left(x\right) & =\left(A\left(x\right)+B\left(U_{t}\right)\right)^{i}x^{j}+x^{i}\left(A\left(x\right)+B\left(U_{t}\right)\right)^{j}\nonumber \\
 & \quad+\frac{1}{2}\left(D\left(x\right)D\left(x\right)\transpose\right)^{ij}+\frac{1}{2}\left(D\left(x\right)D\left(x\right)\transpose\right)^{ji}\nonumber \\
 & =\left(A\left(x\right)+B\left(U_{t}\right)\right)^{i}x^{j}+x^{i}\left(A\left(x\right)+B\left(U_{t}\right)\right)^{j}\nonumber \\
 & \quad+\left(D\left(x\right)D\left(x\right)\transpose\right)^{ij}\label{eq:gen-xixj}
\end{align}
Substituting (\ref{eq:gen-xi}) into (\ref{eq:inc-E-h}) yields
\[
d\mu_{t}^{i}=\E_{t}^{\mathcal{N}}\left(A\left(x\right)+B\left(U_{t}\right)\right)^{i}dt+\int_{\boldsymbol{y}\in\boldsymbol{Y}}\E_{t^{-}}^{\mc N}\left[X_{t^{-}}^{i}\omega_{t^{-}}^{\boldsymbol{y}}\right]I\left(dt,d\boldsymbol{y}\right),
\]
or in vector notation
\begin{align}
d\mu_{t} & =\E_{t}^{\mc N}\left[A\left(X_{t}\right)\right]dt+B\left(U_{t}\right)dt+\int_{\boldsymbol{y}\in\boldsymbol{Y}}\E_{t^{-}}^{\mc N}\left[\omega_{t^{-}}^{\boldsymbol{y}}X_{t^{-}}\right]I\left(dt,d\boldsymbol{y}\right)\label{eq:dmu}
\end{align}

To compute $d\Sigma_{t}$ we use the representation $d\Sigma_{t}=d\E_{t}^{\mc N}\left[X_{t}X_{t}\transpose\right]-d\left(\mu_{t}\mu_{t}\transpose\right)$.
The first term is computed by substituting (\ref{eq:gen-xixj}) into
(\ref{eq:inc-E-h}), yielding
\begin{align*}
d\E_{t}^{\mathcal{N}}\left[X_{t}^{i}X_{t}^{j}\right] & =\E_{t}^{\mathcal{N}}\left[\left(A\left(X_{t}\right)+B\left(U_{t}\right)\right)^{i}X_{t}^{j}+X_{t}^{i}\left(A\left(X_{t}\right)+B\left(U_{t}\right)\right)^{j}\right]\\
 & \quad+\E_{t}^{\mathcal{N}}\left[\left(D\left(X_{t}\right)D\left(X_{t}\right)\transpose\right)^{ij}\right]+\int_{\boldsymbol{y}\in\boldsymbol{Y}}\E_{t^{-}}^{\mc N}\left[X_{t^{-}}^{i}X_{t^{-}}^{j}\omega_{t^{-}}^{\boldsymbol{y}}\right]I\left(dt,d\boldsymbol{y}\right),
\end{align*}
or in matrix notation, after some rearranging,
\begin{align}
d\E_{t}^{\mc N}\left[X_{t}X_{t}\transpose\right] & =\E_{t}^{\mc N}\left[A\left(X_{t}\right)X_{t}\transpose+X_{t}A\left(X_{t}\right)\transpose\right]dt+\left[B\left(U_{t}\right)\mu_{t}\transpose+\mu_{t}B\left(U_{t}\right)\transpose\right]dt\nonumber \\
 & \quad+\E_{t}^{\mc N}\left[D\left(X_{t}\right)D\left(X_{t}\right)\transpose\right]dt\nonumber \\
 & \quad+\int_{\boldsymbol{y}\in\boldsymbol{Y}}\E_{t^{-}}^{\mc N}\left[\omega_{t^{-}}^{\boldsymbol{y}}X_{t^{-}}X_{t^{-}}\transpose\right]I\left(dt,d\boldsymbol{y}\right)\label{eq:dXXT}
\end{align}
To calculate $d\left(\mu_{t}\mu_{t}\transpose\right)$ from (\ref{eq:dmu})
we separately handle the continuous terms, and the jump term involving
$N\left(dt,d\boldsymbol{y}\right)$. The continuous terms are the
continuous part of $d\mu_{t}\mu_{t}\transpose+\mu_{t}d\mu_{t}\transpose$.
To compute the jump terms, we note that when $\mu_{t}$ jumps by $\Delta_{t}$,
the corresponding jump in $\mu_{t}\mu_{t}\transpose$ is $\Delta_{t}\mu_{t}\transpose+\mu_{t}\Delta_{t}\transpose+\Delta_{t}\Delta_{t}\transpose$,
therefore
\begin{align}
d\left(\mu_{t}\mu_{t}\transpose\right)= & \left(\E_{t}^{\mc N}\left[A\left(X_{t}\right)\right]\mu_{t}\transpose+\mu_{t}\E_{t}^{\mc N}\left[A\left(X_{t}\right)\right]\transpose\right)dt+\left[B\left(U_{t}\right)\mu_{t}\transpose+\mu_{t}B\left(U_{t}\right)\transpose\right]dt\nonumber \\
 & -\int_{\boldsymbol{y}\in\boldsymbol{Y}}\left(\E_{t}^{\mc N}\left[\omega_{t}^{\boldsymbol{y}}X_{t}\right]\mu_{t}\transpose+\mu_{t}\E_{t}^{\mc N}\left[\omega_{t}^{\boldsymbol{y}}X_{t}\right]\transpose\right)\hat{\lambda}_{t}\left(\boldsymbol{y}\right)f\left(d\boldsymbol{y}\right)dt\\
 & +\int_{\boldsymbol{y}\in\boldsymbol{Y}}\Big(\E_{t^{-}}^{\mc N}\left[\omega_{t^{-}}^{\boldsymbol{y}}X_{t^{-}}\right]\mu_{t}\transpose+\mu_{t}\E_{t}^{\mc N}\left[\omega_{t^{-}}^{\boldsymbol{y}}X_{t}\right]\transpose\nonumber \\
 & \qquad+\E_{t^{-}}^{\mc N}\left[\omega_{t^{-}}^{\boldsymbol{y}}X_{t^{-}}\right]\E_{t^{-}}^{\mc N}\left[\omega_{t^{-}}^{\boldsymbol{y}}X_{t^{-}}\right]\transpose\Big)N\left(dt,d\boldsymbol{y}\right)\label{eq:dmu-mu}
\end{align}
Subtracting (\ref{eq:dmu-mu}) from (\ref{eq:dXXT}), and noting that
$\E_{t}^{\mc N}\omega_{t}^{\boldsymbol{y}}=0$, so that
\begin{align*}
\E_{t^{-}}^{\mc N}\left[\omega_{t^{-}}^{\boldsymbol{y}}\tilde{X}_{t^{-}}\tilde{X}_{t^{-}}\transpose\right] & =\E_{t^{-}}^{\mc N}\left[\omega_{t^{-}}^{\boldsymbol{y}}X_{t^{-}}X_{t^{-}}\transpose\right]\\
 & \quad-\E_{t^{-}}^{\mc N}\left[\omega_{t^{-}}^{\boldsymbol{y}}X_{t^{-}}\right]\mu_{t^{-}}\transpose-\mu_{t^{-}}\E_{t^{-}}^{\mc N}\left[\omega_{t^{-}}^{\boldsymbol{y}}X_{t^{-}}\right]\transpose
\end{align*}
yields (\ref{eq:var-I}).

Writing (\ref{eq:mean-I})-(\ref{eq:var-I}) according to the decomposition
described in section \ref{subsec:exact},
\begin{align}
d\mu_{t}^{\pi} & =\left(\E_{t}^{\mathcal{N}}\left[A\left(X_{t}\right)\right]+B\left(U_{t}\right)\right)dt\nonumber \\
d\Sigma_{t}^{\pi} & =\E_{t}^{\mathcal{N}}\left[A\left(X_{t}\right)\tilde{X}_{t}\transpose+\tilde{X}_{t}A\left(X_{t}\right)\transpose+D\left(X_{t}\right)D\left(X_{t}\right)\transpose\right]dt\nonumber \\
d\mu_{t}^{\mathrm{c}} & =-\int_{\boldsymbol{y}\in\boldsymbol{Y}}\E_{t}^{\mathcal{N}}\left[\omega_{t}^{\boldsymbol{y}}X_{t}\right]\hat{\lambda}\left(\boldsymbol{y}\right)f\left(d\boldsymbol{y}\right)dt\label{eq:mean-c}\\
d\Sigma_{t}^{\mathrm{c}} & =-\int_{\boldsymbol{y}\in\boldsymbol{Y}}\E_{t}^{\mathcal{N}}\left[\omega_{t}^{\boldsymbol{y}}\tilde{X}_{t}\tilde{X}_{t}\transpose\right]\hat{\lambda}\left(\boldsymbol{y}\right)f\left(d\boldsymbol{y}\right)dt\label{eq:var-c}\\
d\mu_{t}^{N} & =\int_{\boldsymbol{y}\in\boldsymbol{Y}}\E_{t^{-}}^{\mathcal{N}}\left[\omega_{t^{-}}^{\boldsymbol{y}}X_{t^{-}}\right]N\left(dt,d\boldsymbol{y}\right)\label{eq:mean-N}\\
d\Sigma_{t}^{N} & =\int_{\boldsymbol{y}\in\boldsymbol{Y}}\bigg(\E_{t^{-}}^{\mathcal{N}}\left[\omega_{t^{-}}^{\boldsymbol{y}}\tilde{X}_{t^{-}}\tilde{X}_{t^{-}}\transpose\right]-\E_{t^{-}}^{\mathcal{N}}\left[\omega_{t^{-}}^{\boldsymbol{y}}X_{t^{-}}\right]\E_{t^{-}}^{\mathcal{N}}\left[\omega_{t^{-}}^{\boldsymbol{y}}X_{t^{-}}\transpose\right]\bigg)N\left(dt,d\boldsymbol{y}\right).\label{eq:var-N}
\end{align}

\subsection{ADF approximation for Gaussian tuning\label{subsec:ADF-gaussian-derivation}}

We now proceed to apply the Gaussian ADF approximation $p_{t}^{\,\mathcal{N}}\left(x\right)\approx\mathcal{N}\left(x;\mu_{t},\Sigma_{t}\right)$
to (\ref{eq:mean-c})-(\ref{eq:var-N}) in the case of Gaussian neurons
(\ref{eq:gauss-tc}), deriving approximate filtering equations written
in terms of the population density $f\left(d\boldsymbol{y}\right)$.
From here on we use $\mu_{t},\Sigma_{t}$, and $p_{t}^{\,\mathcal{N}}$
to refer to the ADF approximation rather than to the exact values.

We use the following algebraic results. The first is a slightly generalized
form of a well-known result about the sum of quadratic forms, which
is useful for multiplying Gaussians with possibly degenerate precision
matrices. 
\begin{claim}
\label{claim:quad-sum}Let $x,a,b\in\mathbb{R}^{n}$ and $A,B\in\mathbb{R}^{n\times n}$
be symmetric matrices such that $A+B$ is non-singular. Then
\[
\left\Vert x-a\right\Vert _{A}^{2}+\left\Vert x-b\right\Vert _{B}^{2}=\left\Vert a-b\right\Vert _{A\left(A+B\right)^{-1}B}^{2}+\left\Vert x-\left(A+B\right)^{-1}\left(Aa+Bb\right)\right\Vert _{A+B}^{2}.
\]
\end{claim}

\begin{proof}
By straightforward expansion of each side.
\end{proof}
Next we note two matrix inversion lemmas, the first of which is known
as the Woodbury identity. These are useful for transferring variance
matrices between state and perceptual coordinates in our model. Derivations
may be found in \citet{Henderson1981}.
\begin{claim}
\label{claim:inversion}For $U,V\transpose\in\mathbb{R}^{m\times n}$
and non-singular $A\in\mathbb{R}^{m\times m},C\in\mathbb{R}^{n\times n}$,
the following equalities hold
\begin{align}
\left(A+UCV\right)^{-1} & =A^{-1}-A^{-1}U\left(C^{-1}+VA^{-1}U\right)^{-1}VA^{-1},\label{eq:woodbury}\\
A^{-1}U\left(C^{-1}+VA^{-1}U\right)^{-1} & =\left(A+UCV\right)^{-1}UC,\nonumber 
\end{align}
whenever all the relevant inverses exist.
\end{claim}

To evaluate the posterior of expectations in (\ref{eq:mean-c})-(\ref{eq:var-N})
we first simplify the expression
\begin{equation}
p_{t}^{\,\mathcal{N}}\left(x\right)\omega_{t}\left(x;\boldsymbol{y}\right)=\frac{p_{t}^{\,\mathcal{N}}\left(x\right)\lambda\left(x;\boldsymbol{y}\right)}{\int p_{t}^{\,\mathcal{N}}\left(\xi\right)\lambda\left(\xi;\boldsymbol{y}\right)d\xi}-p_{t}^{\,\mathcal{N}}\left(x\right).\label{eq:P-omega}
\end{equation}

Using the Gaussian ADF approximation $p_{t}^{\,\mathcal{N}}\left(x\right)=\mathcal{N}\left(x;\mu_{t},\Sigma_{t}\right)$,
equation (\ref{eq:gauss-tc-y}), and Claim \ref{claim:quad-sum},
we find 
\begin{align*}
p_{t}^{\,\mathcal{N}}\left(x\right) & \lambda\left(x;h,\theta,H,R\right)=\\
 & =h\mathcal{N}\left(x;\mu_{t},\Sigma_{t}\right)\exp\left(-\frac{1}{2}\left\Vert Hx-\theta\right\Vert _{R}^{2}\right)\\
 & =\frac{h\exp\left(-\frac{1}{2}\left\Vert x-\mu_{t}\right\Vert _{\Sigma_{t}^{-1}}^{2}-\frac{1}{2}\left\Vert Hx-\theta\right\Vert _{R}^{2}\right)}{\sqrt{\left(2\pi\right)^{n}\left|\Sigma_{t}\right|}}\\
 & =\frac{h}{\sqrt{\left(2\pi\right)^{n}\left|\Sigma_{t}\right|}}\exp\left(-\frac{1}{2}\left\Vert H_{r}^{-1}\theta-\mu_{t}\right\Vert _{Q_{t}}^{2}-\frac{1}{2}\left\Vert x-\mu_{t}^{\theta}\right\Vert _{\Sigma_{t}^{-1}+H^{T}RH}^{2}\right),
\end{align*}
where $H_{r}^{-1}$ is any right inverse of $H$, and 
\begin{eqnarray}
Q_{t} & \triangleq & \Sigma_{t}^{-1}\left(\Sigma_{t}^{-1}+H\transpose RH\right)^{-1}H\transpose RH,\label{eq:Q_t}\\
\mu_{t}^{\theta} & \triangleq & \left(\Sigma_{t}^{-1}+H\transpose RH\right)^{-1}\left(\Sigma_{t}^{-1}\mu_{t}+H\transpose R\theta\right).\label{eq:mu-theta}
\end{eqnarray}
To simplify the notation, we suppress the dependence of these and
other quantities on $H$ and $R$ throughout this section. Claim \ref{claim:inversion}
establishes the relation $Q_{t}=H\transpose S_{t}H$, where 
\[
S_{t}\triangleq\left(R^{-1}+H\Sigma_{t}H\transpose\right)^{-1},
\]
yielding 
\begin{align}
 & p_{t}^{\,\mathcal{N}}\left(x\right)\lambda\left(x;h,\theta,H,R\right)=h\left(2\pi\right)^{-n/2}\left|\Sigma_{t}\right|^{-1}\label{eq:post-times-rate}\\
 & \quad\times\exp\left(-\frac{1}{2}\left\Vert \theta-H\mu_{t}\right\Vert _{S_{t}}^{2}-\frac{1}{2}\left\Vert x-\mu_{t}^{\theta}\right\Vert _{\Sigma_{t}^{-1}+H\transpose RH}^{2}\right),\nonumber 
\end{align}
and by normalizing this Gaussian (see (\ref{eq:P-omega})) we find
that 
\begin{align*}
p_{t}^{\,\mathcal{N}}\left(x\right)\omega_{t}\left(x;h,\theta,R\right) & =\mathcal{N}\left(x;\mu_{t}^{\theta},\left(\Sigma_{t}^{-1}+H\transpose RH\right)^{-1}\right)-p_{t}^{\,\mathcal{N}}\left(x\right)
\end{align*}
yielding
\begin{align}
\E_{t}^{\mc N}\left[\omega_{t}\left(\boldsymbol{y}\right)X_{t}\right] & =\mu_{t}^{\theta}-\mu_{t}\label{eq:mu-diff}\\
\E_{t}^{\mc N}\left[\omega_{t}\left(\boldsymbol{y}\right)\tilde{X}_{t}\tilde{X}_{t}^{T}\right] & =\left(\Sigma_{t}^{-1}+H\transpose RH\right)^{-1}+\left(\mu_{t}-\mu_{t}^{\theta}\right)\left(\mu_{t}-\mu_{t}^{\theta}\right)\transpose-\Sigma_{t}.\nonumber 
\end{align}
Using Claim \ref{claim:inversion}, the difference $\mu_{t}^{\theta}-\mu_{t}$
may be rewritten as 
\begin{align*}
\mu_{t}^{\theta}-\mu_{t} & =\left(\Sigma_{t}^{-1}+H\transpose RH\right)^{-1}\left(\Sigma_{t}^{-1}\mu_{t}+H\transpose R\theta\right)-\mu_{t}\\
 & =-\left(\Sigma_{t}^{-1}+H\transpose RH\right)^{-1}H\transpose RH\mu_{t}+\left(\Sigma_{t}^{-1}+H\transpose RH\right)^{-1}H\transpose R\theta\\
 & =-\left(\Sigma_{t}^{-1}+H\transpose RH\right)^{-1}H\transpose R\delta_{t}\\
 & =-\Sigma_{t}H\transpose S_{t}\delta_{t},
\end{align*}
where $\delta_{t}=H\mu_{t}-\theta$, and an application of the Woodbury
identity (\ref{eq:woodbury}) yields
\begin{align*}
\left(\Sigma_{t}^{-1}+H\transpose RH\right)^{-1}-\Sigma_{t} & =\Sigma_{t}\left(\left(I+H\transpose RH\Sigma_{t}\right)^{-1}-I\right)\\
 & =\Sigma_{t}\left(I-H\transpose\left(R^{-1}+H\Sigma_{t}H\transpose\right)^{-1}H\Sigma_{t}-I\right)\\
 & =-\Sigma_{t}H\transpose S_{t}H\Sigma_{t}.
\end{align*}
 Now,
\begin{align*}
\mathrm{E}_{P}^{t} & \left[\omega_{t}\left(\boldsymbol{y}\right)X_{t}\right]=-\Sigma_{t}H^{T}S_{t}\delta_{t},\\
\mathrm{E}_{P}^{t} & \left[\omega_{t}\left(\boldsymbol{y}\right)\tilde{X}_{t}\tilde{X}_{t}^{T}\right]=\Sigma_{t}H^{T}\left(S_{t}\delta_{t}\delta_{t}\transpose S_{t}-S_{t}\right)H\Sigma_{t}.
\end{align*}
Plugging this result into (\ref{eq:mean-c})-(\ref{eq:var-N}) yields
(\ref{eq:mean-c-gaussian})-(\ref{eq:var-N-gaussian}). Integrating
(\ref{eq:post-times-rate}) over $x$ yields
\[
\hat{\lambda}\left(h,\theta,H,R\right)=\frac{h}{\sqrt{\left|\left(I+\Sigma_{t}H\transpose RH\right)\right|}}\exp\left(-\frac{1}{2}\left\Vert \theta-H\mu_{t}\right\Vert _{S_{t}}^{2}\right).
\]
Sylvester's determinant identity yields the equality $\left|I+\Sigma_{t}H^{T}RH\right|=\left|R\right|/\left|S_{t}^{R}\right|$,
from which (\ref{eq:self-rate-gaussian}) follows.

The continuous precision update (\ref{eq:prec-c-gaussian}) follows
directly from (\ref{eq:var-c-gaussian}) and the relation
\[
\frac{d\Sigma_{t}^{-1}}{dt}=-\Sigma_{t}^{-1}\frac{d\Sigma_{t}}{dt}\Sigma_{t}^{-1},
\]
which holds whenever $\Sigma_{t}$ is differentiable -- i.e., between
spikes. To derive (\ref{eq:prec-N-gaussian}), consider a spike at
time $t$ with mark $\boldsymbol{y}=\left(h,\theta,H,R\right)$. According
to (\ref{eq:var-N-gaussian}), 
\begin{align*}
\Sigma_{t}^{-1} & =\left(\Sigma_{t^{-}}-\Sigma_{t^{-}}H\transpose S_{t^{-}}H\Sigma_{t^{-}}\right)^{-1}\\
 & =\left(I-H\transpose S_{t^{-}}H\Sigma_{t^{-}}\right)^{-1}\Sigma_{t^{-}}^{-1}\\
 & =\left(I+H\transpose\left(S_{t^{-}}^{-1}-H\Sigma_{t^{-}}H\transpose\right)^{-1}H\Sigma_{t^{-}}\right)\Sigma_{t^{-}}^{-1}\tag{Woodbury}\\
 & =\Sigma_{t^{-}}^{-1}+H\transpose RH\tag{from \eqref{eq:S}}
\end{align*}
Finally, equation (\ref{eq:mean-N}) and (\ref{eq:mu-diff}) yields
$d\mu_{t}^{N}=\int\left(\mu_{t^{-}}^{\theta}-\mu_{t^{-}}\right)N\left(dt,d\boldsymbol{y}\right)$,
so at spike times 
\[
\mu_{t^{+}}=\mu_{t^{-}}^{\theta}=\left(\Sigma_{t^{-}}^{-1}+H\transpose RH\right)^{-1}\left(\Sigma_{t^{-}}^{-1}\mu_{t^{-}}+H\transpose R\theta\right).
\]
The finite-population case is (\ref{eq:mean-spike}).

\subsection{Approximation of continuous terms for specific population distributions}

\subsubsection{Gaussian population}

When the preferred stimuli are normally distributed (\ref{eq:gaussian-f}),
the continuous update terms (\ref{eq:mean-c-gaussian})-(\ref{eq:var-c-gaussian})
may be computed analogously to the derivation in section (\ref{subsec:ADF-gaussian-derivation})
above. First, starting from (\ref{eq:self-rate-gaussian}), an analogous
computation to the derivation of (\ref{eq:post-times-rate}) and (\ref{eq:self-rate-gaussian})
above yields
\begin{align*}
\hat{\lambda}_{t}\left(h,\theta,H,R\right)f\left(d\theta\right) & =\hat{\lambda}_{t}^{f}\cdot\mathcal{N}\left(\theta;\mu_{t}^{f},\left(\Sigma_{\mathrm{pop}}^{-1}+S_{t}\right)^{-1}\right)d\theta,
\end{align*}
where 
\[
\mu_{t}^{f}\triangleq\left(\Sigma_{\mathrm{pop}}^{-1}+S_{t}\right)^{-1}\left(S_{t}\mu_{t}+\Sigma_{\mathrm{pop}}^{-1}c\right).
\]
Integrating this equation over $\theta$ and applying Sylvester's
determinant lemma as in the derivation of (\ref{eq:self-rate-gaussian})
yields (\ref{eq:total-rate-gauss-gauss}). The matrix $Z_{t}^{H,R}$
(eq. (\ref{eq:Z^HR})) and vector $\mu_{t}^{f}$ play an analogous
role to that of $S_{t}$ and $\mu_{t}^{\theta}$ respectively in section
(\ref{subsec:ADF-gaussian-derivation}). Substituting into (\ref{eq:mean-c-gaussian})-(\ref{eq:var-c-gaussian})
and simplifying analogously yields (\ref{eq:c-gauss-gauss}).

\subsubsection{Uniform population on an interval}

In this case $R,h$ are constant, $m=n=1$ and $H=1$, and the preferred
stimulus distribution is $f\left(d\theta\right)=1_{\left[a,b\right]}\left(\theta\right)d\theta$,
so (\ref{eq:mean-c-gaussian})-(\ref{eq:var-c-gaussian}) take the
form

\begin{align}
d\mu_{t}^{\mathrm{c}} & =\frac{\sigma_{t}^{2}}{\sigma_{t}^{2}+\alpha^{2}}\int_{a}^{b}\left(\mu_{t}-\theta\right)\hat{\lambda}_{t}\left(\theta\right)d\theta dt\label{eq:mean-c-gaussian-1d}\\
d\sigma_{t}^{2,\mathrm{c}} & =\left(\hat{\lambda}_{t}^{f}-\frac{\int_{a}^{b}\left(\theta-\mu_{t}\right)^{2}\hat{\lambda}_{t}\left(\theta\right)d\theta}{\sigma_{t}^{2}+\alpha^{2}}\right)\frac{\sigma_{t}^{2}}{\sigma_{t}^{2}+\alpha^{2}}\sigma_{t}^{2}dt\label{eq:var-c-gaussian-1d}\\
\hat{\lambda}_{t}\left(\theta\right) & =h\sqrt{2\pi\alpha^{2}}\mathcal{N}\left(\theta;\mu_{t},\sigma_{t}^{2}+\alpha^{2}\right),\nonumber 
\end{align}
where $\sigma_{t}^{2}=\Sigma_{t},\sigma_{\mathrm{r}}^{2}=R^{-1}$,
and we suppressed the dependence of $\hat{\lambda}$ on $h,R$ from
the notation, since $h,R$ are fixed. The integrals can be computed
from the following identities
\begin{eqnarray*}
\int_{a}^{b}\mathcal{N}\left(x;\mu,\sigma^{2}\right)\left(x-\mu\right)\,dx & = & -z\sigma,\\
\int_{a}^{b}\mathcal{N}\left(x;\mu,\sigma^{2}\right)\left(x-\mu\right)^{2}\,dx & = & \left(Z-z'\right)\sigma^{2},
\end{eqnarray*}
where 
\begin{align*}
a' & \triangleq\frac{a-\mu}{\sigma}\,, & Z & \triangleq\int_{a'}^{b'}\phi=\Phi\left(b'\right)-\Phi\left(a'\right)\\
b' & \triangleq\frac{b-\mu}{\sigma}\,, & z & \triangleq\phi\left(b'\right)-\phi\left(a'\right),\\
\phi\left(x\right) & \triangleq\mathcal{N}\left(x;0,1\right), & z' & \triangleq b'\phi\left(b'\right)-a'\phi\left(a'\right).
\end{align*}
Writing
\begin{align*}
a & '_{t}\triangleq\frac{a-\mu_{t}}{\sqrt{\sigma_{t}^{2}+\alpha^{2}}} & Z_{t} & \triangleq\Phi\left(b'_{t}\right)-\Phi\left(a'_{t}\right)\\
b'_{t} & \triangleq\frac{b-\mu_{t}}{\sqrt{\sigma_{t}^{2}+\alpha^{2}}} & z_{t} & \triangleq\phi\left(b'_{t}\right)-\phi\left(a'_{t}\right),\\
 &  & z'_{t} & \triangleq b'_{t}\phi\left(b'_{t}\right)-a'_{t}\phi\left(a'_{t}\right).
\end{align*}
we find that
\begin{align*}
\int_{a}^{b}\left(\mu_{t}-\theta\right)\hat{\lambda}_{t}\left(\theta\right)d\theta & =h\sqrt{2\pi\alpha^{2}}\int_{a}^{b}\left(\mu_{t}-\theta\right)\mathcal{N}\left(\theta;\mu_{t},\sigma_{t}^{2}+\alpha^{2}\right)d\theta\\
 & =h\sqrt{2\pi\alpha^{2}}z_{t}\sqrt{\sigma_{t}^{2}+\alpha^{2}},\\
\int_{a}^{b}\left(\theta-\mu_{t}\right)^{2}\hat{\lambda}_{t}\left(\theta\right)d\theta & =h\sqrt{2\pi\alpha^{2}}\int_{a}^{b}\left(\theta-\mu_{t}\right)^{2}\mathcal{N}\left(\theta;\mu_{t},\sigma_{t}^{2}+\alpha^{2}\right)d\theta\\
 & =h\sqrt{2\pi\alpha^{2}}\left(Z_{t}-z'_{t}\right)\left(\sigma_{t}^{2}+\alpha^{2}\right)\\
\hat{\lambda}_{t}^{f}=\int_{a}^{b}\hat{\lambda}_{t}\left(\theta\right)d\theta & =h\sqrt{2\pi\alpha^{2}}Z_{t}.
\end{align*}

Substitution into (\ref{eq:mean-c-gaussian-1d})-(\ref{eq:var-c-gaussian-1d})
yields (\ref{eq:c-interval}).

\section{Implementation Details}

\subsection{State dynamics}

All simulations in this paper use linear dynamics of the form
\begin{equation}
dX_{t}=AX_{t}dt+D\,dW_{t},\label{eq:linear}
\end{equation}
which are implemented via a straightforward Euler scheme. Specifically,
for step-size $\Delta t$, we approximate $X_{k\Delta t}$ by $x_{k}$
where
\begin{align*}
x_{0} & =X_{0},\\
x_{k+1} & =x_{k}+Ax_{k}\Delta t+D\xi_{k}\sqrt{\Delta t},
\end{align*}
and $\xi_{k}$ is a sequence of independent standard normal variables
(independent of $X_{0}$).

\subsection{Continuous neural population}

The simulation of marked point processes, used to model continuous
neural populations (see Section \ref{subsec:marked}), involves the
generation of random times and corresponding random marks. In the
case of a finite population, there is a finite number of marks, and
the point process corresponding to each mark may be simulated separately.
For an infinite population, a different approach is required.

Given the intensity kernel $\lambda\left(X_{t};\boldsymbol{y}\right)f\left(d\boldsymbol{y}\right),$
we simulate a marked point process in two stages: first generating
the random times (spike times), and then the random marks (neuron
parameters). To generate the random times $\left(t_{1},\ldots t_{N_{T}}\right),$
note that the total history-conditioned firing rate at time $t$ is
given by 
\[
r\left(X_{t}\right)\triangleq\int_{Y}\lambda\left(X_{t};\boldsymbol{y}\right)f\left(d\boldsymbol{y}\right),
\]
and the unmarked process $N_{t}$ is a doubly-stochastic Poisson process
with random rate $r\left(X_{t}\right)$. Conditioned on $X$ and the
point process history, each random mark $\boldsymbol{y}_{i}$ is distributed
\[
\boldsymbol{y}_{i}|X,t_{1},\ldots t_{i},\boldsymbol{y}_{1}\ldots\boldsymbol{y}_{i-1}\sim\kappa\left(X_{t_{i}},d\boldsymbol{y}\right)\triangleq\frac{\lambda\left(X_{t_{i}};\boldsymbol{y}\right)f\left(d\boldsymbol{y}\right)}{r\left(X_{t_{i}}\right)}.
\]
(see \citet[Theorem T6]{Bremaud81}). 

\begin{table}
\centering{}\begin{tiny}%
\begin{tabular}{ccc}
\toprule 
$f\left(d\theta\right)$ & $r\left(x\right)$ & $\kappa\left(x;d\theta\right)$\tabularnewline
\midrule
\midrule 
$\delta_{\theta_{0}}\left(d\theta\right)$ & $\lambda\left(x;\theta_{0}\right)$ & $\delta_{\theta_{0}}\left(d\theta\right)$\tabularnewline
\midrule 
$d\theta$ & $h\sqrt{\frac{\left(2\pi\right)^{m}}{\det\left(R\right)}}$ & $\mathcal{N}\left(\theta;Hx,R^{-1}\right)d\theta$\tabularnewline
\midrule 
\multirow{2}{*}{$\mathcal{N}\left(\theta;c,G\right)d\theta$} & \multirow{2}{*}{$h\sqrt{\frac{\left(2\pi\right)^{m}}{\det\left(R\right)}}\mathcal{N}\left(c;Hx,R^{-1}+G\right)$} & $\mathcal{N}\left(\theta;GR_{\mathrm{G}}Hx+R^{-1}R_{\mathrm{G}}c,\left(R+G^{-1}\right)^{-1}\right)d\theta,$\tabularnewline
 &  & where $R_{\mathrm{G}}=\left(R^{-1}+G\right)^{-1}$\tabularnewline
\midrule 
\multirow{2}{*}{$\boldsymbol{1}\left\{ a\leq\theta\leq b\right\} d\theta$} & $h\sqrt{\frac{2\pi}{R}}\left[\Phi(z_{b}\left(x\right))-\Phi(z_{a}\left(x\right))\right]$, & $\mathcal{N}_{\left[a,b\right]}\left(\theta;Hx,R^{-1}\right)d\theta$\tabularnewline
 & where $z_{s}\left(x\right)=\sqrt{R}\left(s-Hx\right)$ & (truncated normal distribution)\tabularnewline
\bottomrule
\end{tabular}\end{tiny}\caption{Total population rates $r\left(x\right)$ and mark sampling distributions
$\kappa\left(x;d\theta\right)$ for the preferred stimulus distributions
$f\left(d\theta\right)$ of section \ref{subsec:cont-filtering} with
Gaussian tuning $\lambda\left(x;\theta\right)=h\exp\left(-\frac{1}{2}\left\Vert Hx-\theta\right\Vert _{R}^{2}\right)$.
The derivations of these closed forms is straightforward for the Dirac
and uniform distributions; the derivation for Gaussian distribution
is by multiplication of Gaussians, completely paralleling the computations
in section \ref{subsec:ADF-gaussian-derivation}. }
\label{total-rate-mark-sampling}
\end{table}
Accordingly, the simulation of the marked process $N$ proceeds as
follows:
\begin{enumerate}
\item Using the generated trajectory $x_{k}\approx X_{k\Delta t}$, simulate
a Poisson process with rate $r\left(X_{t}\right)$, yielding the random
times $(t_{1},\ldots t_{N_{T}})$. This may be accomplished either
via direct generation of a point for each time step with probability
$r\left(x_{k}\right)\Delta t$, or more efficiently via time rescaling
(see, e.g. \citet{brown2002time}).
\item Generate random marks $(\boldsymbol{y}_{1},\ldots\boldsymbol{y}_{N_{T}})$
by sampling independently from the distribution $\kappa\left(X_{t_{i}},d\boldsymbol{y}\right)$. 
\end{enumerate}
As in section \ref{subsec:cont-filtering}, when $h,H,R$ are fixed
across the population, we abuse notation and write $f\left(dh',d\theta,dH',dR'\right)=\delta_{h}\left(dh'\right)\delta_{H}\left(dH'\right)\delta_{R}\left(dR'\right)f\left(d\theta\right)$,
and similarly for $\kappa\left(x;d\theta\right)$. The functions $r\left(x\right)$
and the distribution $\kappa(x,d\theta)$ for each of the distributions
of preferred stimuli in section \ref{subsec:cont-filtering} are given
in closed form in Table \ref{total-rate-mark-sampling} . For a finite
heterogeneous mixture population, $r$ and $\kappa$ may be obtained
through the appropriate weighted summation; however, it is easier
to simulate each component separately.

\subsection{Filter}

Similarly to the state dynamics, we approximate the filter equations
\begin{align*}
d\mu_{t} & =d\mu_{t}^{\pi}+d\mu_{t}^{\mathrm{c}}+d\mu_{t}^{N},\\
d\Sigma_{t} & =d\Sigma_{t}^{\pi}+d\Sigma_{t}^{\mathrm{c}}+d\Sigma_{t}^{N},
\end{align*}
using a Euler approximation
\begin{align*}
\mu_{\left(k+1\right)\Delta t} & =\mu_{k\Delta t}+\left[\left(\frac{d\mu_{t}^{\pi}}{dt}+\frac{d\mu_{t}^{\mathrm{c}}}{dt}\right)\Delta t+d\mu_{t}^{N}\right]_{t=k\Delta t},\\
\Sigma_{\left(k+1\right)\Delta t} & =\Sigma_{k\Delta t}+\left[\left(\frac{d\Sigma_{t}^{\pi}}{dt}+\frac{d\Sigma_{t}^{\mathrm{c}}}{dt}\right)\Delta t+d\Sigma_{t}^{N}\right]_{t=k\Delta t}.
\end{align*}

For the linear dynamics (\ref{eq:linear}) used in simulations throughout
this paper, the prior terms $d\mu_{t}^{\pi},d\Sigma_{t}^{\pi}$ are
given by (\ref{eq:p-linear}). The continuous update terms $d\mu_{t}^{\mathrm{c}},d\Sigma_{t}^{\mathrm{c}}$
depend on the population as described in section \ref{subsec:cont-filtering}.
The discontinuous updates $d\mu_{t}^{N},d\Sigma_{t}^{N}$ are given
by (\ref{eq:mean-N-gaussian})-(\ref{eq:var-N-gaussian}), and are
non-zero only at time-steps containing a spike.

\section{Variance as proxy for MSE \label{sec:Variance-as-proxy}}

In section (\ref{sec:Encoding}), we studied optimal encoding using
the posterior variance as a proxy for the MSE. Letting $\mu_{t},\Sigma_{t}$
denote the approximate posterior moments given by the filter, the
MSE and posterior variance are related as follows,
\begin{eqnarray*}
\mathrm{MSE}_{t} & \triangleq & \E\left[\mathrm{tr}\left(X_{t}-\mu_{t}\right)\left(X_{t}-\mu_{t}\right)^{T}\right]=\E\E_{t}^{\mc N}\mathrm{tr}\left(X_{t}-\mu_{t}\right)\left(X_{t}-\mu_{t}\right)^{T}\\
 & = & \E\left[\mathrm{tr}\left(\mathrm{Var}_{t}^{\mc N}X_{t}\right)\right]+\E\left[\mathrm{tr}\left(\mu_{t}-\E_{t}^{\mc N}X_{t}\right)\left(\mu_{t}-\E_{t}^{\mc N}X_{t}\right)^{T}\right],
\end{eqnarray*}
where $\E_{t}^{\mc N}\left[\cdot\right],\mathrm{Var}_{t}^{\mc N}\left[\cdot\right]$
are resp.\ the mean and covariance matrix conditioned on $\mathcal{N}_{t}$,
and $\mathrm{tr}$ is the trace operator. Thus for an exact filter,
having $\mu_{t}=\E_{t}^{\mc N}X_{t},\Sigma_{t}=\mathrm{Var}_{t}^{\mc N}X_{t}$,
we would have $\mathrm{MSE}_{t}=\E[\mathrm{tr}\Sigma_{t}]$. Conversely,
if we find that $\mathrm{MSE}_{t}\approx\E[\mathrm{tr}\Sigma_{t}]$,
it suggests that the errors are small (though this is not guaranteed,
since the errors in $\mu_{t}$ and $\Sigma_{t}$ may effect the MSE
in opposite directions, if the variance is underestimated).

Figure \ref{variance-mse} shows the variance and MSE in estimating
the same process as in Figure \ref{harper}, after averaging across
10000 trials. The results indicate that the filter has good accuracy
with these parameters, so that the variance is a reasonable approximation
for the MSE.

\begin{figure}
\centering{}\subfloat[Narrow prior]{\begin{centering}
\includegraphics[height=3.5cm]{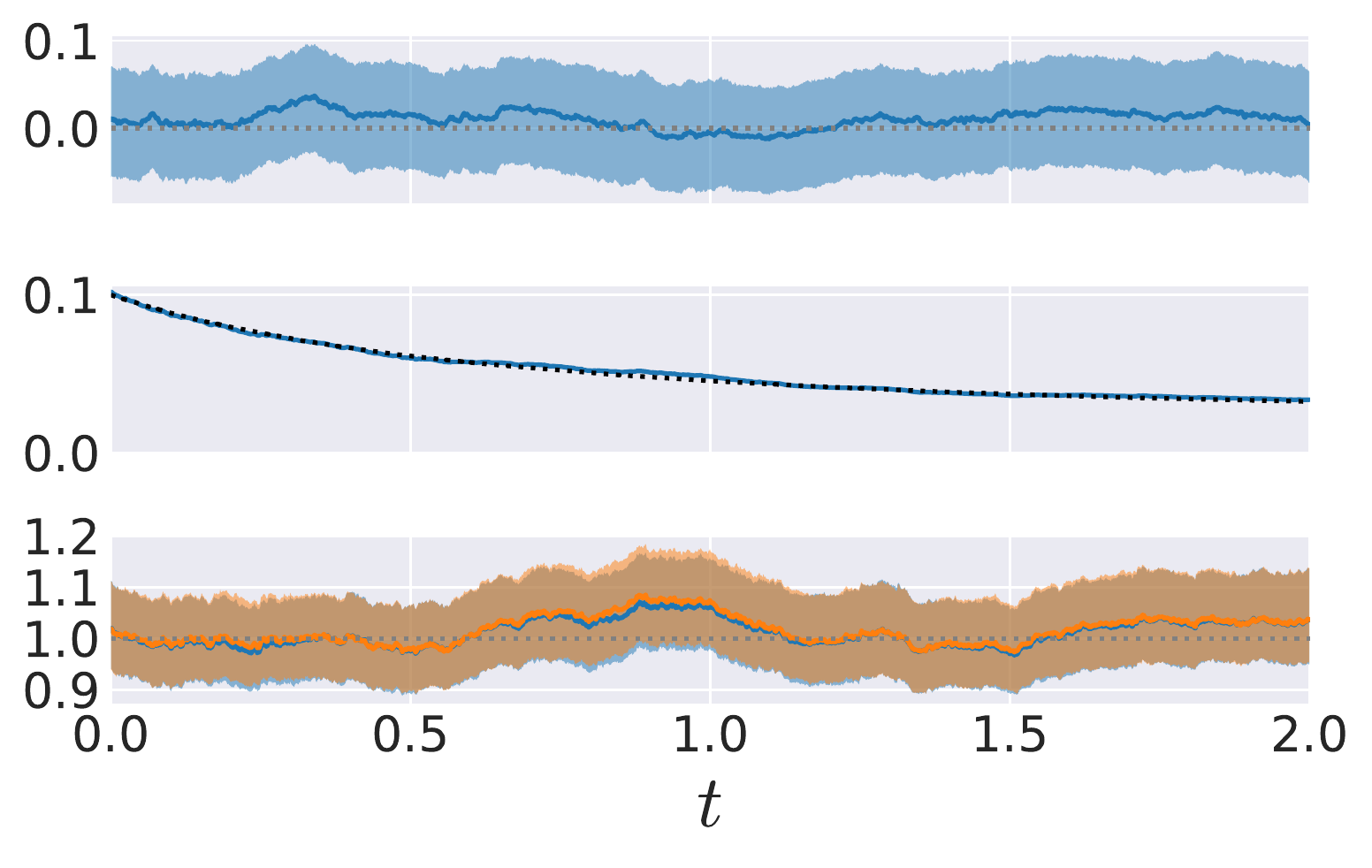}
\par\end{centering}
}\subfloat[Wide prior]{\begin{centering}
\includegraphics[height=3.5cm]{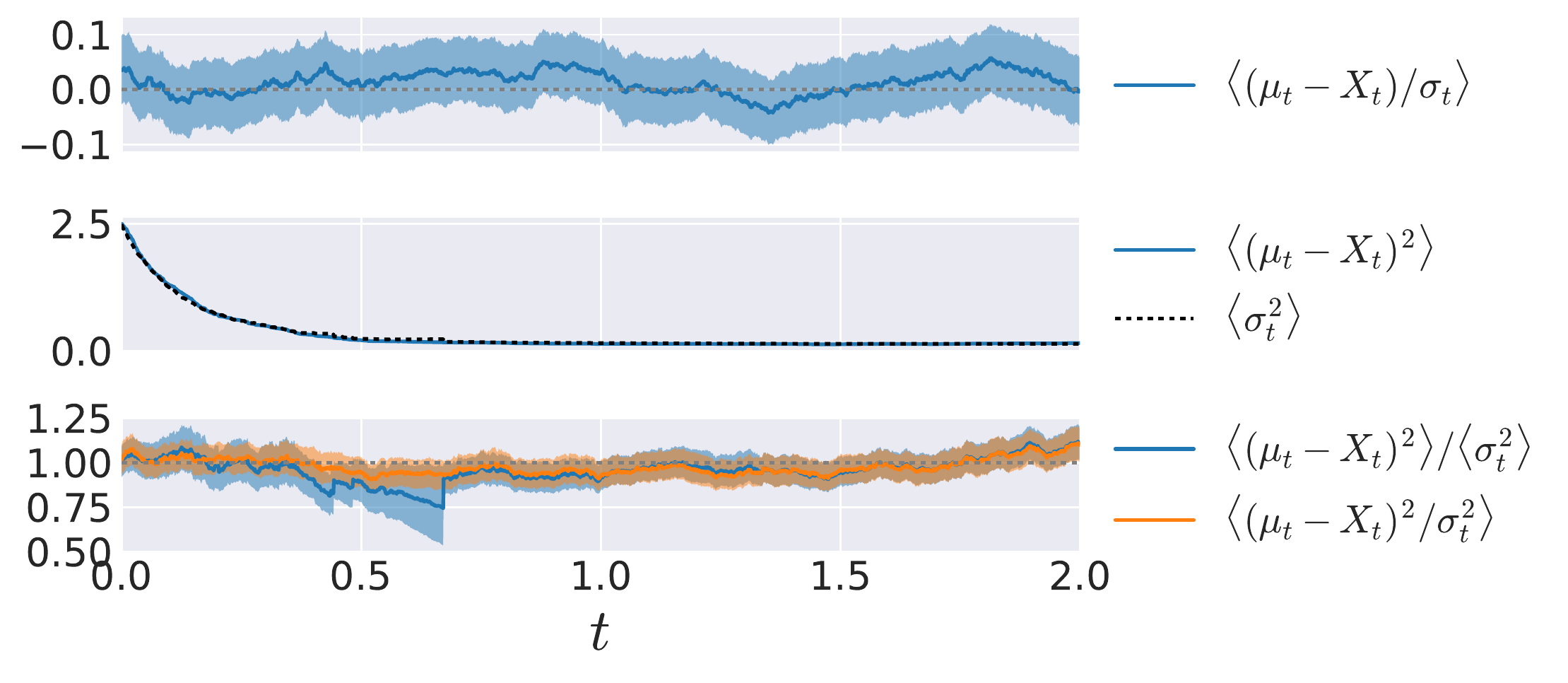}
\par\end{centering}
}\caption{Posterior variance vs. MSE when filtering the one-dimensional process
$dX_{t}=aX_{t}\,dt+s\,dW_{t}$ of Figure \ref{harper}. As in Figure
\ref{harper}, parameters on the left are $a=-0.05,s=0.1$, and on
the right $a=-0.05,s=0.5$. The sensory population is Gaussian (\ref{eq:gaussian-f}),(\ref{eq:gauss-tc-y})
with population parameters $c=0,\sigma_{\mathrm{pop}}^{2}=4$ with
and tuning parameters $\alpha=1,h=50$. The top plots show the normalized
measured bias $\left\langle \left(\mu_{t}-X_{t}\right)/\sigma_{t}\right\rangle $,
where $\left\langle \cdot\right\rangle $ denotes averaging across
trials. The center plots show the MSE $\langle\left(\mu_{t}-X_{t}\right)^{2}\rangle$
and mean posterior variance $\langle\sigma_{t}^{2}\rangle$. The bottom
plot shows the ratio of means $\langle\left(\mu_{t}-X_{t}\right)^{2}\rangle/\langle\sigma_{t}^{2}\rangle$
and the mean ratio $\langle\left(\mu_{t}-X_{t}\right)^{2}/\sigma_{t}^{2}\rangle$.
The means were taken across 10000 trials. Shaded areas indicate 95\%
confidence intervals obtained via bootstrapping.}
\label{variance-mse}
\end{figure}
Figure (\ref{harper-mse}) is a variant of Figure (\ref{harper})
showing the MSE rather than the variance. The results are noisier
but qualitatively similar. The largest differences are observed in
Figure \ref{mse-wide} for small population variance, where the ADF
estimation is poor due to very few spikes occurring.

\begin{figure}
\captionsetup{position=top}\subfloat{\includegraphics[width=0.5\columnwidth]{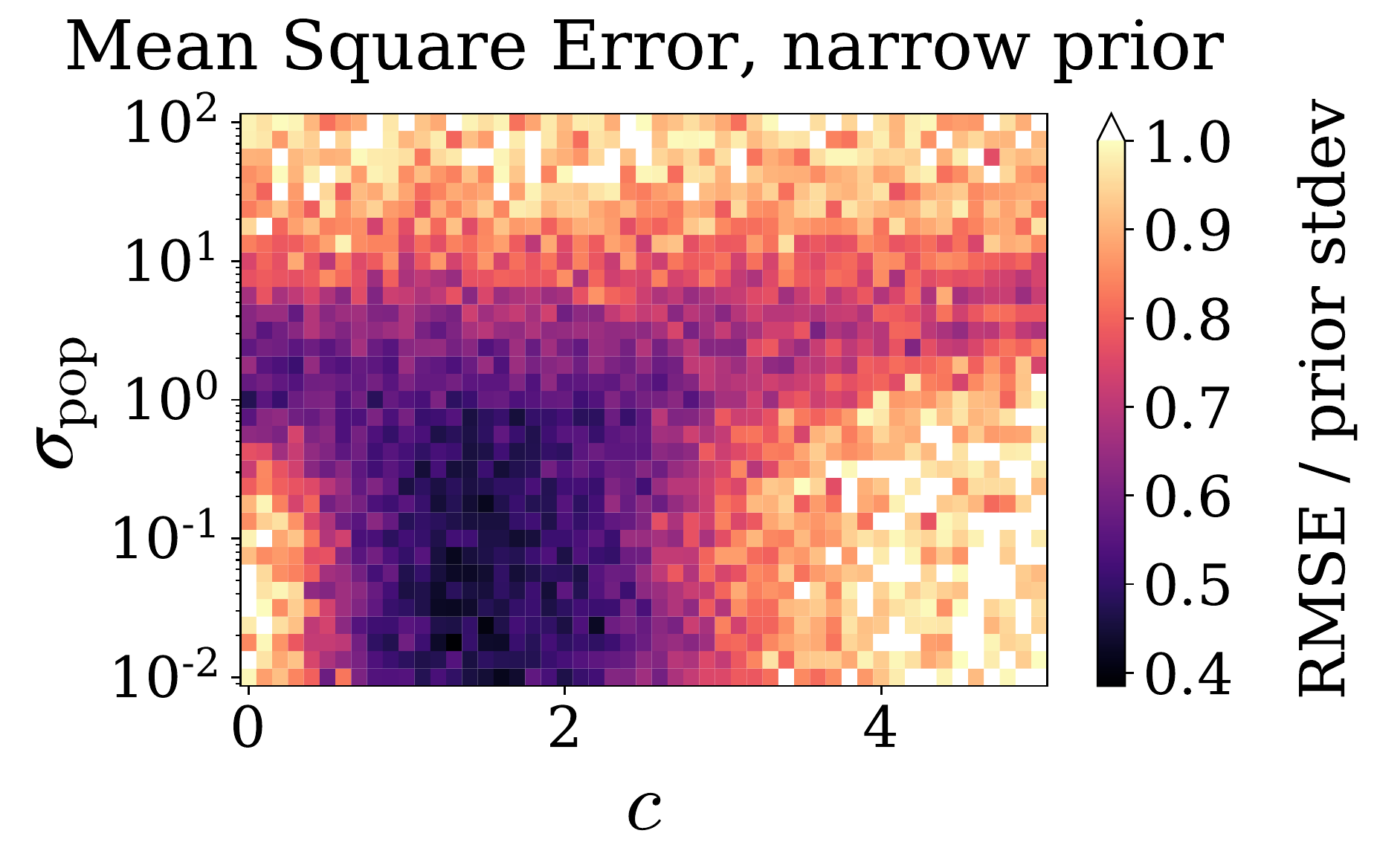}\label{mse-narrow}}\subfloat{\includegraphics[width=0.5\columnwidth]{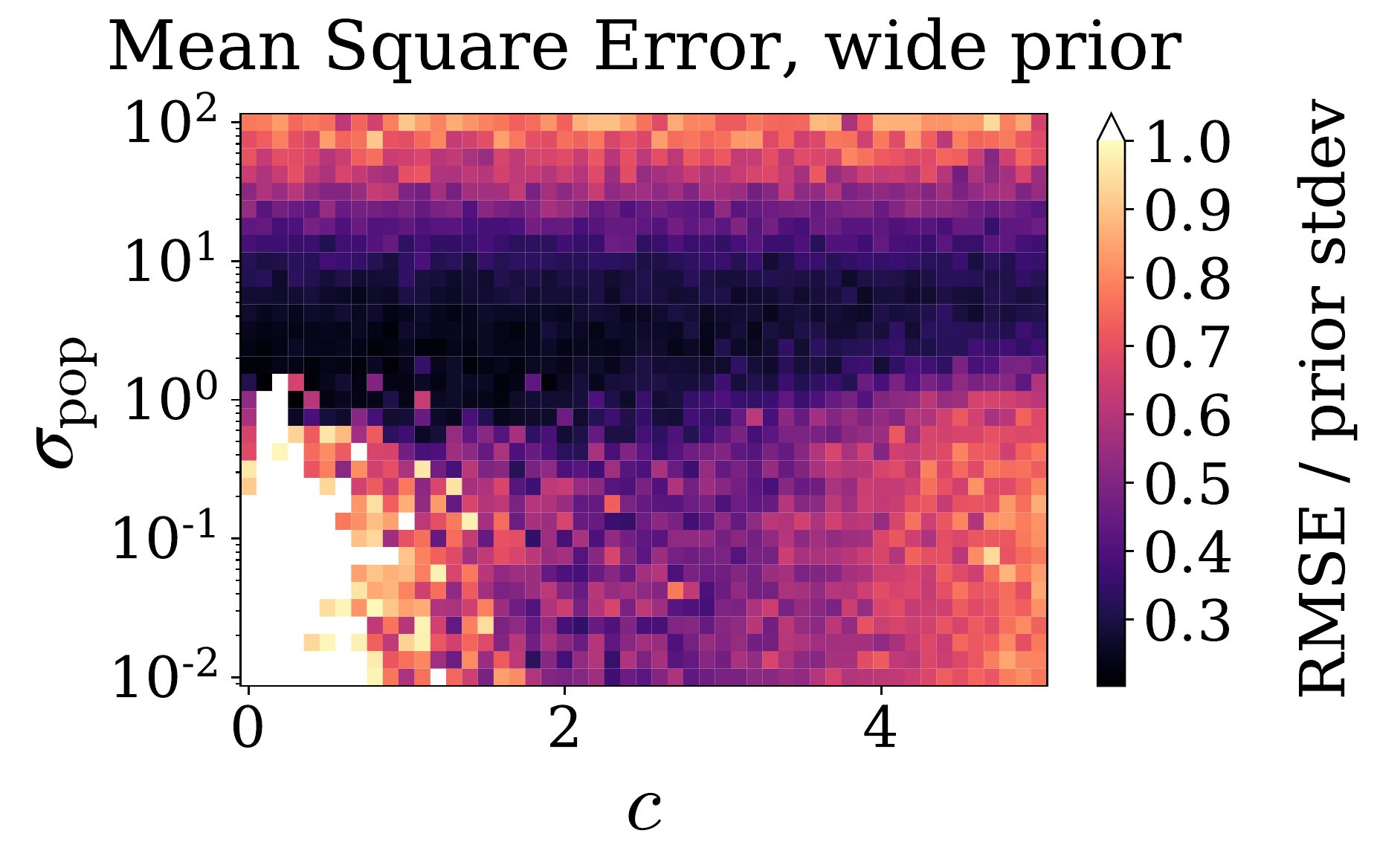}\label{mse-wide}}\caption{Optimal population distribution depends on prior variance relative
to tuning curve width. This figure is based on the same data as Figure
\ref{harper} in the main text, with MSE plotted instead of estimated
posterior variance. See Figure \ref{harper} for more details.}
\label{harper-mse}
\end{figure}

\section{Comparison to previous works -- additional details\label{sec:EB-appendix}}

As noted in section \ref{subsec:Decoding-Previous}, we are aware
of a single previous work \citet{EdenBrown2008} deriving a closed-form
filter for \emph{non-uniform }coding of a diffusion process in \emph{continuous
time}. We now present a detailed comparison of our work to \citet{EdenBrown2008}.

The filter derived in \citet{EdenBrown2008} is a continuous-time
version of a discrete-time filter presented in \citet{Eden2004}.
The setting of \citet{Eden2004} involves linear discrete-time state
dynamics, and a finite population of neurons with arbitrary tuning
functions, firing independently (given the state). The derivation
of the filter relies on an approximation applied to the measurement
update for the posterior density,

\[
\log p\left(X_{k}|\mathcal{N}_{k}\right)=\log p\left(X_{k}|\mathcal{N}_{k-1}\right)+\log P\left(\Delta\mathcal{N}_{k}|X_{k}\right)+\mathrm{const},
\]
where $X_{k}$ is the external state at time $t_{k}$, $\mathcal{N}_{k}$
is the spiking history up to time $t_{k}$, and $\Delta\mathcal{N}_{k}$
the spike counts in the interval $(t_{k-1},t_{k}]$. A Gaussian density
is substituted for each of the terms $p\left(X_{k+1},\mathcal{N}_{k+1}\right),p\left(X_{k+1},\mathcal{N}_{k}\right)$,
and each side is expanded to a second-order Taylor series about the
point $\E\left[X_{k}|\mathcal{N}_{k-1}\right]$. This yields equations
relating the first two moments of $p\left(X_{k}|\mathcal{N}_{k}\right)$
to those of $p\left(X_{k}|\mathcal{N}_{k-1}\right)$. The time update
equations for the first two moments (relating $p\left(X_{k+1}|\mathcal{N}_{k}\right)$
to $p\left(X_{k}|\mathcal{N}_{k}\right)$) require no approximation
since the dynamics are linear. The resulting discrete-time filtering
equations (equations (2.7)-(2.10) in \citet{Eden2004}) depend on
the gradient and Hessian of the logarithm of the tuning functions
at the point $\E\left[X_{k}|\mathcal{N}_{k-1}\right]$. The continuous-time
version (equations (6.3)-(6.4) in \citet{EdenBrown2008}) is derived
by taking the limit as the time discretization step $\Delta t$ approaches
0.

In contrast, the starting point in our derivation is the exact update
equations for the first two moments expressed in terms of posterior
expectations. A Gaussian posterior is substituted into these expectations,
resulting in tractable integrals in the case of Gaussian tuning functions.
The tractability of these integrals depends on the Gaussian form of
the tuning functions.

To compare the resulting filters, we consider a finite population
of Gaussian neurons: this is the intersection of the setting of the
current work with that of \citet{EdenBrown2008}. In this case, the
filtering equations in \citet{EdenBrown2008} yield the same discontinuous
update terms, but the continuous update terms take the form\begin{subequations}\label{eq:eb}
\begin{align}
d\mu_{t}^{\mathrm{c,EB}} & =\sum_{i}\lambda^{i}\left(\mu_{t}\right)\Sigma_{t}H_{i}\transpose R_{i}\delta_{t}^{i}dt\label{eq:eb-mean}\\
d\Sigma_{t}^{\mathrm{c,EB}} & =\sum_{i}\lambda^{i}\left(\mu_{t}\right)\Sigma_{t}H_{i}\transpose\left(R_{i}-R_{i}\delta_{t}^{i}\left(\delta_{t}^{i}\right)\transpose R_{i}\right)H_{i}\Sigma_{t}dt,\label{eq:eb-var}
\end{align}
where
\[
\delta_{t}^{i}\triangleq H_{i}\mu_{t}-\theta_{i}.
\]
\end{subequations}(see section \ref{subsec:EB-derivation}). These
equations differ from (\ref{eq:mean-c-gaussian-finite})-(\ref{var-c-gaussian-finite})
in the use of the tuning function shape matrix $R_{i}$ in place of
$S_{t}^{i}$ (defined in (\ref{eq:S-finite})), and $\lambda^{i}\left(\mu_{t}\right)$
in place of $\hat{\lambda}_{t}^{i}$. The difference between $\lambda^{i}\left(\mu_{t}\right)$
and $\hat{\lambda}_{t}^{i}$ similarly involves substituting $R_{i}$
for $S_{t}^{i}$,
\begin{align*}
\lambda^{i}\left(\mu_{t}\right) & =h\exp\left(-\frac{1}{2}\left(\delta_{t}^{i}\right)\transpose R_{i}\delta_{t}^{i}\right)\\
\hat{\lambda}_{t}^{i} & =h\sqrt{\frac{\left|S_{t}^{i}\right|}{\left|R_{i}\right|}}\exp\left(-\frac{1}{2}\left(\delta_{t}^{i}\right)\transpose S_{t}^{i}\delta_{t}^{i}\right).
\end{align*}
Since $S_{t}^{i}=\left(R_{i}^{-1}+H_{i}\transpose\Sigma_{t}H_{i}\right)^{-1}$,
our filtering equations take into account the posterior variance $\Sigma_{t}$
in several places where it is absent in (\ref{eq:eb}). Note that
when $\Sigma_{t}=0$ we have $R_{i}=S_{t}^{i}$ and $\lambda\left(\mu_{t}\right)=\hat{\lambda}_{t}$,
so the equations become increasingly similar when $\Sigma_{t}\to0$.
We refer to the filter (\ref{eq:eb}) as the Eden-Brown (EB) filter.

We compared the performance of our filter and (\ref{eq:eb}) in simulations
of a simple one-dimensional setup. Figure \ref{fig-eb-example-1}
shows an example of filtering a static one-dimensional state observed
through two Gaussian neurons (\ref{eq:gauss-tc}), using both the
ADF approximation (\ref{eq:mean-c-gaussian-finite})-(\ref{var-c-gaussian-finite})
and the EB filter (\ref{eq:eb}) for comparison. Since the Mean Square
Error is highly sensitive to outliers where the approximation fails
and the filtering error becomes large, we compare the filters by observing
the distribution of Absolute Errors (AE) in estimating the state.
Figure \ref{eb-comparison-1} compares the AE in estimating the state
for the two filters. As expected from the analysis above, the ADF
filter has an advantage particularly in earlier times, when the posterior
variance is large. This is seen most clearly in the 95th percentile
of AEs in Figure \ref{eb-comparison-1}(a), and in the tail histograms
(\ref{eb-comparison-1}(c)), but a small advantage may also be observed
for the median in \ref{eb-comparison-1}(d). However, in some trials
the error in the EB filter remains large throughout.

In Figure \ref{eb-comparison-1} (left) we chose the preferred stimuli
$-0.51$ and $0.5$. These are not symmetric around 0 due to a limitation
of the EB filter: when applied to a homogeneous population and the
current posterior mean is precisely in the average of the population's
preferred stimuli, the posterior mean remains constant until the next
spike, while the posterior variance evolves as 
\[
d\sigma_{t}^{2,\mathrm{c,EB}}=\sum_{i}\lambda^{i}\left(\mu_{t}\right)\frac{1}{\alpha^{2}}\left(1-\frac{\left(\mu_{t}-\theta_{i}\right)^{2}}{\alpha^{2}}\right)\sigma_{t}^{4}dt=\mathrm{const}\cdot\sigma_{t}^{4}dt,
\]
which diverges in finite time if the constant coefficient is positive.
The coefficient is positive when $\theta_{i}$ are close to $\mu_{t}$.
This causes divergence of the filter when the preferred stimuli are
symmetric around the initial estimate $\mu_{0}$ and sufficiently
near it, and the first spike is sufficiently delayed. To avoid this
behavior we chose preferred stimuli that are not symmetric around
$\mu_{0}=0$. This asymmetry causes an eventual shift in the posterior
mean in the absence of spikes, which suppresses the growth of the
posterior variance. This effect may cause high estimation errors before
the first spike, as evident in Figure \ref{eb-comparison-1}d (left).

\begin{figure}
\includegraphics[width=1\columnwidth]{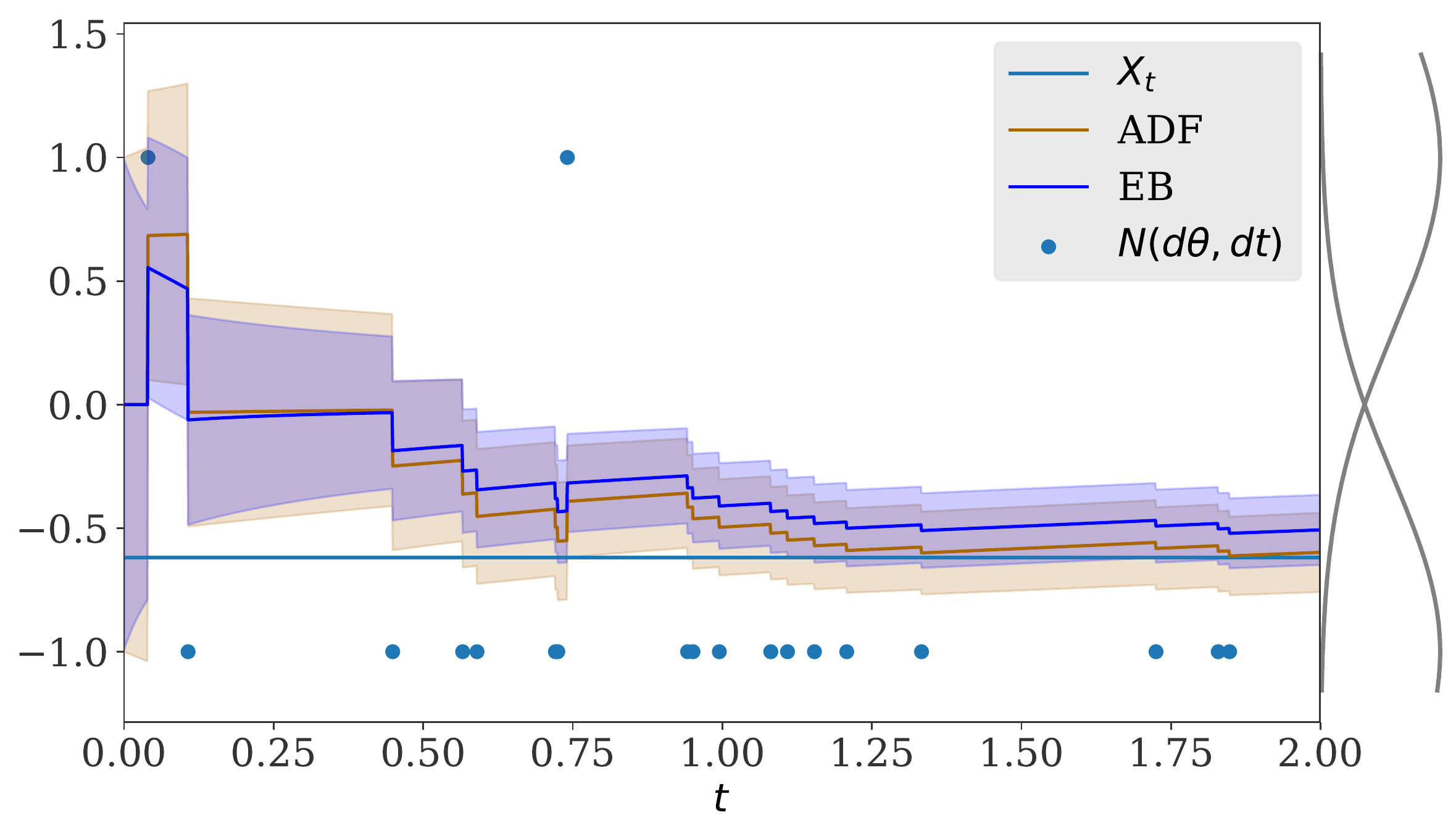}\caption{An example of a static 1-d state observed through two sensory neurons
and filtered by ADF (\ref{eq:mean-c-gaussian-finite})-(\ref{var-c-gaussian-finite})
and by the EB filter (\ref{eq:eb-mean})-(\ref{eq:eb-var}). Each
dot corresponds to a spike with the vertical location indicating the
neuron's preferred stimulus $\theta$. The approximate posterior means
obtained from ADF and the EB filter are shown in orange and blue respectively,
with the corresponding posterior standard deviations in shaded areas
of the respective colors. The curves to the right of the graph show
the tuning functions of the two neurons. Parameters are: $h=10,H=1,R^{-1}=0.5,\mu_{0}=0,\sigma_{0}^{2}=1$,
and the neurons are centered at $-1$ and $1$.The dynamics were discretized
with time step $\Delta t=10^{-3}$. }
\label{fig-eb-example-1}
\end{figure}
\begin{sidewaysfigure}
\includegraphics[width=0.45\columnwidth]{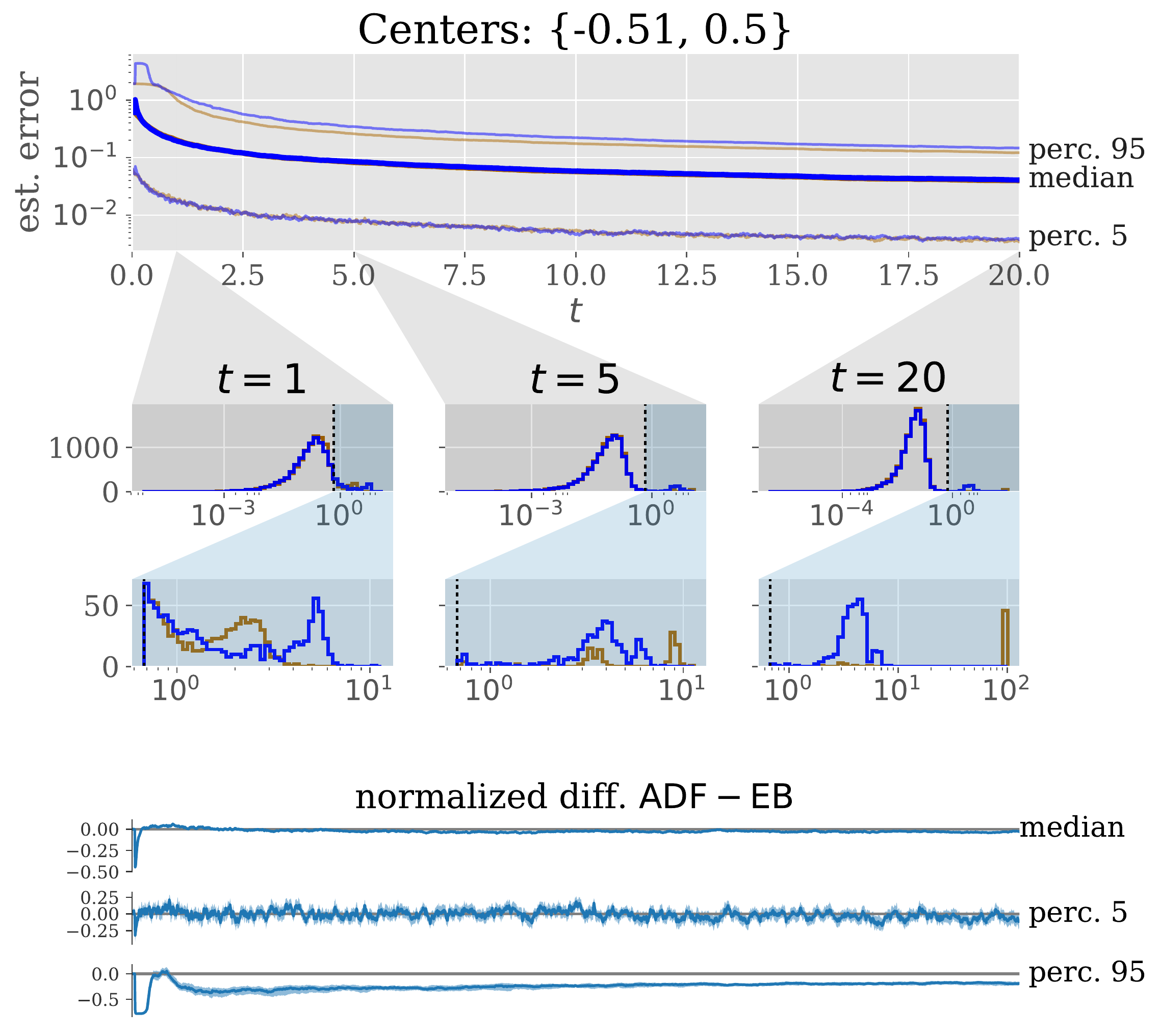}\hfill{}\includegraphics[width=0.45\columnwidth]{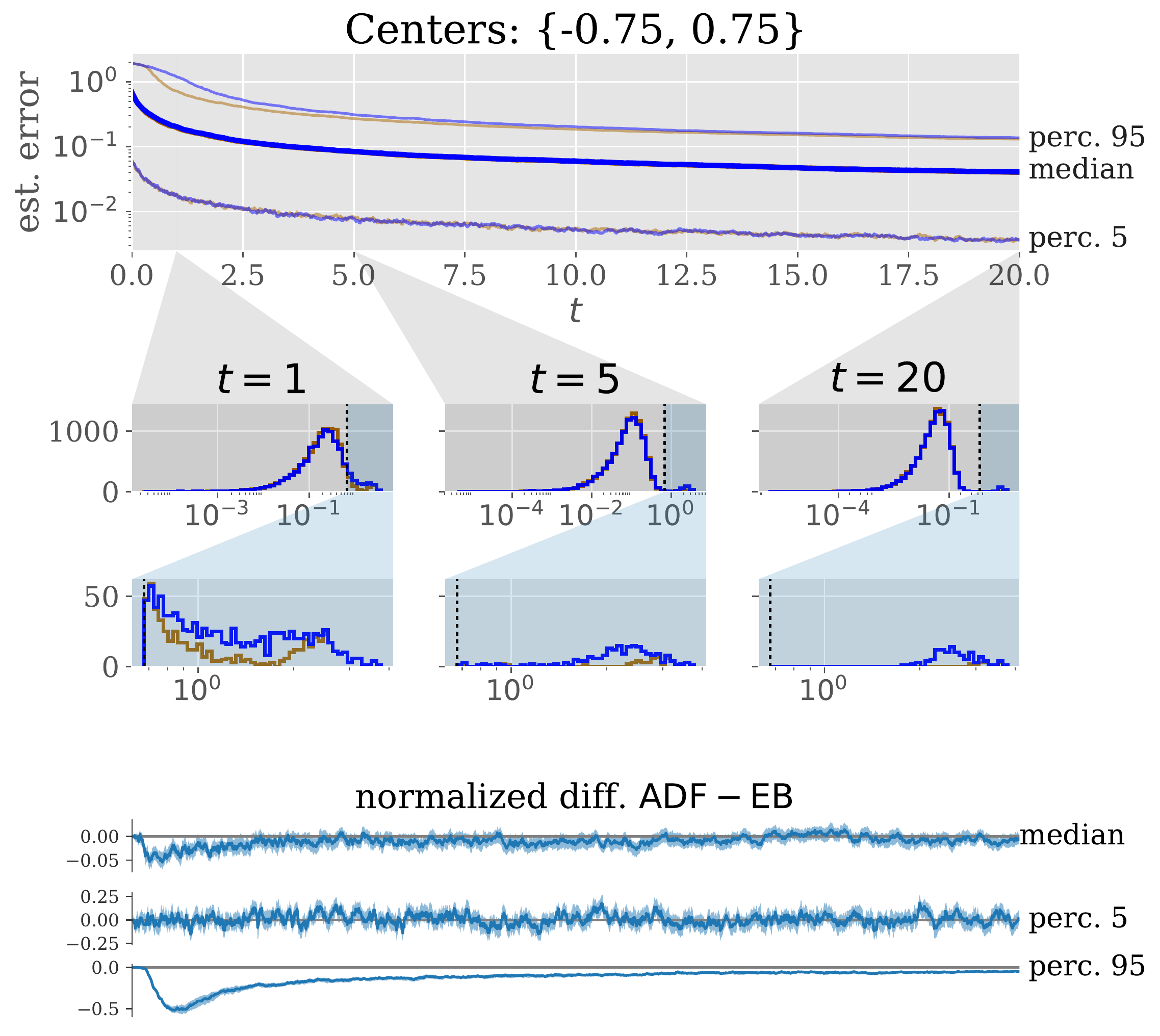}\caption{Distribution of absolute estimation errors as a function of time,
collected from 10,000 trials. Parameters are the same as in Figure
(\ref{fig-eb-example-1}), with the state sampled from $\mathcal{N}\left(0,1\right)$
in each trial. Preferred stimuli are noted above each subfigure.\emph{
}\textbf{(a)} Medians of the distribution of AEs for ADF (orange)
and the EB filter (blue), along with 5th and 95th percentiles, as
a function of time. The medians are indistinguishable at this scale.
\textbf{(b)} Histograms of AEs for some specific times, with logarithmically
spaced bins. The vertical dotted line indicates the AE at filter initialization,
which equals the prior expectation $\protect\E\left|X_{0}-\mu_{0}\right|=\sigma_{0}\Phi^{-1}\left(\frac{3}{4}\right)$
where $\Phi^{-1}$ is the quantile function for the standard normal
distribution. \textbf{(c)} Histogram of AEs larger than the initial
AE. \textbf{(d) }Normalized differences of AE percentiles $\left(p_{\mathrm{ADF}}^{r}-p_{\mathrm{EB}}^{r}\right)/\frac{1}{2}\left(p_{\mathrm{ADF}}^{r}+p_{\mathrm{EB}}^{r}\right)$,
where $p_{\mathrm{ADF}}^{r},p_{\mathrm{EB}}^{r}$ are the $r$th quantile
of the AE distribution for the ADF and EB filters respectively, for
$r=0.5,0.05,0.95$. Negative values indicate an advantage of ADF over
EB. Shaded areas indicate 95\% confidence intervals derived via bootstrapping
(e.g. \citet{Efron1994introduction}).}
\label{eb-comparison-1}
\end{sidewaysfigure}

\subsection{Derivation of (\ref{eq:eb})\label{subsec:EB-derivation}}

In our notation, the filtering equation of \citet{EdenBrown2008}
read\begin{subequations}\label{eq:eb-general}
\begin{align}
d\mu_{t}^{\mathrm{c,EB}} & =-\sum_{i}\Sigma_{t}\nabla\log\lambda^{i}\left(\mu_{t}\right)\lambda^{i}\left(\mu_{t}\right)dt\label{eq:eb-mean-c-general}\\
d\Sigma_{t}^{\mathrm{c,EB}} & =-\sum_{i}\Sigma_{t}\nabla^{2}\lambda^{i}\left(\mu_{t}\right)\Sigma_{t}dt\\
d\mu_{t}^{N,EB} & =\sum_{i}\Sigma_{t^{+}}\nabla\log\lambda^{i}\left(\mu_{t^{-}}\right)dN_{t}^{i}\\
d\Sigma_{t}^{N,EB} & =-\sum_{i}\Sigma_{t^{-}}S{}_{t^{-}}^{\mathrm{EB}}\Sigma_{t^{-}}dN_{t}^{i}\label{eq:eb-var-N-general}
\end{align}
\end{subequations}where $\nabla,\nabla^{2}$ denote the gradient
and Hessian respectively, and $S_{t^{-}}^{\mathrm{EB}}$ is defined
as
\[
S_{t}^{\mathrm{EB}}\triangleq\nabla^{2}\log\lambda^{i}\left(\mu_{t}\right)\left(\Sigma_{t}\nabla^{2}\log\lambda^{i}\left(\mu_{t}\right)-I\right)^{-1}
\]
\textbf{Note }this is a corrected version of the definition used in
\citet{EdenBrown2008}, which is unusable when the Hessian is singular
(this occurs in our model whenever $m<n$). The definition of $S_{t}^{\mathrm{EB}}$
given here extends it to the singular case.

For a Gaussian firing rate, (\ref{eq:gauss-tc}), the relevant gradient
and Hessians are given by
\begin{align*}
\nabla\log\lambda^{i}\left(x\right) & =-H_{i}\transpose R_{i}\left(H_{i}x-\theta_{i}\right)\\
\nabla^{2}\log\lambda^{i}\left(x\right) & =-H_{i}\transpose R_{i}H_{i}\\
\nabla^{2}\lambda^{i}\left(x\right) & =-\lambda^{i}\left(x\right)H_{i}\transpose\left(R_{i}-R_{i}\left(H_{i}x-\theta_{i}\right)\left(H_{i}x-\theta_{i}\right)\transpose R_{i}\right)H_{i},
\end{align*}
so
\begin{align*}
S_{t}^{\mathrm{EB}} & =\nabla^{2}\log\lambda^{i}\left(\mu_{t}\right)\left(\Sigma_{t}\nabla^{2}\log\lambda^{i}\left(\mu_{t}\right)-I\right)^{-1}\\
 & =-H_{i}\transpose R_{i}H_{i}\left(-\Sigma_{t}H_{i}\transpose R_{i}H_{i}-I\right)^{-1}\\
 & =H_{i}\transpose R_{i}H_{i}\left(\Sigma_{t}^{-1}+H_{i}\transpose R_{i}H_{i}\right)^{-1}\Sigma_{t}^{-1}\\
 & =H_{i}\transpose\left(R_{i}^{-1}+H_{i}\Sigma_{t}H_{i}\transpose\right)^{-1}H_{i}\tag{using Claim \ref{claim:inversion}}\\
 & =H_{i}\transpose S_{t}^{i}H_{i}
\end{align*}
where $S_{t}^{i}$ is given by (\ref{eq:S-finite}). Substituting
into (\ref{eq:eb-general}) yields the continuous updates (\ref{eq:eb-mean})-(\ref{eq:eb-var}),
and the discontinuous updates
\begin{align*}
d\mu_{t}^{N,\mathrm{EB}} & =\sum_{i}\Sigma_{t^{+}}H_{i}\transpose R_{i}\delta_{t^{-}}^{i}dN_{t}^{i},\\
d\Sigma_{t}^{N,\mathrm{EB}} & =-\sum_{i}\Sigma_{t^{-}}H_{i}\transpose S{}_{t^{-}}^{i}H_{i}\Sigma_{t^{-}}dN_{t}^{i}.
\end{align*}
The discontinuous variance update is identical to the ADF update (\ref{eq:mean-N-gaussian-finite}).
The discontinuous mean update can be rewritten in terms of $\Sigma_{t^{-}}$
by noting that
\begin{align*}
\Sigma_{t^{+}}H_{i}\transpose R_{i} & =\left(\Sigma_{t^{-}}-\Sigma_{t^{-}}H_{i}\transpose S{}_{t^{-}}^{i}H_{i}\Sigma_{t^{-}}\right)H_{i}\transpose R_{i}\\
 & =\Sigma_{t^{-}}H_{i}\transpose\left(I-S_{t^{-}}^{i}H_{i}\Sigma_{t^{-}}H_{i}\transpose\right)R_{i}\\
 & =\Sigma_{t^{-}}H_{i}\transpose\left(I-\left(R_{i}^{-1}+H_{t}\Sigma_{t}H_{i}\transpose\right)^{-1}H_{i}\Sigma_{t^{-}}H_{i}\transpose\right)R_{i}\\
 & =\Sigma_{t^{-}}H_{i}\transpose\left(I+R_{i}H_{i}\Sigma_{t^{-}}H_{i}\transpose\right)^{-1}R_{i}\tag{Woodbury}\\
 & =\Sigma_{t^{-}}H_{i}\transpose S_{t}^{i}
\end{align*}
so the mean update may be rewritten as
\begin{align*}
d\mu_{t}^{N,\mathrm{EB}} & =\sum_{i}\Sigma_{t^{-}}H_{i}\transpose S_{t}^{i}\delta_{t^{-}}^{i}dN_{t}^{i},
\end{align*}
in agreement with (\ref{eq:mean-N-gaussian-finite}).
\end{document}